%% file: main.tex
\documentclass[12pt]{article}

\usepackage[a4paper, margin=1in]{geometry}
\usepackage{amsmath, amssymb, amsfonts,amsthm}
\usepackage{booktabs}
\usepackage{graphicx}
\usepackage{float}
\usepackage{caption}
\usepackage{setspace}
\usepackage{lmodern}
\usepackage{hyperref}
\usepackage{authblk}
\usepackage{multirow}
\usepackage{algorithm}
\usepackage{algpseudocode} 
\usepackage[round]{natbib}
\usepackage{array}
\usepackage{tikz}
\usepackage{subcaption}
\usetikzlibrary{arrows.meta, positioning}
\input{notation.tex}

\allowdisplaybreaks[4] 
\usepackage[final]{microtype}
\title{Identification and Inference for Structural Accelerated Failure Time Models via Instrument Interactions}
\author[1,2]{Qiushi Bu}
\author[1]{Wen Su}
\author[2]{Xinyu Zhang}
\author[3]{Xingqiu Zhao}
\author[4]{Zhonghua Liu$^*$}

\affil[1]{\small Department of Biostatistics, City University of Hong Kong, Hong Kong}
\affil[2]{\small Academy of Mathematics and Systems Science, Chinese Academy of Sciences, Beijing, China}
\affil[3]{\small Department of Applied Mathematics, The Hong Kong Polytechnic University, Hong Kong}
\affil[4]{\small Department of Biostatistics, Columbia University, New York, NY, USA}

\begin{document}
\maketitle
\begin{abstract}
    We study causal inference for time-to-event outcomes under right censoring in the presence of unmeasured confounding. Focusing on structural accelerated failure time  models, we develop an identification and inference framework that exploits interactions among instrumental variables. The proposed approach does not rely on classical instrumental variable validity and yields valid causal inference under both valid and invalid instruments, provided that the interaction-based identification condition holds. To accommodate right censoring, we construct a censoring-adjusted observed data moment function using an augmented inverse probability censoring weighting approach. The resulting moment function is Neyman orthogonal with respect to  nuisance functions and enjoys a double robustness property, enabling valid inference under flexible nuisance estimation. Estimation and inference are conducted using generalized empirical likelihood, which is well suited to settings with many potentially weak interaction-based moment conditions. We establish consistency, and asymptotic normality under many weak moment asymptotics, and develop diagnostic tools to assess interaction-based identification strength and overidentifying restrictions.
    Simulation studies demonstrate favorable finite sample performance across a range of censoring rates and instrument configurations. An application to UK Biobank data illustrates the practical relevance of the proposed method for causal survival analysis in large-scale observational studies.
    

\end{abstract}

\section{Introduction}

Estimating causal effects on the time to disease onset is important in epidemiology. Many modern cohorts record outcomes as time-to-event data with long follow-up. The UK Biobank is a representative resource that links genome-wide genetic data with longitudinal outcomes from health records and registries \citep{Sudlow2015}. This type of data supports large-scale survival analyses, and recent work has analyzed hundreds of time-to-event outcomes constructed from linked electronic health records (EHR) in the UK Biobank \citep{dey2022}. These data sets are increasing, and they motivate methods for causal inference with censored survival outcomes.

A main difficulty in observational studies is unmeasured confounding, for example, when unobserved lifestyle factor affects both the exposure and the time to disease onset. Ignoring such unmeasured confounding generally leads to biased estimation and invalid statistical inference. Instrumental variables (IV) address this problem by using variables that affect the exposure but do not otherwise affect the outcome. Mendelian randomization applies this idea using genetic variants as instruments \citep{Katan1986,SmithEbrahim2003}. In biobank studies, many variants are available for a given exposure, which makes it natural to use many candidate instruments in one analysis.

Instrumental variable validity is often summarized by three requirements \cite{angrist2001}: (i) Relevance, which means that the instrument is associated with the exposure. (ii) Independence, which means that the instrument is independent of unmeasured confounders. (iii) Exclusion restriction, which means that the instrument affects the outcome only through the exposure. However, in genetic studies, these requirements may not be fully satisfied. Many variants have only weak associations with the exposure, which can yield weak identification and substantial weak IV bias \citep{burgess2011bias,stock2002survey}. Moreover, horizontal pleiotropy, defined as a genetic variant influencing the outcome through pathways other than the exposure of interest, will also lead to spurious findings \citep{Verbanck2018,solovieff2013pleiotropy,ye2024}.

Motivated by the prevalence of weak and invalid instruments in Mendelian randomization, a substantial body of work has been devoted to developing identification and inference methods that address these challenges, primarily in settings with uncensored outcomes. One line of work establishes identification by imposing restrictions on the set of instruments, assuming that a majority of candidate instruments are valid \citep{Han2008,Kang2016,Bowden2016}. A second line studies random coefficient formulations, showing that valid causal inference may still be possible even when all candidate instruments are invalid, provided that the effects of the instruments on the exposure are Neyman orthogonal to their direct effects on the outcome \citep{Kolesar2015,Bowden2015,Zhao2020,Ye2021}. A third line exploits heteroskedasticity in the exposure models to identify the causal effect in the presence of invalid instruments \citep{Lewbel2012,spiller2019,Tchetgen2021,sun2022,liu2023,ye2024}. More recently, additional identification methods have been obtained by constructing moment conditions from gene-by-gene interactions, which can deliver identification even when main effect instruments violate exclusion restrictions \citep{zhang2025}. These results show that identification can be possible with many weak and potentially invalid instruments, but most theory is developed for fully observed outcomes.

However, work for right-censored survival outcomes is much more limited. Existing instrumental variable methods for censored time-to-event data usually focus on handling censoring under valid instruments \citep{Nie2011,huling2019}. There is also work that considers censoring together with weak instruments, for example, in applications where the outcome is a censored duration \citep{Ertefaie2018}. These papers do not cover the setting of many weak and potentially invalid genetic instruments that are common in Mendelian randomization. \cite{Tchetgen2021} develop estimation methods for censored outcomes in the presence of potentially weak and invalid instruments. However, their identification strategies rely on heteroskedasticity-based conditions.

This paper develops a causal inference framework for right-censored outcomes with many weak and potentially invalid instruments through gene-gene interaction identification. We propose an observed data estimation strategy that combines interaction-based identification with a censoring-adjusted moment construction. Identification is achieved through interaction terms that remain informative under weak instrument strength and allow violations of standard instrumental variable assumptions. Estimation and inference are conducted using generalized empirical likelihood (GEL), which combines many moment functions and is suitable for weak identification settings \citep{smith1997,newey2004,Owen2001,newey2009,Hansen1982}.

In summary, our contributions are mainly threefold:
\begin{enumerate}
\item We develop \textbf{iGSAFT}, (Interaction-IV based Generalized Structural Accelerated Failure Time), a causal inference framework for right-censored survival outcomes with many weak and potentially invalid instrumental variables. The proposed model is identified through interaction-based conditions, under which the causal effect is identifiable regardless of whether the instruments are valid, provided that the interaction-based assumptions hold.

\item We develop an estimation and inference procedure based on a censoring-adjusted augmented inverse probability weighted moment function under the GSAFT model. The moment function is Neyman orthogonal and doubly robust, yielding consistent estimation if either the censoring model or the outcome-related augmentation component is correctly specified. These properties support valid inference with flexible nuisance estimation and allow overidentification through many moment conditions under right censoring.

\item We develop two model diagnosis tools, focusing on weak identification and overidentification, respectively. The weak identification test assesses whether the gene-by-gene interaction-based identification conditions are sufficiently informative in the observed data, while the overidentification test evaluates the compatibility of the imposed moment functions with the observed data. Together, these procedures offer data-driven checks of the key identification assumptions and improve the reliability and scalability of the proposed methodology in large observational and biobank studies.
\end{enumerate}

The remainder of the paper is organized as follows. Section \ref{sec2} introduces the problem setup and the structural accelerated failure time model. Section \ref{sec3} presents the identification strategy and the censoring-adjusted moment function. Section \ref{sec4} describes the estimation procedure. Section \ref{sec5} establishes the asymptotic properties of the proposed estimator. Section \ref{sec6} develops model diagnostic tools, including a weak identification test and an  overidentification test. Section \ref{sec7} reports simulation results. Section \ref{sec8} applies the proposed method to UK Biobank data to study the causal effect of body mass index on disease onset. Section \ref{sec9} concludes with a discussion. Technical proofs are provided in the Supplementary Material.

\section{Structural AFT model with possibly invalid instruments} \label{sec2}
The classical structural accelerated failure time (SAFT) model quantifies treatment effects on survival outcomes, which characterizes the causal effect through an acceleration factor acting on the counterfactual failure time \citep{robins1991correcting,robins1992estimation,vandebosch2005structural,hernan2005structural}. Let $T^*(d)$ denote the counterfactual failure time that would be observed if the exposure were set to level $d\in\mathbb{R}$. The baseline counterfactual failure time $T^*(0)$ is not always observed and represents the survival outcome that would have been realized in the absence of exposure. The causal effect of the exposure is defined through the time ratio $T^*(d)/T^*(0)$.

Under this model setup, the classical SAFT model \citep{robins1991correcting,robins1993analytic,hernan2005structural} can be written as 
$$
T^*(0)=T^*(d)\exp(-\beta_0 d),
$$
where $\beta_0$ belongs to a compact set $\mathcal{B}$. In randomized experiments, the exposure $d$ is independent of the baseline counterfactual failure time $T^*(0)$, so that $d$ is exogenous and $\beta_0$ is identifiable. In observational settings, however, $d$ may be correlated with $T^*(0)$, rendering $d$ endogenous and the causal effect no longer identifiable without additional assumptions. In this case, instrumental variables can be used to achieve identification in SAFT model \citep{robins1991correcting,robins1993analytic}. However, if the instrumental variables are invalid, for example, due to violations of the exclusion restriction or the independence condition, identification again fails. While identification under invalid instrumental variables have been studied in other causal frameworks \citep{Kang2016,zhang2025}, these approaches do not directly generalize to SAFT models, where methodological developments remain limited.

To accommodate possibly invalid instrumental variables, we generalize the SAFT model with valid instrumental variables in \citet{robins1993analytic}, and refer to the resulting model as the generalized structural accelerated failure time (GSAFT) model. Let $T^*(d,\z)$ denote the counterfactual failure time that would be observed if the exposure were set to level $d\in\mathbb{R}$ and the candidate instrumental variables were set to $\z=(z_1,\ldots,z_p)^\top\in\mathbb{R}^p$. The baseline counterfactual failure time $T^*(0,\mathbf{0})$ represents the survival outcome that would have been realized when both exposure and candidate instruments are set to zero. Then the GSAFT model is given by 
\be\label{eq:AFTT}
T^*(0,\mathbf{0})=T^*(d,\z)\exp\!\left\{-\beta_0 d-\sum_{k=1}^p\gamma_{k}z_k\right\},
\ee
\be
\bbE\{\log T^*(0,\mathbf{0})|\z\} = \sum_{k=1}^p\alpha_{k}z_k.
\ee
Here, the coefficients $\gamma_{k}$ capture direct effects of the candidate instruments on the failure time; instruments with $\gamma_{k}\neq 0$ violate the exclusion restriction and are therefore invalid. The coefficients $\alpha_{k}$ characterize the dependence between the candidate instruments and the baseline counterfactual failure time $T^*(0,\mathbf{0})$. When $\alpha_{k}\neq 0$, the instruments are associated with unobserved confounders, implying that the IV independence assumption is violated. These two sources of invalidity are jointly summarized by $\phi_{k}=\gamma_{k}+\alpha_{k}$. In what follows, $\phi_{k}$ serves as an index of instrument invalidity. A nonzero value of $\phi_{k}$ indicates that the candidate instrument violates at least one of the core IV assumptions considered here, while $\phi_{k}=0$ corresponds to the absence of these two violations, and the model will reduce to classical SAFT model. A special case occurs when $\gamma_k = -\alpha_k \neq 0$. In this situation, $\phi_k = 0$ by construction, but the instrument is still invalid, as both the exclusion restriction and the independence condition are violated despite the algebraic cancellation. Figure \ref{fig:iv} illustrates the relationship among the instrumental variables, exposure, outcome, and confounders.
\begin{figure}[h]
\centering

\begin{tikzpicture}[
 node distance=2cm and 3cm,
 every node/.style={font=\normalsize},
 arrow/.style={-{Latex[width=2mm,length=2mm]}, line width=0.8pt},
 redarrow/.style={arrow, red}
]

\node (Z) at (0,0) {$\Z$};
\node (D) at (3,0) {$D$};
\node (T) at (6,0) {$T$};
\node (U) at (3,-2) {$U$};

\draw[arrow] (Z) -- (D);
\draw[redarrow] (Z) -- node[below] {$\alpha$} (U);
\draw[arrow] (D) -- node[below] {$\beta_0$} (T);

\draw[redarrow]
  (Z) to[out=30,in=150] node[above] {$\gamma$}
  (T);

\draw[red, line width=0.8pt];

\draw[arrow] (U) -- (D);
\draw[arrow] (U) -- (T);

\end{tikzpicture}
\caption{Invalid instrumental variable $\Z$ due to violation of IV validity requirements (ii) independence ($\alpha\neq 0$) and (iii) exclusion restriction ($\gamma\neq 0$).}
\label{fig:iv}
\end{figure}
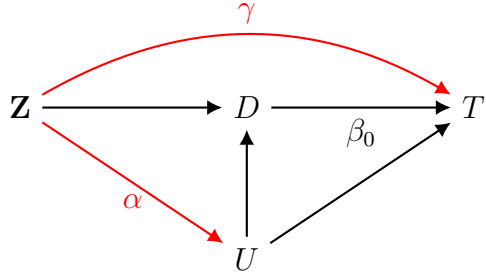
 
Under model \eqref{eq:AFTT}, $\exp(\beta_0)$ represents the causal effect associated with a one-unit increase in the exposure while holding the instruments fixed. A positive value of $\beta_0$ corresponds to prolonged survival, whereas a negative value corresponds to shortened survival. To link the counterfactual outcome model \eqref{eq:AFTT} to the realized failure time, we impose the consistency assumption \citep{rubin1974estimating}: if an individual actually receives treatment level $(D,\Z)=(d,\z)$, then the realized failure time equals the corresponding counterfactual outcome. Substituting the realized treatment levels into the structural relation among counterfactual outcomes yields a log-linear representation for the counterfactual failure time. Consequently, \eqref{eq:AFTT} can be written as 
\begin{equation}\label{modelT}
\log(T^*) = \beta_0 D + \sum_{k=1}^{p} \phi_{k} Z_k +  \varepsilon.
\end{equation}
The error term $\varepsilon$ is defined as
$\varepsilon=\log T^*(0,\mathbf{0})-E\{\log T^*(0,\mathbf{0})\mid \Z\},$
and thus satisfies $E(\varepsilon\mid \Z)=0$. If the IV independence assumption holds, that is, $\alpha=0$, then $\bbE\{\log T^*(0,\mathbf{0})|\z\}=\bbE\{\log T^*(0,\mathbf{0})\}$ and thus $\varepsilon$ does not depend on $Z$. Despite its regression-like form, the structural model in \eqref{modelT} generally involves an endogenous exposure $D$ in observational studies. 

Before specifying the structural equation for $D$, we first introduce notation for the interaction terms. The $k$th order interaction terms $\Int_k(\Z)$ and centered $k$th order interaction terms $\Int_k(\Z;\bzeta)$, which are both $C(p,k)=\binom{p}{k}$ dimensional vectors, capture higher order interactions among the instruments, defined as
$$
\Int_k\tp(\Z) = \left\{ \prod_{j \in \Xi} Z_j : \Xi \in S_k \right\}, \text{   and   } \Int_k\tp(\Z;\bzeta) = \left\{ \prod_{j \in \Xi} (Z_j-\zeta_j) : \Xi \in S_k \right\}, 
$$
where $S_k$ denotes the collection of all distinct size $k$ subsets of $\{1,\dots,p\}$, and $\bzeta=(\zeta_1,\dots,\zeta_p)\tp$ with $\zeta_k = \mathbb{E}(Z_k)$ is the mean of the $j$th instrument. For example, if there are three instruments, then $\Int_2(\Z) = (Z_1 Z_2, Z_1 Z_3, Z_2 Z_3)\tp$, and $\Int_3(\Z) = (Z_1Z_2Z_3) $. The centering $(Z_j-\zeta_j)$ yields interaction terms $\Int_k(\Z;\bzeta)$ whose components have mean zero and play an important role in establishing Neyman orthogonality properties exploited in subsequent identification and estimation. For a range of interaction orders $q_1 \leq q_2$, we further define $\Int_{q_1:q_2}(\Z) = \{ \Int_{q_1}\tp(\Z), \dots, \Int_{q_2}\tp(\Z) \}\tp,$ and $\Int_{q_1:q_2}(\Z;\bzeta) = \{ \Int_{q_1}\tp(\Z;\bzeta), \dots, \Int_{q_2}\tp(\Z;\bzeta) \}\tp,$
which collects all interaction terms from order $q_1$ to order $q_2$. Then we can write the structural equation for $D$ as
\begin{equation}\label{modelD}
D = \sum_{k=1}^{p} \theta_k Z_k + \sum_{k=2}^p \boldsymbol{\varphi}_{k}\tp \Int_k(\Z) + \nu,
\end{equation}
where $\nu$ is the stochastic error term satisfying $E(\nu\mid \Z)=0$, and $\boldsymbol{\varphi}_{k}$ is a coefficient vector with the same dimension as $\Int_k(\Z)$.

Finally, we account for right censoring, which commonly arises in survival studies due to administrative end of follow-up, loss to follow-up, or study termination before the occurrence of the event of interest. For notational simplicity, denote $T = \log(T^*)$ as the log-transformed failure time in the following. Let $C$ denote the log-transformed censoring time, representing the time at which observation of the event process ceases. The log-transformed observed outcome is then defined as
\[
Y=\min(T,C), \qquad \delta=I(T\le C),
\]
so that only partial information about the failure time $T$ is available when censoring occurs.

\section{Identification and moment function construction}\label{sec3}

\subsection{Identification under uncensored data} 
We first outline the identification argument in the setting without censoring. The presence of unobserved confounders prevents identification of the causal effect in the SAFT model through direct regression, since the exposure is no longer conditionally exogenous given the observed instruments. To illustrate, the standard instrumental variable method identifies the causal parameter through the following equation:
$$
\bbE[(\Z-\bzeta)(T-D\beta_0)]=\bbE[(\Z-\bzeta)(\sum_{k=1}^{p} \phi_{k} Z_k + \varepsilon)]=0
$$
supposing that $\phi_{k}=0$, i.e., the candidate IVs are valid. The remaining term involving $\varepsilon$ vanishes since $\bbE(\varepsilon \mid \Z)=0$ by definition. However, with invalid IVs $\phi_{k}\neq 0$, the equation no longer holds since $\bbE[\phi_{k}(Z_k-\zeta_k)(\Z_k)]\neq 0$,  and thus $\beta$ cannot be identified.

To overcome this difficulty, we achieve causal identification by exploiting interaction terms constructed from the candidate instrumental variables. The new moment condition is given by
\be \label{eq:z1z2}
\bbE[(Z_1-\zeta_1)(Z_2-\zeta_2)(T-D\beta_0)]=\bbE[(Z_1-\zeta_1)(Z_2-\zeta_2)(\sum_{k=1}^{p} \phi_{k} Z_k + \varepsilon)]=0.
\ee
Even when $\phi_{k}\neq 0$, if the instruments are mutually independent, the interaction terms cancel out the effect induced by $\phi_{k}$, so that the resulting moment condition has mean zero. Then a closed-form solution to $\beta_0$ can be obtained by 
$$
\beta_0 = \frac{\bbE[(Z_1-\zeta_1)(Z_2-\zeta_2)T]}{\bbE[(Z_1-\zeta_1)(Z_2-\zeta_2)D]}
$$
if the denominator is non-zero: $\bbE[(Z_1-\zeta_1)(Z_2-\zeta_2)D] \neq0$. Equation \eqref{eq:z1z2} presents only one particular interaction term. In fact, any linear combination of such interaction terms can be used to construct a valid moment function. Collecting moment functions of order 2 to $q$ yields an $r(p,q)$-dimensional moment function, $\bbE[\Int_{2:q}(\Z;\bzeta)(T-\beta_0 D)]$, where $q\leq p$ and $r(p,q) = \sum_{k=2}^q C(p,k)$ is the number of interaction terms.

The identifiability conditions discussed above can be summarized in the following assumption.
\begin{assumption}\label{condidentification}
(i).(IV independence) All $p$ candidate IVs are independent.\\
(ii).(Interaction relevance) $\bbE\left[ \Int_{2:q}(\Z;\bzeta_0)D\right]\neq \mathbf{0}$, that is, at least one constructed interaction term is correlated with the exposure.
\end{assumption}
Assumption \ref{condidentification} (i) requires that all $p$ candidate IVs are mutually independent. This condition can be satisfied in genetic studies, as one can perform linkage disequilibrium clumping to obtain an independent set of variants before the MR analysis \citep{purcell2007plink,Bowden2016,Zhao2020}. Assumption \ref{condidentification} (ii) requires that at least one constructed interaction term is correlated with the exposure. Let $m$ denote the total number of interaction terms. Then assumption \ref{condidentification} (ii) implies that our framework achieves $r(p,q)$-robust identification, i.e., identification is guaranteed as long as at least one component of the $r(p,q)$-dimensional vector $\Int_{2:q}(\Z;\bzeta_0)D$ is nonzero.

Figure \ref{fig:identification} illustrates the motivation for using interaction-based instruments. Identification can be achieved by constructing interaction-based instruments that avoid the violation of dependence and exclusion restriction.
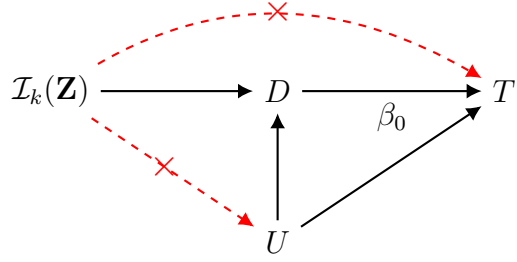
\begin{figure}[h]
\centering

\begin{tikzpicture}[
 node distance=2cm and 3cm,
 every node/.style={font=\normalsize},
 arrow/.style={-{Latex[width=2mm,length=2mm]}, line width=0.8pt},
 dashedred/.style={arrow, red, dashed}
]

\node (Z) at (0,0) {$\Int_k(\Z)$};
\node (D) at (3,0) {$D$};
\node (T) at (6,0) {$T$};
\node (U) at (3,-2) {$U$};

\draw[arrow] (Z) -- (D);
\draw[arrow] (D) -- node[below] {$\beta_0$} (T);
\draw[dashedred] (Z) -- (U);
\draw[dashedred]
  (Z) to[out=30,in=150]
  (T);

\draw[red, line width=0.8pt]
  (3,1.05) node {\large $\times$};

\draw[red, line width=0.8pt]
  (1.5,-1) node {\large $\times$};

\draw[arrow] (U) -- (D);
\draw[arrow] (U) -- (T);

\end{tikzpicture}

\caption{Identification can be achieved by constructing interaction-based instruments from the candidate instruments.}
\label{fig:identification}
\end{figure}

The moment $\Int_{2:q}(\Z;\bzeta)(T-D\beta)$ can be used as a moment function and identify the causal effect. However, the moment is not Neyman orthogonal with respect to the nuisance $\bzeta$. A formal definition of Neyman orthogonality is provided in Section \ref{sec:neyman}. Following \cite{zhang2025}, we construct a new moment function that satisfies the Neyman orthogonality property. Define $V_k = \{1, \Int_{1:k-1}^\top(\Z)\}$ as the collection of all interaction terms of $\Z$ up to order $k-1$, including an intercept. To partial out the components of $Y$ and $D$ explained by $V_k$, we regress $Y$ and $D$ on $V_k$ and denote the resulting coefficient vectors by $\vartheta_{k}$ and $\omega_{k}$, respectively, and let $\bLambda_k = \{\bzeta, \vartheta_{k}, \omega_{k}\}$. Under this setup, a moment function of order $k$ is constructed to achieve identification in the presence of potential weak and invalid instrumental variables by leveraging interaction terms of the instrumental variables, given by
$$ 
g_k(\beta,\bLambda_k;\O_T)
=  \Int_k(\Z;\bzeta)\big\{(T-V_k\vartheta_{k}) - \beta(D-V_k\omega_{k}) \big\}.
$$ 
Then aggregate interaction terms of orders up to $q$ and stack the resulting moments into a $r(p,q)$-dimensional moment function: $g(\beta,\bLambda;\O_T) = (g_2(\beta,\bLambda_2;\O_T)\tp,\dots, g_q(\beta,\bLambda_q;\O_T)\tp)\tp$, where $\bLambda = \{\bzeta,\vartheta_{2},\omega_{2},\dots,\vartheta_{q},\omega_{q}\}$. Throughout the paper, we use a subscript $0$ to indicate the true value of a nuisance function. For example, let $\bLambda_0$ denote the true value of $\bLambda$.  With the independent IV assumption, $\beta$ can be uniquely determined.

\begin{lemma}\label{lemmaidentification}
Under the structural models \eqref{modelT}, \eqref{modelD} and Assumption \ref{condidentification}, $\beta_0$ is the unique solution to the moment function
$$
\bbE g(\beta,\bLambda_0;\O_T) = 0.
$$
\end{lemma}

\subsection{AIPCW moment function under right censoring} 
In the presence of right censoring, the moment function $g_k(\beta,\bLambda;\O_T)$ cannot be evaluated since the true failure time $T$ is not always observed. To accommodate censoring, we construct an augmented inverse probability censoring weighting (AIPCW) moment function. AIPCW methods originate from the general semiparametric theory for censored data developed by \cite{robins1992recovery}. In contrast, the AFT literature has primarily relied on the Buckley-James estimator \citep{buckley1979linear}, rank-based estimating equations \citep{jin2003rank}, and inverse probability of censoring weighted estimating equations without augmentation \citep{fan2011fitting}. To the best of our knowledge, there appears to be limited work on explicit AIPCW moment function tailored to the AFT model. We generalize the estimator of \cite{robins1994estimation} and propose an AIPCW estimator under the GSAFT model:
\be\label{eq:psik} 
&&\psi_k(\beta,\boldeta_k;\O)= \frac{\delta}{G(Y| \Z,D)}g_k(\beta,\bLambda;\O)\n\\
&&-\frac{\delta}{G(Y| \Z,D)}\xi_{k}(\beta;Y,\Z,D)+\xi_{k}(\beta;-\infty,\Z,D)+\int_{-\infty}^{Y}\frac{d\,\xi_{k}(\beta;u,\Z,D)}{G(u| \Z,D)},  
\ee
where
\[
G(u| \Z,D)=P(C > u| \Z,D),\quad \xi_{k}(\beta;u,\Z,D)=\mathbb{E}\big[g_k(\beta,\bLambda;\O_T)| T\ge u, \Z,D\big],
\]
and $\boldeta_k=\{G,\xi_{k},\bLambda_k\}$ collects the nuisance functions. The lower limit of integration is $-\infty$, as we apply a logarithmic transformation to the time variable from $Y=\log(Y^*)$, resulting in support on $(-\infty, \infty)$.

The AIPCW-type moment function contains two parts.  
The first term is an inverse probability censoring weighting term:
\[
\frac{\delta}{G(Y| \Z,D)}g_k(\beta,\bLambda;\O).
\]
It reweighs the uncensored data moment $g_k(\beta,\bLambda;\O)$ by the inverse probability of being uncensored, thereby amplifying the contribution from uncensored observations $(\delta=1)$.  
The second part is the augmentation term
\[
-\frac{\delta}{G(Y| \Z,D)}\xi_{k}(\beta;Y,\Z,D)+\xi_{k}(\beta;-\infty,\Z,D)+\int_{-\infty}^{Y}\frac{d\,\xi_{k}(\beta;u,\Z,D)}{G(u| \Z,D)}.
\]
The resulting moment function is constructed to satisfy Neyman orthogonality and to enjoy a doubly robust property. The inverse probability term uses only uncensored observations and reweighs them by the inverse censoring survival probability $G$. As a result, its performance depends directly on the estimation accuracy of $G$, and the resulting moment function can be sensitive to perturbations in the nuisance parameters. The augmentation term is introduced to offset this sensitivity. In particular, $\xi_k$ uses the observed data and the censoring mechanism to impute the unobserved contribution of a censored individual to the moment function. This allows censored observations to contribute to the moment function rather than being discarded altogether. As a result, the augmentation component plays a crucial role in ensuring that the resulting moment function is Neyman orthogonal with respect to the nuisance parameters and enjoys a doubly robust property. These properties will be formally established and discussed in the next subsection.

Similar to the case without censoring, we stack the interaction of orders up to $q$ and obtain the final $r(p,q)$-dimensional AIPCW moment function as $\psi(\beta,\boldeta;\O) = (\psi_2(\beta,\boldeta_2;\O)\tp, \dots, \psi_q(\beta,\boldeta_q;\O)\tp)\tp$, with $\boldeta = (G,\boldeta_2,\xi_2,\dots,\boldeta_q,\xi_{q},\bLambda)$ collecting all nuisance parameters. Also, let $\boldeta_0$ denote the collection of all components of $\boldeta$ evaluated at their true values. We now show that the proposed AIPCW moment function \eqref{eq:psik} is uniquely satisfied at the true parameter $\beta_0$ with the following assumptions. Let $\tau <\infty$ be the maximum follow-up horizon after logarithmic transformation.
\begin{assumption}[Conditional non-informative censoring]\label{condindependent}
$T\perp C |\Z,D$.
\end{assumption} 

\begin{assumption}[Positivity of censoring distribution]\label{condposint}
There exists a constant $c>0$ such that
$$
P(C> \tau\mid \Z,D)\ge c.
$$
\end{assumption}
Assumption \ref{condindependent} states that, conditional on $\Z$ and $D$, the censoring mechanism is independent of the failure time, which is commonly used in the analysis of censored data \citep{cai2009,tang2020}. Assumption \ref{condposint} rules out vanishing censoring probabilities and prevents the denominator in the inverse-probability weights from approaching zero \citep{cui2023estimating}.

\begin{theorem} \label{thmidentification}
Under the structural models \eqref{modelT}, \eqref{modelD} and Assumption \ref{condidentification}-\ref{condposint}, $\beta_0$ is the unique solution to the moment function $$\mathbb{E}\psi(\beta,\boldeta_0;\O)=0.$$ 

\end{theorem}
Theorem \ref{thmidentification} shows that the causal effect $\beta_0$ is identified by the proposed AIPCW moment function under additional assumptions on censoring. Similar to Lemma \ref{lemmaidentification}, the identification condition remains unchanged. In particular, $\beta_0$ is identified if and only if at least one constructed interaction term is correlated with the exposure. One additional remark is that when censoring rate is zero, the censoring probability satisfies $G_0(Y \mid \Z,D)\equiv1$, and the additional censoring adjustments vanish and the moment function $\psi$ reduces to its uncensored form $g$, confirming coherence between the censored and uncensored settings.

\subsection{Neyman orthogonality and double robustness of  the proposed AIPCW moment function}\label{sec:neyman}

The AIPCW moment function $\psi(\beta,\boldeta;\O)$ defined in \eqref{eq:psik} depends on several unknown components, including the censoring-related functions $G_0$ and $\xi_0$ and $\bzeta,\vartheta, \omega$ collected in $\bLambda_0$. These components are not of primary interest for the target parameter $\beta_0$ and are treated as nuisance functions. Since these nuisance functions are typically estimated using nonparametric methods, their estimation errors may converge at rates slower than $n^{-1/2}$, such as the estimation of $G$ and $\xi$ \citep{he2013,tang2020}. To ensure valid inference for $\beta_0$ in such settings, it is essential that the moment function possesses robustness properties that mitigate the impact of nuisance estimation.

A central property of the proposed AIPCW moment is Neyman orthogonality \citep{neyman1959,neyman1979,Chernozhukov2018}. Let $\boldeta_0$ denote the true nuisance functions. The moment function $\psi(\beta,\boldeta;\O)$ is said to be Neyman orthogonal at $(\beta_0,\boldeta_0)$ if its expectation is locally insensitive to perturbations of the nuisance functions. Formally, Neyman orthogonality requires that
\[
\left.
\frac{\partial}{\partial t}
\bbE\!\left[
\psi\bigl(\beta_0,\boldeta_0+t(\boldeta-\boldeta_0);\O\bigr)
\right]
\right|_{t=0}
=0
\]
for all admissible directions $\boldeta-\boldeta_0$. This condition eliminates the first-order contribution of nuisance estimation error in the asymptotic expansion of the estimator for $\beta_0$ \citep{bickel1993efficient,Chernozhukov2018}, thereby preserving root-$n$ consistency and asymptotic normality even when nuisance functions are estimated at slower rates.

For the proposed construction, the augmentation term in \eqref{eq:psik} is chosen to enforce Neyman orthogonality, in the sense that the directional derivative of $\bbE\{\psi(\beta_0,\boldeta;\O)\}$ with respect to $\boldeta$ equals zero at $\boldeta=\boldeta_0$. This design feature distinguishes our moment construction from AIPCW approaches that primarily target censoring adjustment \citep{tang2020}.

In addition to Neyman orthogonality, the proposed moment function is doubly robust \citep{bang2005doubly}. Specifically, we consider the following two specifications for the nuisance functions:
$$
\mathcal{M}_1:\{\xi,\bLambda \textrm{ are correctly specified}\}, \quad\quad \mathcal{M}_2:\{G,\bLambda \textrm{ are correctly specified}\}, 
$$
The proposed moment function $\psi(\beta_0,\boldeta;\O)$ has zero expectation under the union model $\mathcal{M}_1\cup\mathcal{M}_2$, that is the expectation of $\psi$ equals zero if either the exposure-related component $\xi_0$ or the conditional censoring distribution $G_0$ is correctly specified, but not necessarily both. As a result, identification and consistency of $\beta_0$ do not require simultaneous correctness of all nuisance functions. This double robustness property provides additional protection against model misspecification, which is particularly valuable in nuisance estimation.

The following theorem formalizes the Neyman orthogonality and double robustness properties of the proposed AIPCW moment function.

\begin{theorem}\label{thmneyman}
Under Assumptions~\ref{condindependent}--\ref{condposint},
\begin{enumerate}
\item[(i)] The moment function $\psi(\beta_0,\boldeta;\O)$ is Neyman orthogonal with respect to $\boldeta$ at $(\beta_0,\boldeta_0)$.
\item[(ii)] $\bbE\{\psi(\beta_0,\boldeta;\O)\}=0$ under $\mathcal{M}_1\cup\mathcal{M}_2$.
\end{enumerate}
\end{theorem}

Together, Neyman orthogonality and double robustness ensure that the proposed moment function is well suited for estimation and inference in right-censored models. Neyman orthogonality shields the asymptotic distribution from first-order nuisance estimation error, while double robustness guarantees consistency even when part of the model is misspecified \citep{bickel1993efficient,Chernozhukov2018}.

\section{Estimation procedure}\label{sec4}

This section presents an implementation of the proposed estimator based on the AIPCW moment functions developed in Section~\ref{sec3}.

We begin by estimating the nuisance functions required for constructing the AIPCW moment functions. To ensure asymptotic independence between nuisance function estimation and moment evaluation, we employ a cross-fitting scheme \citep{Chernozhukov2018}. Specifically, the sample is randomly partitioned into two folds; nuisance functions are estimated on one fold and the corresponding moment functions are evaluated on the other. The roles of the two folds are then interchanged, and the resulting moment functions are averaged.

Let $\hat\zeta_k = n^{-1}\sum_{i=1}^n Z_{k,i}$ denote the sample mean of the $k$th instrument, and let $\hat\bzeta = (\hat\zeta_1,\ldots,\hat\zeta_p)$. For each interaction order $k$, we regress $Y$ on $V_k$ and $D$ on $V_k$ separately, and denote the resulting coefficient estimators by $\hat\vartheta_{k}$ and $\hat\omega_{k}$. Using these estimators, the uncensored-data interaction-based moment function of order $k$ is defined as
\[
g_k(\beta,\hat\bLambda_k;\O)
=
\Int_k(\Z;\hat\bzeta)
\Big\{(Y - V_k^\top \hat\vartheta_{k}) - \beta (D - V_k^\top \hat\omega_{k})\Big\},
\]
where $\hat\bLambda_k=\{\hat\bzeta,\hat\vartheta_{k},\hat\omega_{k}\}$ collects the nuisance functions in uncensored-data moment function $g$ of $k$th order interaction.

To account for right censoring, we apply an AIPCW transformation to $g_k$. The conditional censoring survival function
$G(y\mid \Z,D)=\Pr(C\ge y\mid \Z,D)$
is estimated using a local Kaplan-Meier estimator \citep{dabrowska1989,tang2020},
\[
\widehat G(y\mid \Z,D)
=
\prod_{i=1}^n
\left[
1-
\frac{B_{ni}(\Z,D)}
{\sum_{j=1}^n I(Y_j\ge Y_i)B_{nj}(\Z,D)}
\right]^{I(Y_i\le y,\delta_i=0)},
\]
where
\[
B_{nj}(\Z,D)
=
\frac{
K\!\left(\frac{(\Z,D)-(\Z_j,D_j)}{h}\right)
}{
\sum_{i=1}^n K\!\left(\frac{(\Z,D)-(\Z_i,D_i)}{h}\right)
}.
\]
Here $K(\cdot)$ is a kernel function and $h$ is a bandwidth.

Next, we estimate 
$\xi_k(\beta;u,\Z,D)=\bbE\{g_k(\beta,\bLambda_k;\O_T)\mid Y\ge u,\Z,D\}$
by
\[
\widehat\xi_k(\beta;u,\z,d)
=
\frac{
\sum_{j=1}^n I(Y_j\ge u)\,B_{nj}(\z,d)\,
\frac{\delta_j}{\widehat G(Y_j\mid \z,d)}\,
g_k(\beta,\hat\bLambda_k;\O_j)
}{
\sum_{j=1}^n I(Y_j\ge u)\,B_{nj}(\z,d)\,
\frac{\delta_j}{\widehat G(Y_j\mid \z,d)}
}.
\]

Finally, $\int_{-\infty}^{Y} {d\,\xi_{k}(\beta;u,\Z,D)}/{G(u| \Z,D)}$ is approximated by a sum over the observed failure times. Let $u_1<\cdots<u_K$ denote the distinct failure times. For subject $i$, we approximate
$\int_{-\infty}^{Y_i} d\widehat\xi_k(\beta;u,\Z_i,D_i)/\widehat G(u\mid \Z_i,D_i)$
by
\[
\sum_{u_k\le Y_i}
\frac{\Delta_{\widehat\xi_k}(u_k;\Z_i,D_i)}{\widehat G(u_k\mid \Z_i,D_i)},
\]
where
$\Delta_{\widehat\xi_k}(u_k;\Z_i,D_i)
=
\widehat\xi_k(\beta;u_k,\Z_i,D_i)
-
\widehat\xi_k(\beta;u_{k-1},\Z_i,D_i)$
and $u_0=-\infty$.

The resulting AIPCW-transformed observed data moment function of order $k$ is
\bse
\psi_k(\beta,\hat\boldeta_k;\O)&=&\frac{\delta}{\widehat G(Y\mid \Z,D)}\Big\{g_k(\beta,\hat\bLambda_k;\O)-\widehat\xi_k(\beta;Y,\Z,D)\Big\}\\
&&+\widehat\xi_k(\beta;-\infty,\Z,D) +\int_{-\infty}^{Y}\frac{d\widehat\xi_k(\beta;u,\Z,D)}{\widehat G(u\mid \Z,D)},
\ese
where $\hat\boldeta_k$ collects all nuisance function estimators involved in the $k$th-order moment function.

After constructing the AIPCW moment functions, the remaining task is to aggregate a potentially large collection of interaction-based moment functions into a single estimator of $\beta$. In practice, particularly in genetic applications, it is common that the instruments and their interaction terms are only weakly correlated with the treatment variable $D$ \citep{davies2015many}. This weak association implies that the corresponding coefficient ${\boldsymbol{\varphi}_{k}}$ in the structural equation \eqref{modelD} is small, reflecting limited variation in $D$ explained by the interaction-based instruments. In the presence of many and weakly informative moment functions, conventional quadratic GMM estimators \citep{Hansen1982} may suffer from substantial weak-instrument bias \citep{newey2009}. We therefore aggregate the moment functions using a generalized empirical likelihood (GEL) criterion \citep{owen1988,qin1994,kitamura1997}. Relative to quadratic GMM, GEL avoids committing to a fixed quadratic approximation and instead determines data-adaptive implied probabilities through a tilting parameter, which can stabilize estimation when the moment functions are nearly redundant, when sample moments and their derivatives are highly correlated, or when identification is weak \citep{newey2004}.

We aggregate the moment functions and estimate $\beta$ by minimizing the GEL objective
\[
\hat\beta
=
\argmin_{\beta\in\calB}\hat Q(\beta,\hat\boldeta),
\qquad
\hat Q(\beta,\hat\boldeta)
=
\sup_{\blambda\in L(\beta,\hat\boldeta)}
\frac{1}{n}\sum_{i=1}^n
\rho\!\left(\blambda^\top \psi(\beta,\hat\boldeta;\O_i)\right),
\]
where $\rho(\cdot)$ is a concave tilting function and
$L(\beta,\hat\boldeta)
=
\{\blambda:\blambda^\top\psi(\beta,\hat\boldeta;\O_i)\in\calV,\ i=1,\ldots,n\}$.
Different choices of $\rho$ correspond to familiar members of the GEL family, including empirical likelihood with $\rho(v)=\log(1-v)$ \citep{owen1988}, exponential tilting with $\rho(v)=-\exp(v)+1$ \citep{qin1994}, and the continuous updating estimator with $\rho(v)=-v-v^2/2$ \citep{Hansen1982,Imbens2002}. These criteria behave similarly in a neighborhood of zero, but they differ substantially in their second-order derivatives, leading to distinct curvature properties. Empirical likelihood and exponential tilting exhibit state-dependent curvature, whereas the continuous updating estimator has constant curvature and is known to suffer from larger higher-order bias and less favorable finite sample performance \citep{kitamura2004}. For these reasons, we primarily use empirical likelihood or exponential tilting in implementation.

Although GEL provides a principled way to aggregate many moment functions, the construction of interaction-based moments can still pose practical challenges in applications. In particular, the number of interaction terms $\Int_k(\Z)$ grows combinatorially with the dimension of the instrument vector $\Z$. For a $p$-dimensional $\Z$, interaction terms of order 2 to $q$ involves $r(p,q)=\sum_{k=2}^q C(p,k)$ terms, so even moderate values of $p$ can lead to an extremely large set of candidate interaction terms. To mitigate this dimensionality burden while retaining interaction terms that contribute meaningfully to identification, we introduce a preliminary screening step based on penalized regression.

A variety of penalization methods can be used for this screening step, including lasso \citep{tibshirani1996}, ridge regression \citep{hoerl1970}, SCAD \citep{fan2001}, and MCP \citep{zhang2010}. For simplicity, we employ the adaptive lasso \citep{zou2006} in our implementation. We also conducted a brief comparison with alternative selection procedures and found that the resulting estimates are qualitatively similar across methods. Since identification in our setting requires only that at least one interaction moment be nonzero, the procedure does not rely on selection consistency, and exact recovery of the true nonzero interaction terms is unnecessary. The penalized selection step functions primarily as a dimension-reduction device that removes clearly uninformative components rather than perfectly separating zero from nonzero interactions. Identification is preserved as long as the selected set contains at least one interaction term associated with a nonzero moment condition. As a result, moderate selection errors do not undermine identification, making the method comparatively robust to imperfect model selection.

Suppose we retain an $m$-dimensional moment function via penalized regression, then the above steps together define the proposed estimation procedure, which is summarized in Algorithm~\ref{alg:estimation}. Note that the nuisances in the moment function of order $k$ depend explicitly on $k$. Therefore, the associated nuisance functions must be estimated individually for each $k$, rather than pooled across orders.

\begin{algorithm}[H]
\caption{Estimation Procedure}
\label{alg:estimation}
\begin{algorithmic}[1]
\State Split the sample into two folds.
\For{$k=2,\ldots,q$}
    \For{each fold}
        \State Estimate nuisance functions $\hat\boldeta_k= (\hat\bzeta,\hat\vartheta_{k},\hat\omega_{k},\hat G,\hat\xi_k)$ on the auxiliary fold.
        \State Construct AIPCW moment functions $\psi_k$ on the evaluation fold.
    \EndFor
    \State Concatenate the two-fold moment functions to obtain $\psi_k(\beta,\hat\boldeta_k;\O)$.
\EndFor
\State Stack moment functions across orders $\psi(\beta,\!\hat\boldeta;\!\O)\!=\!(\psi_2(\beta,\!\hat\boldeta_2;\!\O)\tp,\ldots,\psi_q(\beta,\!\hat\boldeta_q;\!\O)\tp\!)\tp$.
\State Select $m$ interaction terms from all $r(q,p)$ terms via penalized regression.
\State Form the final $m$-dimensional moment vector $\psi^{\dagger}(\beta,\hat{\boldeta};\O)$ using the selected components.
\State Estimate $\hat\beta$ by minimizing the GEL criterion $\hat Q(\beta,\hat\boldeta)$ based on the selected moments.
\end{algorithmic}
\end{algorithm}

\section{Asymptotic properties}\label{sec5}

This section establishes the consistency and asymptotic normality of the proposed estimator. We first introduce the required notation. Let $\psi'(\boldeta;\O)=\partial \psi(\beta,\boldeta;\O)/\partial \beta$, $\psi^* = \bbE\psi'(\boldeta_0;\O)$. Moreover, let $\Omega(\beta,\boldeta;\O) = \bbE \psi(\beta,\boldeta;\O)\psi(\beta,\boldeta;\O)\tp$, and  $\Omega_0 = \Omega(\beta_0,\boldeta_0;\O)$. $\lambda_{\max}(\Omega)$ and $\lambda_{\min}(\Omega)$ represent the maximal and minimal eigenvalues of matrix $\Omega$. $\|\cdot\|$ denotes the Euclidean norm for vectors and the Frobenius norm for matrices. For a measurable function $f$, $\|f\|$ denotes its $L_2$ norm with respect to the underlying probability measure.

We next give some regularity conditions required to establish the asymptotic properties.

\begin{condition}[Many weak moment asymptotics]\label{condweak}
(i). The total number of moments $m$ satisfies $m^3/n \to 0$.

(ii). There exist scalars $\mu_n, c, c^{\prime}>0$, such that 
$$
\mu_n^2 c \leq n \psi^{*\top}\Omega_0^{-1} \psi^{*} \leq \mu_n^2 c^{\prime} ,
$$
$\mu_n$ satisfies $\mu_n \rightarrow \infty$ as $n \rightarrow \infty$ and $m / \mu_n^2$ is bounded. 
\end{condition}
Condition~\ref{condweak} (i) imposes the conditions that $m^{3}/n\to 0$, which restricts how quickly the number of moment functions may diverge. Here $m$ denotes the number of constructed interaction-based moments rather than the dimension of the original instrument vector $Z$, so it can be larger than $p$. The condition ensures that the expanding moment set remains well behaved for the asymptotic analysis. Condition \ref{condweak} (ii) places restrictions on the concentration parameter $n\psi^{*\top}\Omega_0^{-1}\psi^{*}$ \citep{Staiger1997}, which summarizes the overall strength of the moment functions. When $\mu_n$ scales as $\sqrt{n}$, it falls into the classical strong moment regime. In contrast, our framework permits $\mu_n$ to increase more slowly than $\sqrt{n}$, allowing the concentration parameter to diverge at a slower rate and therefore accommodating many weak interactions. The quantity $\mu_n$ governs the convergence rate of the estimator $\widehat{\beta}$, and all asymptotic results remain valid without imposing the conventional $\sqrt{n}$ scaling. 
 
\begin{condition}[Convergence rate of nuisance]\label{condkm}
 The estimator of $G$ and $\xi$ satisfy $\|\hat G-G_0\|=o_p(n^{-1/4}m^{-1/4})$, and $\|\hat \xi-\xi_0\|=o_p(n^{-1/4}m^{-1/4})$.
\end{condition}
Condition \ref{condkm} imposes a generic requirement on the convergence rates of the estimators $\widehat G$ and $\widehat\xi$. The condition does not rely on any specific estimation procedure, and other approaches, such as Cox model \citep{cox1972}, may also be used provided that the required convergence rates are achieved. It only requires that the estimators achieve rates of order $o_p(n^{-1/4}m^{-1/4})$, which is commonly required in double robustness literature, such as in \cite{Chernozhukov2018}. In our method, $G_0$ is estimated using a local Kaplan-Meier estimator, which attains the rate of $O((\log(n)/(nh^p))^{1/2})$ \citep{tang2020}, where $h$ is the bandwidth. Further regulartity conditions are included in the supplementary material. 
\begin{theorem}\label{thmnormality}
   Under Conditions \ref{condweak}-\ref{condkm} and Conditions 3-5 in the supplementary material, $\widehat\beta$ is asymptotically normal as $n\to \infty$, i.e., 
$$
\frac{\mu_n(\widehat\beta-\beta_0)}{\sqrt{H^{-1}(H+V)H^{-1}}} \to N(0,1),
$$
where $H = n\psi^{*\top}\Omega_0^{-1}\psi^{*}$, $V = \mu_n^{-2}\bbE[U_i\tp\Omega_0^{-1}U_i]$, and  \\$U_i = \psi'(\boldeta_0;\O_i)- \psi^{*}-\{\Omega_0^{-1} \bbE(\psi(\beta_0,\boldeta_0;\O) \psi^{'}(\boldeta_0;\O)\tp)\}\tp \psi(\beta_0,\boldeta_0;\O_i)$ is the population residual of the least squares regression of $\psi^{'}(\boldeta_0;\O_i)-\psi^{*}$ on $\psi(\beta_0,\boldeta_0;\O_i)$.
\end{theorem}

Theorem \ref{thmnormality} establishes that $\hat\beta$ attains a convergence rate of $\mu_n$, which is typically slower than the usual $\sqrt{n}$ rate. This loss in convergence rate is the price paid for weak instruments, under which the moment conditions provide substantially less information than that in the strong moment setting. The variance of $\hat\beta$ contains two distinct components. The first variance term $H^{-1}$ matches the standard GMM variance in settings with strong instruments \citep{Hansen1982}. The second variance component $H^{-2}V$ originates from the $U$-statistic term associated with weak interactions This additional component does not vanish when the number of interactions diverges and plays a central role in delivering valid inference. When all moments are strong and the number of interaction terms stays fixed, the $U$-statistic variance converges to zero, recovering the classical GMM variance $H^{-1}$, which coincides with the variance lower bound under the strong moment regime \citep{tang2020}. In contrast, in the presence of weak moments, the additional variance term reflects the intrinsic loss of information induced by weak identification. The variance lower bound in such weak moment settings remains an open question and deserves further investigation. 

Figure~\ref{fig:weak} illustrates the influence of the strength of the moment conditions $\mu_n$. Panel (a) shows that, under weaker moments, the norm $\|\psi\|$ exhibits smaller curvature as $\beta$ deviates from the true value $\beta_0$, leading to a flatter curve and a more challenging estimation problem for $\beta$. Panel (b) further demonstrates the impact of weak moments on inference. In particular, the asymptotic variance of the estimator involves the term $V = \mu_n^{-2}\bbE[U_i\tp\Omega_0^{-1}U_i]$, which increases as $\mu_n$ decreases. This variance inflation is reflected in the progressively more dispersed density of the GEL estimator in Panel (b), and when $\mu_n=1$ so that the moment strength does not diverge with $n$, the resulting limiting distribution departs substantially from normal distribution.
\begin{figure}[htbp]
\centering
\begin{subfigure}{0.49\linewidth}
  \centering
  \includegraphics[width=\linewidth]{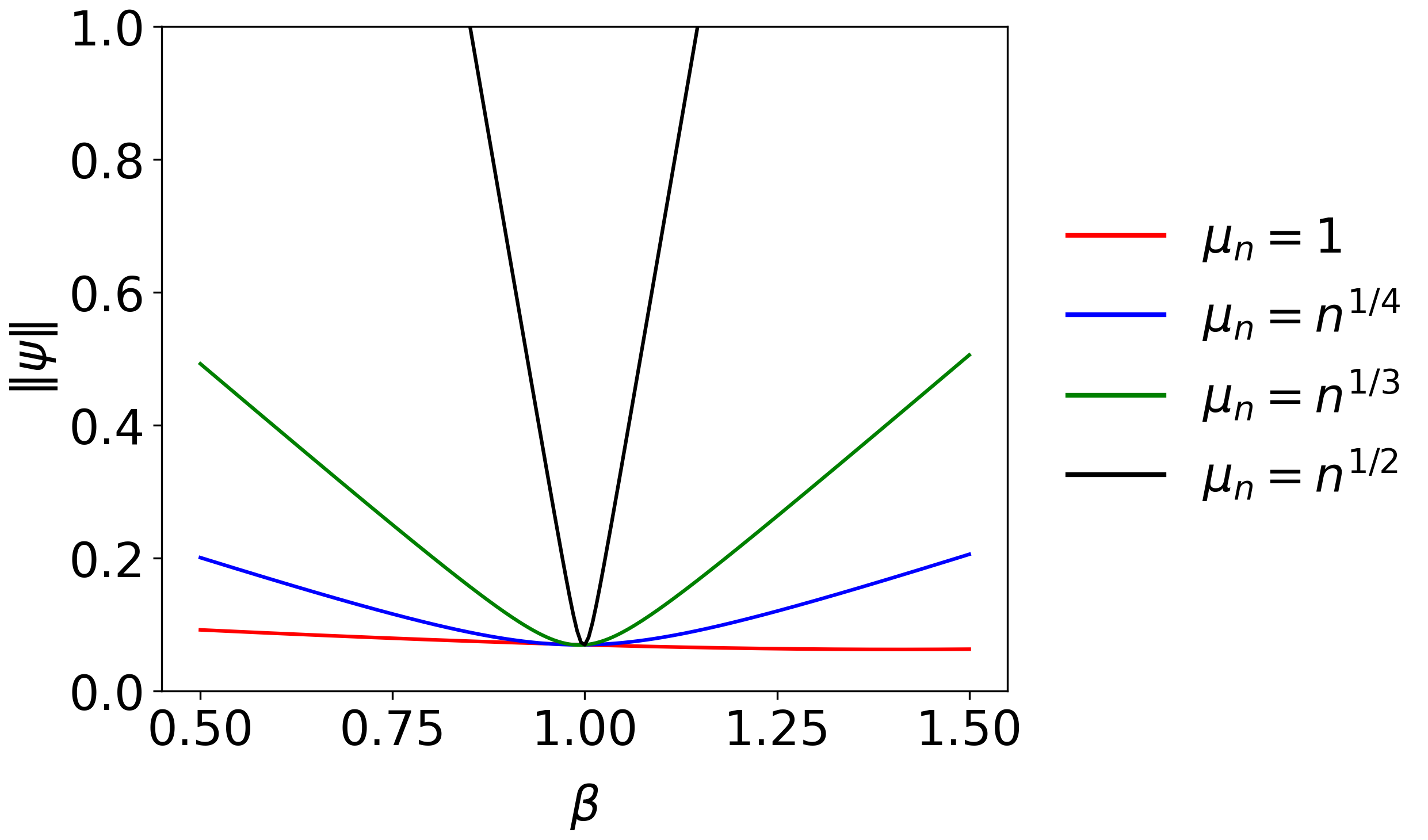}
  \caption{Norm of $\psi$ \quad \quad  \quad  \quad \quad  }
\end{subfigure}
\hfill
\begin{subfigure}{0.49\linewidth}
  \centering
  \includegraphics[width=\linewidth]{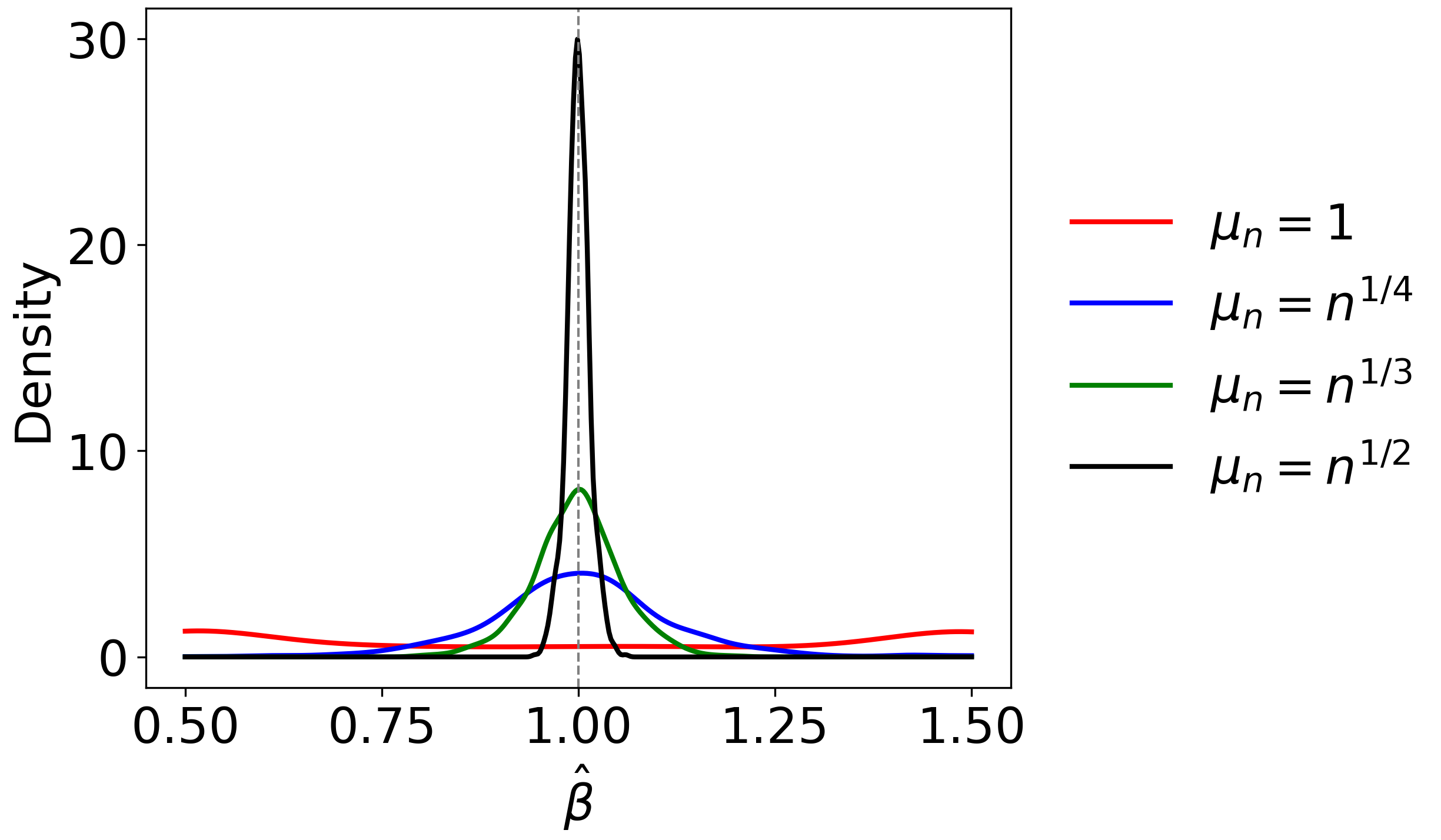}
  \caption{Density of $\hat\beta$  \quad \quad  \quad  \quad \quad  }
\end{subfigure}
\caption{Behavior of the moment function and the GEL estimator under different moment strengths indexed by $\mu_n$. The true value is $\beta_0=1$. Panel (a) plots the $\|\psi\|$ curves under varying moment strengths, where smaller $\mu_n$ yields a flatter norm curve. When $\mu_n=1$ ( $\mu_n$ does not diverge with $n$), the curve becomes flat, indicating lack of identification. Panel (b) displays the density of the GEL estimator of $\beta$ under the same moment regimes. As $\mu_n$ decreases, the density becomes more dispersed with larger variance; when $\mu_n=1$, the limiting distribution tends to be flat.}
\label{fig:weak}
\end{figure}

A distinctive feature of our result is that the moment function is Neyman orthogonal to the nuisance functions. Neyman orthogonality eliminates the first-order impact of estimating $\boldeta_0$, ensuring that the asymptotic distribution depends only on the limiting behavior of the moment functions themselves, as long as the nuisance estimators achieve a minimal convergence rate required by the Neyman orthogonality condition. As a result, the influence of nuisance estimation is asymptotically negligible, and the asymptotic variance takes a clean two-component form driven solely by the structure of the weak interactions.

In practice, the variance term $H$ and $V$ can be consistently estimated by
$$
\hat H = \frac{\partial^2\hat Q(\hat\beta, \hat\boldeta)}{\partial\beta^2 },\quad \hat V_1= \hat H^{-1} \hat D\tp\bar\Omega(\hat\beta,\hat\boldeta)^{-1}\hat D\hat H^{-1}, 
$$
respectively, where
$$
\hat D = \sum_{i=1}^n \frac{\rho^{'}(\hat\blambda(\hat\beta,\hat\boldeta)\tp \psi(\hat\beta,\hat\boldeta;\O_i))\psi^{'}(\hat\beta,\hat\boldeta;\O_i)}{\sum_{j=1}^n \rho^{'}(\hat\blambda(\hat\beta,\hat\boldeta)\tp \psi(\hat\beta,\hat\boldeta;\O_j))},
$$
and $\bar\Omega(\hat\beta,\hat\boldeta) = 1/n\sum_{i=1^n}\psi(\hat\beta,\hat\boldeta;\O_i)\psi(\hat\beta,\hat\boldeta;\O_i)\tp.$
\section{Model diagnoses}\label{sec6}

When more moment conditions are imposed than are needed to identify the parameter of interest, it becomes important to assess whether these restrictions are compatible with the data. This assessment can be formalized as a hypothesis test on the joint validity of the $m$-dimensional moment functions:
\[
H_0:\ \bbE\{\psi(\beta_0,\boldeta_0;\O)\}=\mathbf{0} .
\]
Under correct specification, the moment conditions are satisfied at the true parameter values. This motivates a specification test analogous to the classical $J$-test in the GMM literature \citep{sargan1958,Hansen1982}, which assesses the joint validity of overidentifying restrictions by examining whether the sample moments are sufficiently close to zero.

Specifically, we employ the statistic $2n\hat Q(\hat\beta,\hat\boldeta)$ \citep{newey2009,ye2024} to quantify the aggregate deviation of the sample moments from their population counterparts. Under appropriate regularity conditions, this statistic admits a well-defined asymptotic distribution under the null hypothesis and can therefore be used to conduct a formal overidentification test, as stated in the following theorem.

\begin{theorem}[Overidentification test]\label{thm:overid}
Under Conditions \ref{condweak}-\ref{condkm} and Conditions 3-5 in the supplementary material, and  $m,n\to\infty$, if
\[
\bbE\{\psi(\beta_0,\boldeta_0;\O)\}=\mathbf{0},
\]
then
\[
P\!\left(2n\widehat Q(\widehat\beta,\widehat\boldeta) \ge \chi^2_{1-\alpha}(m-1)\right)\to \alpha ,
\]
as $n\to\infty$.
\end{theorem}

Theorem~\ref{thm:overid} implies a simple decision rule for assessing the validity of the imposed moment conditions. At significance level $\alpha$, the null hypothesis is rejected whenever the value of $2n\widehat Q(\widehat\beta,\widehat\boldeta)$ exceeds the corresponding chi-squared critical value. Rejection of the null hypothesis suggests that the moment conditions may fail.

\section{Simulation studies}\label{sec7}

We assess the finite sample performance of the estimator under several right censoring scenarios. The candidate instruments $Z_1,\ldots,Z_p$ are generated independently from the standard normal distributions. The structural equations for the outcome and exposure are
$$
T = \beta_0 D + \sum_{k=1}^p \phi_{k} Z_k + \varepsilon,
$$
$$
D = \sum_{k=1}^p \theta_k Z_k
    + \sum_{1 \le j < k \le p} \varphi_{D,jk} Z_jZ_k
    + \nu,
$$
where $\varepsilon$ and $\nu$ are noise terms, generated from a joint normal distribution with mean zero, variance 0.4, and covariance 0.2. The true causal effect is $\beta_0 = 1$. The censoring time $C$ is generated from a uniform distribution $[\tau_1,\tau_2]$. Across different simulation settings, $\tau_1$ and $\tau_2$ are varied to produce empirical censoring rates of approximately $20\%, 40\%, \text{and } 60\%$, respectively. The observed failure times are
$Y = \min(T,C)$, and indicator $\delta = I\{T \le C\}.$

The interaction coefficients decay with sample size: $\varphi_{D,jk} = c n^{-1/4}$ with $c = 4$, representing a weak instrumental variables setting. We test two different IV dimensions: 10 and 20. When $p = 10$, all interaction terms are set to be nonzero. When $p = 20$, we randomly select $40\%$ of the possible interactions and set the corresponding coefficients to be nonzero, with the remaining interactions set to be zero. In the estimation procedure, we use only second-order interaction terms in line with the data-generating process. Moreover, we examine four configurations for $(\theta_k,\phi_{k})$, controlling instrument strength and the proportion of invalid instruments:

\begin{itemize}
    \item \textbf{Case 1:}
    $\theta_k = 1$ for all $k$, and $30\%$ of SNPs are invalid with $\phi_{k} = 0.2$.
    \item \textbf{Case 2:}
    $\theta_k = 1$ for all $k$, and $60\%$ of SNPs are invalid with $\phi_{k}$ taking values $0.2$, $0.4$, and $0.6$ in equal proportions.
    \item \textbf{Case 3:}
    $\theta_k \sim N(1, 1)$ and $\phi_{k} \sim N(0.2, 0.2)$ independently.
    \item \textbf{Case 4:}
    $\theta_k \sim N(1, 1)$ and $70\%$ of SNPs are invalid with $\phi_{k} = 0.5\theta_k$.
\end{itemize}

We consider sample sizes $n = 10000$ and $20000$ and repeat each configuration $500$ times. For each simulated data set we implement the estimator in three steps.  
First, we apply adaptive LASSO to select relevant interaction terms from the set of interaction terms.   
Second, we estimate the nuisance functions $\hat\boldeta$ and use the selected set of interactions to construct the Neyman orthogonal moment $\psi(\beta,\hat\boldeta;\O)$ with sample-splitting method. Finally, we obtain the estimator of $\hat\beta$ by GEL. We also consider the AFT model as a benchmark method, where the relationship between $Y$ and $D$ is modeled directly while ignoring the unmeasured confounder $U$.

Table \ref{tabsimu1} reports the simulation results for sample sizes $n = 10000$ with censoring rate $0.2$. Additional simulation results across different sample sizes and censoring levels are provided in Supplementary Materials. The simulation results are evaluated using four summary measures: the percentage finite sample bias (Bias), the empirical standard deviation over replications (SD), the average estimated standard error (SE), and the empirical coverage probability (CP) of the 95\% confidence interval. These metrics jointly assess accuracy, sampling variability, and inferential validity across the considered settings.

\begin{table}[htbp]
\centering
\caption{Performance of the proposed estimators for $n=10{,}000$ across Cases 1-4, and instrument dimensions $p=10$ and $p=20$. BIAS denotes percentage finite sample bias, SD the empirical standard deviation, SE the average estimated standard error, and CP the empirical coverage probability of the 95\% confidence interval.}
\label{tabsimu1}
\begin{tabular}{cccccccccc}
\hline
&     & \multicolumn{4}{c}{$p=10$}        & \multicolumn{4}{c}{$p=20$}               \\ \hline
& GEL & Bias      & SD   & SE   & CP    & Bias             & SD   & SE   & CP    \\ \hline
\multirow{4}{*}{Case 1} & CUE & 0.238\%   & 0.093 & 0.087 & 0.932 & -0.511\%         & 0.066 & 0.071 & 0.968 \\
& EL  & 0.236\%   & 0.091 & 0.085 & 0.926 & -0.291\%         & 0.064 & 0.066 & 0.952 \\
& ET  & 0.263\%   & 0.091 & 0.086 & 0.934 & -0.448\%         & 0.064 & 0.067 & 0.956 \\
& AFT & -24.163\% & 0.462 & 0.008 & 0.000 & -78.848\%        & 0.418 & 0.014 & 0.000 \\ \hline
\multirow{4}{*}{Case 2} & CUE & 0.618\%   & 0.082 & 0.088 & 0.976 & -0.429\%         & 0.072 & 0.072 & 0.960 \\
& EL  & 0.544\%   & 0.081 & 0.085 & 0.966 & -0.366\%         & 0.068 & 0.066 & 0.936 \\
& ET  & 0.549\%   & 0.082 & 0.086 & 0.970 & -0.328\%         & 0.070 & 0.067 & 0.944 \\
& AFT & -32.038\% & 0.554 & 0.011 & 0.000 & -96.333\%        & 0.201 & 0.017 & 0.000 \\ \hline
\multirow{4}{*}{Case 3} & CUE & -0.047\%  & 0.086 & 0.088 & 0.962 & -0.292\%         & 0.066 & 0.072 & 0.958 \\
& EL  & -0.040\%  & 0.085 & 0.086 & 0.970 & -0.198\%         & 0.063 & 0.066 & 0.950 \\
& ET  & -0.035\%  & 0.085 & 0.086 & 0.968 & -0.228\%         & 0.064 & 0.067 & 0.956 \\
& AFT & -66.156\% & 0.498 & 0.014 & 0.008 & -99.328\%        & 0.083 & 0.014 & 0.002 \\ \hline
\multirow{4}{*}{Case 4} & CUE & 0.995\%   & 0.091 & 0.088 & 0.948 & 0.063\%          & 0.070 & 0.072 & 0.960 \\
& EL  & 0.943\%   & 0.089 & 0.086 & 0.940 & 0.065\%          & 0.068 & 0.066 & 0.938 \\
& ET  & 1.006\%   & 0.089 & 0.086 & 0.942 & 0.031\%          & 0.069 & 0.067 & 0.946 \\
& AFT & -80.035\% & 0.444 & 0.017 & 0.000 & -100.005\% & 0.001 & 0.014 & 0.000 \\ \hline
\end{tabular}
\end{table}
For both $n=10000$ and $n=20000$, the three GEL estimators, CUE, EL, and ET, exhibit stable performance across all four cases. The estimated biases remain small, while the SD closely matches the SE, indicating accurate variance estimation. The empirical coverage probabilities stay near the nominal level of $0.95$ across censoring levels and instrument dimensions, demonstrating that the proposed AIPCW estimator delivers reliable inference in settings with many weak and invalid instruments. Among the three GEL variants, CUE exhibits modestly larger bias in some settings, particularly under large censoring rate, whereas EL and ET remain comparatively stable.

In contrast, the AFT estimator performs poorly across all configurations. The estimated biases are substantial and do not diminish as the sample size increases, reflecting the sensitivity of the AFT specification to unmeasured confounding. Moreover, the estimated standard errors severely understate the sampling variability, leading to empirical coverage probabilities far below the nominal level. These results suggest that the AFT model does not yield reliable causal inference in the presence of many weak interaction instruments.

To check whether at least one constructed interaction term is correlated with the exposure in the observed sample, we conduct a heteroskedasticity-robust $F$ test \citep{white1980,mackinnon1985some}. Across all simulation designs, the heteroskedasticity-robust $F$ test yields $p$-values smaller than $0.001$, indicating strong evidence that the interaction component in the first stage is nontrivial and the interaction relevance. To validate Theorem \ref{thm:overid}, we conduct an overidentification test and report the empirical type I error based on 500 simulation replications. The results are presented in the supplementary materials.

Overall, the simulation results provide clear evidence that the proposed GEL estimators maintain accuracy and inferential validity across a wide range of challenging cases, while the standard AFT approach is unable to accommodate the complexities introduced by the unobserved confounding.

\section{Application to UK Biobank data}\label{sec8}
UK Biobank is a large population-based cohort comprising roughly half a million participants recruited across the United Kingdom between 2006 and 2010. Baseline assessments include anthropometric measurements, blood and urine biomarkers, and detailed lifestyle and medical questionnaires \citep{Sudlow2015}. Longitudinal follow-up is enabled through linkage to national hospital, primary care, and mortality registries, which provide information on diagnoses, procedures, and prescription records. These data allow the construction of prospective time-to-event outcomes from routinely collected electronic health records.

In this application, we restrict attention to individuals who were at least 60 years old at baseline, for whom age related cardiovascular and gastrointestinal conditions become increasingly common.  We further limit the sample to unrelated participants of white British ancestry and apply standard quality control procedures. These procedures exclude individuals with mismatched genetic and reported sex, excessive genotype missingness, or withdrawn consent. Variants with low imputation quality are also removed \citep{liu2023,ye2024}. We consider two right-censored outcomes derived from phenotype code (phecode) based on ICD-coded hospital and primary care records \citep{dey2022}. The first outcome is ischemic heart disease (phecode 411), and the second outcome is gastritis and duodenitis (phecode 535). Their sample sizes were 78995 and 80567, with 7155 and 6320 observed events, respectively. The corresponding censoring rates were 0.909 and 0.922. For each phecode, the failure time for uncensored individuals was defined as the logarithm of the age at disease onset minus the age at enrollment into the UK Biobank. For censored individuals, the failure time was defined as the logarithm of the last observed age from hospital admissions, general practice records, or death registries minus the age at enrollment.

The exposure of interest is body mass index (BMI). For instrument selection, we identify genetic variants associated with BMI using previously established genome-wide association results \citep{locke2015genetic}. Among single nucleotide polymorphisms that reach genome-wide significance at the threshold $p<5\times10^{-8}$, we apply linkage disequilibrium clumping to obtain an approximately independent set, using an $r^{2}$ threshold $0.001$ and distance based pruning consistent with prior Mendelian randomization studies. From the resulting pool, we retain the top $15$ single nucleotide polymorphisms with the strongest marginal associations with BMI as candidate instruments. As benchmark methods, we fit two AFT models for disease onset with BMI as the covariate of interest, denoted as AFT (without covariates) and AFT (with covariates), respectively, where the latter further includes age and sex as covariates.

\begin{table}[htbp]
\centering
\caption{Estimated causal effect of BMI on two survival outcomes, with exponentiated coefficients and standard errors. $p_{F}$ and $p_{\textrm{over-id}}$ represent the $p$ value for heteroskedasticity-robust $F$ test and overidentification test, respectively.}
\label{tab:ukbresults}
\begin{tabular}{ccccc}
\hline
\multicolumn{5}{c}{Gastritis and Duodenitis}                                                                                                                  \\ \hline
                          & \multicolumn{1}{c}{$\text{exp}(\hat\beta)$} & \multicolumn{1}{c}{SE} & $p_{F}$                           & $p_{\textrm{over-id}}$ \\\hline
CUE                       & 0.8418                                      & 0.0163                 & \multirow{3}{*}{\textless{}0.001} & 0.605                  \\
EL                        & 0.8414                                      & 0.0117                 &                                   & 0.652                  \\
ET                        & 0.8365                                      & 0.0156                 &                                   & 0.796                  \\
AFT  (without covariates) & 0.9960                                      & 0.0013                 & -                                 & -                      \\
AFT  (with covariates)    & 0.9957                                      & 0.0014                 & -                                 & -                      \\ \hline
\multicolumn{5}{c}{Ischemic Heart   Disease}                                                                                                                  \\ \hline
                          & \multicolumn{1}{c}{$\text{exp}(\hat\beta)$} & \multicolumn{1}{c}{SE} & $p_{F}$                           & $p_{\textrm{over-id}}$ \\ \hline
CUE                       & 0.8370                                      & 0.1490                 & \multirow{3}{*}{\textless{}0.001} & 0.303                  \\
EL                        & 0.8431                                      & 0.0169                 &                                   & 0.893                  \\
ET                        & 0.8041                                      & 0.0351                 &                                   & 0.690                  \\
AFT  (without covariates) & 0.9771                                      & 0.0014                 & -                                 & -                      \\
AFT  (with covariates)    & 0.9763                                      & 0.0014                 & -                                 & -                      \\ \hline
\end{tabular}
\end{table}
 
Table \ref{tab:ukbresults} reports the estimated causal effects of BMI on the two time-to-event outcomes. For interpretability, the estimates are presented on the exponentiated scale, and the corresponding standard errors are obtained via the delta method. For gastritis and duodenitis, all three GEL-based estimators (CUE, EL, ET) yield highly consistent results, with estimated effects ranging from exp($\widehat\beta$) = 0.836 to 0.842. These estimates imply that a one-unit increase in the exposure multiplies the time to onset by approximately 0.84, corresponding to an average 16\% shorter time to onset on the multiplicative scale. In contrast, the AFT estimator produces an effect estimate extremely close to one, suggesting almost no detectable association, which is likely attributable to residual confounding not adequately addressed by the standard AFT specification. We also conduct two diagnostic methods to test whether the identification assumption and the moment conditions are satisfied in UK Biobank data. Both test results support the validity of the model. In Table \ref{tab:ukbresults}, $p_{F}$ denotes the $p$-value of the heteroskedasticity-robust $F$ test \citep{white1980,mackinnon1985some}, which is used to assess the interaction relevance assumption. The test results strongly rejects the null hypothesis of no interaction effects, indicating that at least one constructed interaction term is correlated with the exposure. In addition, the overidentification test does not show systematic rejection, providing no evidence against the imposed moment conditions. 

For ischemic heart disease, the pattern is similar. The GEL estimators produce multiplicative effects between exp($\widehat\beta$) = 0.804 and 0.834, indicating that each one-unit increase in the exposure corresponding to an approximately 17\% shorter time to onset. Again, the AFT model yields an estimate near one, substantially attenuated relative to the GEL-based results. This attenuation suggests that the traditional AFT model fails to fully account for unmeasured confounding, leading to underestimation of the true causal effect. The close agreement among CUE, EL, and ET further demonstrates the robustness of our proposed AIPCW-GEL framework in capturing meaningful causal signals across both outcomes.

\section{Discussion}\label{sec9}

This paper studies causal inference for time-to-event outcomes under right censoring when many candidate instruments may be weak and invalid. We propose a generalized structural accelerated failure time model that generalizes the classical SAFT framework to accommodate potential weak and invalid IVs. Within this framework, we develop a censoring-adjusted moment construction that allows for right censoring. Identification is achieved by exploiting interaction-based moment restrictions, which remain informative even under weak instrument strength and invalidity. Estimation and inference are conducted using generalized empirical likelihood, enabling the combination of many moments and delivering stable performance in the many weak instrument settings. We establish the large-sample properties of the proposed estimator, and demonstrate its finite sample and empirical performance through simulation studies and an application to the UK Biobank.

Several directions merit further investigation. First, under weak moment conditions, the existence and characterization of a variance lower bound is still unclear and warrants further investigation. Second, an important direction is to extend the proposed framework to other survival models, such as structural Cox model \citep{wang2023} and semiparametric AFT model \citep{huling2019}. More broadly, it would be of interest to develop data-adaptive procedures for selecting or weighting interaction terms in very large genetic data, with finite sample guarantees in regimes where the number of candidate interactions grows rapidly with sample size.

\bibliography{ref.bib}

\newpage
\appendix
\setcounter{section}{0}
\setcounter{condition}{2}
\setcounter{equation}{0}
\renewcommand{\thelemma}{S\arabic{lemma}}  
\renewcommand{\thesection}{S\arabic{section}}
\renewcommand{\theequation}{S\arabic{equation}}
\renewcommand{\thetable}{S\arabic{table}}
\renewcommand{\thefigure}{S\arabic{figure}}
\begin{center}
  \Large\textbf{Supplementary material for `Identification and Inference for Structural Accelerated Failure Time Models via Instrument Interactions'}
\end{center}
\section{Proof of Theorem \ref{thmidentification}}
\begin{proof}
For the first term $\frac{\delta}{G_0(Y| \Z, D)}\{g(\beta,\bLambda_0;\O)-\xi_0(\beta;Y,\Z, D)\}$, since $Y=T$ when $\delta=1$ and the whole term is 0 when $\delta=0$, we have  
\begin{align*}
&\mathbb{E}\left[\frac{\delta}{G_0(Y| \Z, D)}\{g(\beta,\bLambda_0;\O)-\xi_0(\beta;Y,\Z, D)\} \Big|T,\Z, D\right]\\
=&\mathbb{E}\left[\frac{\delta}{G_0(T| \Z, D)}\{g(\beta,\bLambda_0;\O_T)-\xi_0(\beta;T,\Z, D)\} \Big|T,\Z, D\right]\\
=& \mathbb{E}  \left[\frac{I(C\ge T)}{G_0(T| \Z, D)} \Big| T,\Z, D\right]\{g(\beta,\bLambda_0;\O_T)-\xi_0(\beta;T,\Z, D)\}\\
=& g(\beta,\bLambda_0;\O_T)-\xi_0(\beta;T,\Z, D).
\end{align*}     
For the third term $\int_{-\infty}^{Y}\frac{d\xi_0(\beta;u,\Z, D)}{G_0(u| \Z, D)}$, we have
\begin{align}\label{eq:intxi}
&\mathbb{E}  \left[\int_{-\infty}^{Y}\frac{d\xi_0(\beta;u,\Z, D)}{G_0(u| \Z, D)} \Big| T,\Z, D\right]\n\\
=&\mathbb{E}  \left[\int_{-\infty}^{\infty}I(T\ge u)I(C\ge u)\frac{d\xi_0(\beta;u,\Z, D)}{G_0(u| \Z, D)} \Big| T,\Z, D\right]\n\\
=&\int_{-\infty}^{\infty}I(T\ge u) 
\mathbb{E}  \left[\frac{I(C\ge u)}{G_0(u| \Z, D)} \Big| T,\Z, D\right] d\xi_0(\beta;u,\Z, D)\n\\
=&\int_{-\infty}^{\infty}I(T\ge u) 
 d\xi_0(\beta;u,\Z, D)\n\\
=&\int_{-\infty}^{T} d\xi_0(\beta;u,\Z, D)\n\\
=&\xi_0(\beta;T,\Z, D)-\xi_0(\beta;-\infty,\Z, D).
\end{align}
Thus
\begin{align*}
 &\mathbb{E}\psi(\beta_0,\boldeta_0;\O)\\
=&  \mathbb{E}\left(\mathbb{E}\left[\psi(\beta_0,\boldeta_0;\O)\Big|T,\Z, D\right] \right)\\
=&\mathbb{E}\left[g(\beta_0,\bLambda_0;\O_T)-\xi_0(\beta_0;T,\Z, D)+\xi_0(\beta_0;-\infty,\Z, D)+\xi_0(\beta_0;T,\Z, D)-\xi_0(\beta_0;-\infty,\Z, D)    \right]\\
=&\mathbb{E}\left[g(\beta_0,\bLambda_0;\O_T) \right]\\
=&0.
\end{align*}

Then we prove the uniqueness of $\beta$. If there exists $\beta_1$ and $\beta_2$ both satisfy $\bbE\psi(\beta,\boldeta_0;\O)=0$, then we have 
\begin{align*}
&\bbE \psi_k(\beta_1,\boldeta_0;\O) - \bbE \psi_k(\beta_2,\boldeta_0;\O)\\
=&\bbE g_k(\beta_1,\bLambda_0;\O_T) - \bbE g_k(\beta_2,\bLambda_0;\O_T)\\
=& \bbE\left[\Int_k(\Z;\bzeta_0)(Y-V_k\vartheta_{k,0}-\beta_1(D-V_k\omega_{k,0}))  \right]-\bbE\left[\Int_k(\Z;\bzeta_0)(Y-V_k\vartheta_{k,0}-\beta_2(D-V_k\omega_{k,0}))  \right]\\
=& (\beta_1-\beta_2)\bbE\left[ \Int_k(\Z;\bzeta_0)(D-V_k\omega_{k,0}) \right]\\
=&0.
\end{align*}
Since $V_k$ includes only interaction terms up to order $k-1$, while $\Int_k(\Z;\bzeta_0)$ involves $k$th-order interactions, there exists at least one centered component $(Z_j-\bzeta_{j0})$ that appears in $\Int_k(\Z;\bzeta_0)$ but not in any element of $V_k$. Under Assumption \ref{condidentification}, this unmatched centered factor is independent of $V_k$ and has mean zero, which implies $\bbE\left[\Int_k(\Z;\bzeta_0)V_k\right]=0$. Putting interaction terms of order $2$ to $q$ together, given that $\bbE\left[ \Int_{2:q}(\Z;\bzeta_0)D\right]$ has at least one non-zero entry, we have $\bbE\left[ \Int_{2:q}(\Z;\bzeta_0)(D-V_k\omega_{k,0}) \right]=\bbE\left[ \Int_{2:q}(\Z;\bzeta_0)D\right]$ also has at least one non-zero entry, which means $\beta_1=\beta_2$.
\end{proof}

\section{Proof of Theorem \ref{thmneyman}}
\begin{proof}
(i). First, we prove $\psi$ is Neyman orthogonal to $\xi$. For an admissible perturbation direction $h_{\xi}$ 
\begin{align*}
 &\frac{\partial \psi(\beta_0,\xi_0+th_{\xi};\O)}{\partial t}\Big|_{t=0}\\      
=& \frac{\partial}{\partial t}\left[\frac{\delta}{G_0(Y| \Z, D)}\{g(\beta_0,\bLambda_0;\O)-\xi_0(\beta_0;Y,\Z, D)-th_{\xi}(\beta_0;Y,\Z, D)\}\right]\Big|_{t=0}\\
&+  \frac{\partial}{\partial t}\left[\xi_0(\beta_0;-\infty,\Z, D)+th_{\xi}(\beta_0;-\infty,\Z, D)
\right]\Big|_{t=0}\\
&+  \frac{\partial}{\partial t}\left[\int_{-\infty}^{Y}\frac{d\xi_0(\beta_0;u,\Z, D)+th_{\xi}(\beta_0;u,\Z, D)}{G_0(u| \Z, D)}\right]\Big|_{t=0}\\
=& -\frac{\delta}{G_0(Y| \Z, D)}h_{\xi}(\beta_0;Y,\Z, D)+h_{\xi}(\beta_0;-\infty,\Z, D)+\int_{-\infty}^{Y}\frac{dh_{\xi}(\beta_0;u,\Z, D)}{G_0(u| \Z, D)}.
\end{align*}
Taking expectation, we have
\begin{align*}
& \mathbb{E}\frac{\partial \psi(\beta_0,\xi_0+th_{\xi};\O)}{\partial t}\Big|_{t=0}\\
=&  \mathbb{E}\left(\mathbb{E}\left[-\frac{\delta}{G_0(Y| \Z, D)}h_{\xi}(\beta_0;Y,\Z, D)+h_{\xi}(\beta_0;-\infty,\Z, D)+\int_{-\infty}^{Y}\frac{dh_{\xi}(\beta_0;u,\Z, D)}{G_0(u| \Z, D)}\Big| T,\Z,D\right]\right)\\
=& \mathbb{E}\left(-h_{\xi}(\beta_0;T,\Z, D)+ h_{\xi}(\beta_0;-\infty,\Z, D)+ h_{\xi}(\beta_0;T,\Z, D)- h_{\xi}(\beta_0;-\infty,\Z, D)   \right)\\
=& 0,
\end{align*}
where the integration term is similar to \eqref{eq:intxi}. Thus $\psi$ is Neyman orthogonal to $\xi$. 

Next, we prove $\psi$ is Neyman orthogonal to $G$.
\begin{align*}
 &\frac{\partial \psi(\beta_0,G_0+th_{G};\O)}{\partial t}\Big|_{t=0}\\      
=& \frac{\partial}{\partial t}\left[\frac{\delta}{G_0(Y| \Z, D)+th_{G}(Y| \Z, D)}\{g(\beta_0,\bLambda_0;\O)-\xi_0(\beta_0;Y,\Z, D)\}\right]\Big|_{t=0}\\
&+  \frac{\partial}{\partial t}\left[\xi_0(\beta_0;-\infty,\Z, D)+\int_{-\infty}^{Y}\frac{d\xi_0(\beta_0;u,\Z, D)}{G_0(u| \Z, D)+th_{G}(u| \Z, D)}\right]\Big|_{t=0}\\
=&-\frac{\delta h_{G}(Y| \Z, D)}{G_0(Y| \Z, D)^2}\{g(\beta_0,\bLambda_0;\O)-\xi_0(\beta_0;Y,\Z, D)\}-\int_{-\infty}^{Y}\frac{h_{G}(u| \Z, D)d\xi_0(\beta_0;u,\Z, D)}{G_0(u| \Z, D)^2}.
\end{align*}
Taking expectation, we have
\begin{align}\label{eq:partialG}
& -\mathbb{E}\frac{\partial \psi(\beta_0,G_0+th_{G};\O)}{\partial t}\Big|_{t=0}\n\\
=&  \mathbb{E}\left(\mathbb{E}\left[\frac{\delta h_{G}(Y| \Z, D)}{G_0(Y| \Z, D)^2}\{g(\beta_0,\bLambda_0;\O)-\xi_0(\beta_0;Y,\Z, D)\}\Big| T,\Z,D\right]\right)\n\\
&+  \mathbb{E}\left(\mathbb{E}\left[\int_{-\infty}^{Y}\frac{h_{G}(u| \Z, D)d\xi_0(\beta_0;u,\Z, D)}{G_0(u| \Z, D)^2}\Big| T,\Z,D\right]\right)\n\\
=& \mathbb{E}\left(\frac{ h_{G}(T| \Z, D)}{G_0(T| \Z, D)}\{g(\beta_0,\bLambda_0;\O_T)-\xi_0(\beta_0;T,\Z, D)\}\right)\n\\
&+ \mathbb{E}\left(\int_{-\infty}^{T}\frac{h_{G}(u| \Z, D)}{G_0(u| \Z, D)}d\xi_0(\beta_0;u,\Z, D)\right)\n\\
=& \mathbb{E}\left(\frac{ h_{G}(T| \Z, D)}{G_0(T| \Z, D)}\{g(\beta_0,\bLambda_0;\O_T)-\xi_0(\beta_0;T,\Z, D)\}\right)\n\\
&+ \mathbb{E}\left( \frac{ h_{G}(T| \Z, D)}{G_0(T| \Z, D)}\xi_0(\beta_0;u,\Z, D)\Big|_{0}^{T}- \int_{-\infty}^{T}  \xi_0(\beta_0;u,\Z, D) d\frac{ h_{G}(u| \Z, D)}{G_0(u| \Z, D)}\right)\n\\
=&\mathbb{E}\left(\frac{ h_{G}(T| \Z, D)}{G_0(T| \Z, D)}g(\beta_0,\bLambda_0;\O_T)-\frac{ h_{G}(T| \Z, D)}{G_0(T| \Z, D)}\xi_0(\beta_0;T,\Z, D)\right)\n\\
&+\mathbb{E}\left( \frac{ h_{G}(T| \Z, D)}{G_0(T| \Z, D)}\xi_0(\beta_0;T,\Z, D)-\frac{ h_{G}(-\infty| \Z, D)}{G_0(-\infty| \Z, D)}\xi_0(\beta_0;-\infty,\Z, D)\right.\n\\
&-\left.\int_{-\infty}^{T}  \xi_0(\beta_0;u,\Z, D) d\frac{ h_{G}(u| \Z, D)}{G_0(u| \Z, D)}\right)\n\\
=&\mathbb{E}\left(\frac{ h_{G}(T| \Z, D)}{G_0(T| \Z, D)}g(\beta_0,\bLambda_0;\O_T)-\frac{ h_{G}(-\infty| \Z, D)}{G_0(-\infty| \Z, D)}\xi_0(\beta_0;-\infty,\Z, D)\right.\n\\
&-\left.\int_{-\infty}^{T}  \xi_0(\beta_0;u,\Z, D) d\frac{ h_{G}(u| \Z, D)}{G_0(u| \Z, D)}\right).
\end{align}
Let $F(t| \Z, D)=P(T\le t| \Z, D)$ and $S(t| \Z, D)=1-F(t| \Z, D)=P(T\ge t| \Z, D)$. The definition of $\xi$ implies that 
$$\xi_0(\beta;u,\Z, D)=\mathbb{E}[g(\beta_0,\bLambda_0;\O_T)|T\geq u,\Z,D]=\frac{\int_{u}^{\infty}g(\beta; t,\Z, D) dF(t| \Z, D)}{S(u| \Z, D)}.$$
This implies
$$
d[\xi_0(\beta;u,\Z, D)S(u| \Z, D)]=-g(\beta,\bLambda_0; u,\Z, D) dF(u| \Z, D).
$$
Note that $S(-\infty|\Z,D)=1,S(\infty|\Z,D)=0$, we have
\begin{align}\label{eq:EhG}
 &\mathbb{E}\left[\frac{ h_{G}(T| \Z, D)}{G_0(T| \Z, D)}g(\beta_0,\bLambda_0;T,\Z, D)\Big| \Z,D  \right] \n \\    
=& \int_{-\infty}^{\infty}\frac{ h_{G}(u| \Z, D)}{G_0(u| \Z, D)}g(\beta,\bLambda_0; u,\Z, D) dF(u| \Z, D)\n\\
=& -\int_{-\infty}^{\infty}\frac{ h_{G}(u| \Z, D)}{G_0(u| \Z, D)} d[\xi_0(\beta;u,\Z, D)S(u| \Z, D)]\n\\
=& -\frac{ h_{G}(u| \Z, D)}{G_0(u| \Z, D)}\xi_0(\beta;u,\Z, D)S(u| \Z, D)\Big|_{-\infty}^{\infty}+\int_{-\infty}^{\infty}\xi_0(\beta;u,\Z, D)S(u| \Z, D) d\frac{ h_{G}(u| \Z, D)}{G_0(u| \Z, D)}\n\\
=& \frac{ h_{G}(-\infty| \Z, D)}{G_0(-\infty| \Z, D)}\xi_0(\beta;-\infty,\Z, D)+\int_{-\infty}^{\infty}\xi_0(\beta;u,\Z, D)S(u| \Z, D) d\frac{ h_{G}(u| \Z, D)}{G_0(u| \Z, D)}.
\end{align}
On the other hand,
\begin{align}\label{eq:intxi0}
&\mathbb{E}\left[\int_{-\infty}^{T}  \xi_0(\beta_0;u,\Z, D) d\frac{ h_{G}(u| \Z, D)}{G_0(u| \Z, D)}\Big|\Z,D\right]\n\\
=& \mathbb{E}\left[\int_{-\infty}^{\infty} I(T\geq u) \xi_0(\beta_0;u,\Z, D) d\frac{ h_{G}(u| \Z, D)}{G_0(u| \Z, D)}\Big|\Z,D\right] \n \\
=& \int_{-\infty}^{\infty} \mathbb{E}\left[I(T\geq u)|\Z,D\right] \xi_0(\beta_0;u,\Z, D) d\frac{ h_{G}(u| \Z, D)}{G_0(u| \Z, D)} \n \\
=& \int_{-\infty}^{\infty}  \xi_0(\beta_0;u,\Z, D)S(u|\Z,D) d\frac{ h_{G}(u| \Z, D)}{G_0(u| \Z, D)}.
\end{align}
Then \eqref{eq:partialG} becomes
\begin{align*}
&-\mathbb{E}\frac{\partial \psi(\beta_0,G_0+th_{G};\O)}{\partial t}\Big|_{t=0}\\
=&\mathbb{E}\left(\frac{ h_{G}(T| \Z, D)}{G_0(T| \Z, D)}g(\beta_0,\bLambda_0;\O_T)-\frac{ h_{G}(-\infty| \Z, D)}{G_0(-\infty| \Z, D)}\xi_0(\beta_0;-\infty,\Z, D)\right.\n\\
&-\left.\int_{-\infty}^{T}  \xi_0(\beta_0;u,\Z, D) d\frac{ h_{G}(u| \Z, D)}{G_0(u| \Z, D)}\right)\\
=& \mathbb{E}\left(\mathbb{E}\left[\frac{ h_{G}(T| \Z, D)}{G_0(T| \Z, D)}g(\beta_0,\bLambda_0;\O_T)-\int_{-\infty}^{\infty}I(T\geq u)  \xi_0(\beta_0;u,\Z, D) d\frac{ h_{G}(u| \Z, D)}{G_0(u| \Z, D)}\Big| \Z,D\right]\right)\\
&- \mathbb{E}\left(\frac{ h_{G}(-\infty| \Z, D)}{G_0(-\infty| \Z, D)}\xi_0(\beta_0;-\infty,\Z, D)\right)\\
=& \mathbb{E}\left(\frac{ h_{G}(-\infty| \Z, D)}{G_0(-\infty| \Z, D)}\xi_0(\beta;-\infty,\Z, D)+\int_{-\infty}^{\infty}\xi_0(\beta;u,\Z, D)S(u| \Z, D) d\frac{ h_{G}(u| \Z, D)}{G_0(u| \Z, D)}\right)\\
=& -\mathbb{E}\left(\int_{-\infty}^{\infty}\xi_0(\beta;u,\Z, D)S(u| \Z, D) d\frac{ h_{G}(u| \Z, D)}{G_0(u| \Z, D)}-\frac{ h_{G}(-\infty| \Z, D)}{G_0(-\infty| \Z, D)}\xi_0(\beta;-\infty,\Z, D)  \right)\\
=& 0.
\end{align*}
Thus $\psi$ is Neyman orthogonal to $G$.

Finally, we prove $\psi$ is Neyman orthogonal to $\vartheta_{k}$, $\omega_{k}$, and $\bzeta$.  For $\vartheta_{k}$,
\begin{align*}
&\bbE\frac{\partial \psi(\beta_0,\boldeta_0;\O)}{\partial \vartheta_{k}}\Big|_{\vartheta_{k}=\vartheta_{k,0}}\\
=&\bbE\frac{\partial \psi(\beta_0,\boldeta_0;\O)}{\partial g(\beta_0,\bLambda_0;\O)}\frac{\partial g(\beta_0,\bLambda_0;\O)}{\partial \vartheta_{k}}\Big|_{\vartheta_{k}=\vartheta_{k,0}}   \\
=&-\bbE \left[\frac{\delta}{G_0(Y|\Z,D)}\Int_k(\Z;\bzeta_0)V_k\right]\\
=&-\bbE \left[\bbE\left(\frac{\delta}{G_0(Y|\Z,D)}|T,\Z,D\right)\Int_k(\Z;\bzeta_0)V_k\right]\\
=&-\bbE \left[\Int_k(\Z;\bzeta_0)(1,\Int_{1:k-1}(\Z)\tp)\right]\\
=&0.
\end{align*}
Similarly,
$$
\bbE\frac{\partial \psi(\beta_0,\boldeta_0;\O)}{\partial \omega_{k}}\Big|_{\omega_{k}=\omega_{k,0}}=\bbE \left[\frac{\delta}{G_0(Y|\Z,D)}\Int_k(\Z;\bzeta_0)V_k\beta\right] =0.
$$
Under structural model, we have
$$
\bbE(T|\Z) = \alpha_0+ \sum_{j=1}^{k-1} \Int_{j}(\Z)\tp\alpha_j +\sum_{j=k}^{p} \Int_{j}(\Z;\bzeta)\tp\alpha_j,
$$
and
$$
\bbE(D|\Z) = \iota_0+ \sum_{j=1}^{k-1}  \Int_{j}(\Z)\tp\iota_j+\sum_{j=k}^{p} \Int_{j}(\Z;\bzeta)\tp\iota_j.
$$
Here we have centered the interaction terms of order $k$ and higher. Correspondingly, in our setup, the nuisance component is $\vartheta_{k} = (\alpha_0,\dots,\alpha_{k-1})$ and $\omega_{k} = (\iota_0,\dots,\iota_{k-1})$. This implies
$$
\bbE\left[(T-V_k\vartheta_{k})-\beta(D-V_k\omega_{k})|\Z\right]=\sum_{j=k}^{p} \Int_{j}(\Z;\bzeta)\tp(\alpha_j-\beta\iota_j)
$$
For $\bzeta$, we have 
\begin{align*}
&\bbE\frac{\partial \psi_{k}(\beta_0,\boldeta_0;\O)}{\partial \zeta_{l}}\Big|_{\zeta_{l}=\zeta_{l,0}}\\
=&\bbE\frac{\partial \psi(\beta_0,\boldeta_0;\O)}{\partial g(\beta_0,\bLambda_0;\O)}\frac{\partial g(\beta_0,\bLambda_0;\O)}{\partial \zeta_{l}}\Big|_{\zeta_{l}=\zeta_{l,0}}   \\
=&\bbE \left[\frac{\delta}{G_0(Y|\Z,D)}\frac{\partial \Int_{k}(\Z;\bzeta)}{\partial \zeta_{l}}\{(Y-V_k\vartheta_{k})-\beta(D-V_k\omega_{k})\}\right]\\
=&\bbE \left[\bbE\left(\frac{\delta}{G_0(Y|\Z,D)}|T,\Z,D\right)\frac{\partial \Int_{k}(\Z;\bzeta)}{\partial \zeta_{l}}\bbE\{(T-V_k\vartheta_{k})-\beta(D-V_k\omega_{k})|\Z\}\right]\\
=&\bbE \left[\frac{\partial \Int_{k}(\Z;\bzeta)}{\partial \zeta_{l}}\sum_{j=k}^{p} \Int_{j}(\Z;\bzeta)\tp(\alpha_j-\beta\iota_j)\right]
\end{align*}
Taking the derivative of a centered $k$-th order interaction term $\Int_{k}(\Z;\bzeta)$ with respect to $\zeta_l$ yields either zero or a sum of centered $(k-1)$-th order interaction terms that do not involve $\zeta_l$. Note that the remaining term $\Int_{j}(\Z;\bzeta)$ involves centered interactions of order at least $k$. Because we assume mutual independence among the components of $\Z$, each higher-order term must contain at least one variable that does not appear in the derivative part. Consequently, the expectation of such a product is zero. Thus $\psi$ is Neyman orthogonal to all nuisances in $\boldeta$.

(ii). When $G_0$ is correctly specified, but $\xi$ is misspecified, we have
\begin{align*}
 & \mathbb{E}\left[\frac{\delta}{G_0(Y| \Z, D)}\{g(\beta_0,\bLambda_0; \O)-\xi(\beta_0;Y,\Z, D)\}+\xi(\beta_0;-\infty,\Z, D)
+\int_{-\infty}^{Y}\frac{d\xi(\beta_0;u,\Z, D)}{G_0(u| \Z, D)}\right]  \\
=& \mathbb{E}\left[\mathbb{E}\left(\frac{\delta}{G_0(Y| \Z, D)}\{g(\beta_0,\bLambda_0; \O)-\xi(\beta_0;Y,\Z, D)\}\Big| T,\Z,D\right)\right]   \\    
&+\mathbb{E}\left[\mathbb{E}\left(\xi(\beta_0;-\infty,\Z, D)
+\int_{-\infty}^{Y}\frac{d\xi(\beta_0;u,\Z, D)}{G_0(u| \Z, D)}\Big| T,\Z,D\right)\right] \\
=& \mathbb{E}\left[\mathbb{E}\left(\frac{\delta}{G_0(T| \Z, D)}\Big| T,\Z,D\right)\{g(\beta_0,\bLambda_0; \O_T)-\xi(\beta_0;T,\Z, D)\}\right]   \\    
&+\mathbb{E}\left[\xi(\beta_0;-\infty,\Z, D)
+\int_{-\infty}^{\infty}{I(T\geq u)}\mathbb{E}\left(\frac{I(C\geq u)}{G_0(u| \Z, D)}\Big| T,\Z,D\right)d\xi(\beta_0;u,\Z, D)\right] \\
=&\mathbb{E}\left[g(\beta_0,\bLambda_0; \O_T)-\xi(\beta_0;T,\Z, D)+\xi(\beta_0;-\infty,\Z, D)+\int_{-\infty}^{T}d\xi(\beta_0;u,\Z, D) \right]   \\
=&   \mathbb{E}\left[g(\beta_0,\bLambda_0; \O_T)\right]\\
=&0.
\end{align*} 
On the other hand, when $G$ is misspecified but $\xi_0$ is correctly specified, we have
\begin{align*}
 & \mathbb{E}\left[ \frac{\delta}{G(Y| \Z, D)}\{g(\beta_0,\bLambda_0; \O)-\xi_0(\beta_0;Y,\Z, D)\} +\xi_0(\beta_0;-\infty,\Z, D) +\int_{-\infty}^{Y}\frac{d\xi_0(\beta_0;u,\Z, D)}{G(u| \Z, D)} \right]  \\
=& \mathbb{E}\left[\mathbb{E}\left(\frac{\delta}{G(Y| \Z, D)}\{g(\beta_0,\bLambda_0; \O)-\xi_0(\beta_0;Y,\Z, D)\}\Big| T,\Z,D\right)\right]   \\
&+\mathbb{E}\left[\mathbb{E}\left(\xi_0(\beta_0;-\infty,\Z, D)+\int_{-\infty}^{Y}\frac{d\xi_0(\beta_0;u,\Z, D)}{G(u| \Z, D)}\Big| T,\Z,D\right)\right] \\
=& \mathbb{E}\left[\mathbb{E}\left(\frac{\delta}{G(Y| \Z, D)}\Big| T,\Z,D\right)\{g(\beta_0,\bLambda_0; \O)-\xi_0(\beta_0;Y,\Z, D)\}\right]   \\
&+\mathbb{E}\left[\xi_0(\beta_0;-\infty,\Z, D)+\int_{-\infty}^{\infty}I(T\geq u)\mathbb{E}\left(\frac{I(C\geq u)}{G(u| \Z, D)}\Big| T,\Z,D\right)d\xi_0(\beta_0;u,\Z, D)\right] \\
=& \mathbb{E}\left[\frac{G_0(T|\Z,D)}{G(T|\Z,D)}\{g(\beta_0,\bLambda_0; \O_T)-\xi_0(\beta_0;T,\Z, D)\}+\xi_0(\beta_0;-\infty,\Z, D)\right.\\
&+\left.\int_{-\infty}^{T}\frac{G_0(u|\Z,D)}{G(u|\Z,D)} d\xi_0(\beta_0;u,\Z, D)\right]\\ 
=&\mathbb{E}\left[\frac{G_0(T|\Z,D)}{G(T|\Z,D)}g(\beta_0,\bLambda_0; \O_T)-\frac{G_0(T|\Z,D)}{G(T|\Z,D)}\xi_0(\beta_0;T,\Z, D)+\xi_0(\beta_0;-\infty,\Z, D)\right]\\
&+\mathbb{E}\left[\xi_0(\beta_0;u,\Z, D)\frac{G_0(u|\Z,D)}{G(u|\Z,D)}\Big|_{0}^{T}-\int_{-\infty}^{T}\xi_0(\beta_0;u,\Z, D)d\frac{G_0(u|\Z,D)}{G(u|\Z,D)}\right]\\
=&\mathbb{E}\left[\frac{G_0(T|\Z,D)}{G(T|\Z,D)}g(\beta_0,\bLambda_0; \O_T)+\xi_0(\beta_0;-\infty,\Z, D)\right]\\
&+\mathbb{E}\left[-\xi_0(\beta_0;-\infty,\Z, D)\frac{G_0(-\infty|\Z,D)}{G(-\infty|\Z,D)}-\int_{-\infty}^{T}\xi_0(\beta_0;u,\Z, D)d\frac{G_0(u|\Z,D)}{G(u|\Z,D)}\right].
\end{align*}
Similar to \eqref{eq:EhG}, we have
\begin{align*}
& \mathbb{E}\left[\frac{G_0(T|\Z,D)}{G(T|\Z,D)}g(\beta_0,\bLambda_0; \O_T)\Big| \Z,D\right]   \\
=& \frac{G_0(-\infty|\Z,D)}{G(-\infty|\Z,D)}\xi_0(\beta_0;-\infty,\Z, D)+\int_{-\infty}^{\infty}\xi_0(\beta_0;u,\Z, D)S(u| \Z, D) d\frac{G_0(T|\Z,D)}{G(T|\Z,D)},
\end{align*}
and 
\begin{align*}
\mathbb{E}\left[\int_{-\infty}^{T}  \xi_0(\beta_0;u,\Z, D) d\frac{G_0(u|\Z,D)}{G(u|\Z,D)}\Big|\Z,D\right]     
=&\int_{-\infty}^{\infty}  \xi_0(\beta_0;u,\Z, D)S(u|\Z,D) d\frac{G_0(u|\Z,D)}{G(u|\Z,D)}.
\end{align*}
Then \eqref{eq:partialG} becomes
\begin{align*}
&\mathbb{E}\left[\frac{G_0(T|\Z,D)}{G(T|\Z,D)}g(\beta_0,\bLambda_0; \O_T)+\xi_0(\beta_0;-\infty,\Z, D)\right]\\
&+\mathbb{E}\left[-\xi_0(\beta_0;-\infty,\Z, D)\frac{G_0(-\infty|\Z,D)}{G(-\infty|\Z,D)}-\int_{-\infty}^{T}\xi_0(\beta_0;u,\Z, D)d\frac{G_0(u|\Z,D)}{G(u|\Z,D)}\right]\\
=&\mathbb{E}\left[\mathbb{E}\left(\frac{G_0(T|\Z,D)}{G(T|\Z,D)}g(\beta_0,\bLambda_0; \O_T)+\xi_0(\beta_0;-\infty,\Z, D)\Big| \Z,D\right)\right]\\
&+\mathbb{E}\left[\mathbb{E}\left(-\xi_0(\beta_0;-\infty,\Z, D)\frac{G_0(-\infty|\Z,D)}{G(-\infty|\Z,D)}-\int_{-\infty}^{T}\xi_0(\beta_0;u,\Z, D)d\frac{G_0(u|\Z,D)}{G(u|\Z,D)}\Big| \Z,D\right)\right]\\
=& \mathbb{E}\left[\frac{G_0(-\infty|\Z,D)}{G(-\infty|\Z,D)}\xi_0(\beta_0;-\infty,\Z, D)+\int_{-\infty}^{\infty}\xi_0(\beta_0;u,\Z, D)S(u| \Z, D) d\frac{G_0(u|\Z,D)}{G(u|\Z,D)}\right]\\
&+ \mathbb{E}\left[-\frac{G_0(-\infty|\Z,D)}{G(-\infty|\Z,D)}\xi_0(\beta_0;-\infty,\Z, D)-\int_{-\infty}^{\infty}  \xi_0(\beta_0;u,\Z, D)S(u|\Z,D) d\frac{G_0(u|\Z,D)}{G(u|\Z,D)}\right]\\
&+\mathbb{E}\left[\xi_0(\beta_0;-\infty,\Z, D)\right] \\
=&\mathbb{E}\left[\xi_0(\beta_0;-\infty,\Z, D)\right]\\
=&\mathbb{E}\left[\mathbb{E}\left(g(\beta_0,\bLambda_0; \O_T)|T\geq -\infty,\Z,D\right)\right]\\
=& 0.
\end{align*}
Thus, $\mathbb{E}\psi=0$ either $\xi_0$ or $G_0$ is correctly specified.

\end{proof}

\begin{remark} \label{remark1}
From Theorem \ref{thmneyman}, we can deduce that $\psi'$ is also Neyman orthogonal to $\boldeta$. 
$$
\psi_{k}'(\boldeta_0;\O)=\frac{\delta}{G_0(Y| \Z, D)}\{g_{k}'(\bLambda_0;\O)-\xi_{k,0}'(Y,\Z, D)\}+\xi_{k,0}'(-\infty,\Z, D)
+\int_{-\infty}^{Y}\frac{d\xi_{k,0}'(u,\Z, D)}{G_0(u| \Z, D)},
$$
where
$$
g_{k}'(\O) = -\Int_k(\Z;\bzeta)(D-V_k\omega_{k}).
$$
Because $g_k$ and $\xi_k$ is linear in $\beta$, $g_{k}'(\bLambda_0;\O)$ and $\xi_{k,0}'(0,\Z, D)$ do not rely on $\beta$. Also the nuisances in $\psi_{k}'(\boldeta_0;\O)$ do not contain $\vartheta_{k}$. For $\xi$ and $G$, the proof in Theorem \ref{thmneyman} does not rely on the specific form of $g$, and thus the same result holds if we replace $g$ with $g'$. For $\omega_{k}$ and $\zeta$, we have
$$
\bbE\frac{\partial \psi'(\boldeta_0;\O)}{\partial \omega_{k}}\Big|_{\omega_{k}=\omega_{k,0}}=\bbE \left[\frac{\delta}{G_0(Y|\Z,D)}\Int_k(\Z;\bzeta_0)V_k\right] =0,
$$
and
\begin{align*}
\bbE\frac{\partial \psi'(\boldeta_0;\O)}{\partial \zeta_{l}}\Big|_{\zeta_{l}=\zeta_{l,0}}=\bbE \left[\frac{\partial \Int_{j}(\Z;\bzeta)}{\partial \zeta_{l}}\sum_{j=k}^{p} \Int_{j}(\Z;\bzeta)\tp\alpha_j\right]=0.
\end{align*}
Therefore, $\psi'$ is also Neyman orthogonal to $\boldeta$.
\end{remark}

\section{Proof of Theorem \ref{thmnormality}}\label{suppsecpfthmnormality} 
We first introduce some useful notation. Let $\bar \psi(\beta,\boldeta) = 1/n \sum_{i=1}^n \psi(\beta,\boldeta;\O_i)$, $\bar \Omega(\beta,\boldeta) = 1/n \sum_{i=1}^n \psi(\beta,\boldeta;\O_i)\psi(\beta,\boldeta;\O_i)\tp$, $\Omega(\beta,\boldeta) = \bbE \psi(\beta,\boldeta;\O)\psi(\beta,\boldeta;\O)\tp$. Let $\bbP$ denote the true probability measure and $\bbP_n$ the empirical measure based on $n$ observations. Define $D_{\boldeta} f (\boldeta_0)$ to be the pathway derivative of $f$ at the direction $\boldeta_0$.

Before introducing lemmas and the proof of Theorems, we introduce some additional regularity conditions.

\begin{condition}[Bounded higher-order moments]\label{condmoment}
(i). $\bbE(T^8), \bbE(D^8)$, and $\Z$ are bounded;\\
(ii). $\{\bbE(\|\psi(\beta_0,\boldeta_0;\O)\|^4)+\bbE(\|\psi^{'}(\boldeta_0;\O)\|^4)\}m/n$=o(1);\\
(iii). There exists $ l>2$ such that $n^{1/ l}\{\bbE(\sup_{\beta\in \calB}\|\psi(\beta,\boldeta_0;\O)\|^ l)\}^{1/ l}(m+\mu_n)/\sqrt{n}\rightarrow 0$. 
\end{condition}
\begin{condition}[Bounded eigenvalue] \label{condeigen}
There exists a positive constant $c$ such that $ 1/c\leq \lambda_{\min}(\Omega(\beta,\boldeta_0)) < \lambda_{\max}(\Omega(\beta,\boldeta_0))\leq c $ for all $\beta\in\calB$, and $\lambda_{\max}\bbE(\psi^{'}(\boldeta_0;\O)\psi^{'}(\boldeta_0;\O)\tp)\leq c $.
\end{condition}

Condition~\ref{condmoment} (i) and Condition~\ref{condeigen} impose basic regularity conditions ensuring that the key observable quantities and matrices remain well behaved. In particular, they require bounded eighth moments for $T$ and $D$, bounded support for $\Z$, and uniformly bounded eigenvalues of the covariance matrix $\Omega(\beta,\boldeta_0)$. These conditions guarantee that the moment functions and associated weighting matrices do not degenerate as the sample size increases. Condition~\ref{condmoment} (ii) and (iii) place additional restrictions on the AIPCW moment function $\psi(\beta,\boldeta;\O)$ and its derivative. These constraints are especially relevant in settings where the number of moment functions $m$ increases with $n$ or when the interactions generate weak identifying information. They ensure that the empirical processes generated by the score and its derivative remain sufficiently well controlled so that the stochastic expansions used in establishing the asymptotic distribution remain valid. Similar types of moment and growth conditions are standard in the weak IV and many instrument GEL literature; see, for example, Assumption 6 in \citet{newey2009} and Assumptions 5 and 6 in \citet{ye2024}.

\begin{condition}[GEL function] \label{condgel}
(i). $\rho(\cdot)$ is concave and three times continuously differentiable; (ii). $\rho(0)=0$, and $\rho^{'}(0)=\rho^{''}(0)=-1$.
\end{condition}
Condition~\ref{condgel} (i) requires the GEL function $\rho(\cdot)$ to be concave and three-times continuously differentiable, which provides the smoothness needed for a third-order Taylor expansion of the criterion function. Condition~\ref{condgel} (ii) is a normalization condition that sets the scale and centering of the GEL objective at zero by imposing $\rho(0)=0$ and $\rho'(0)=\rho''(0)=-1$. These properties are standard in the GEL literature and are consistent with the regularity conditions imposed in \citet{newey2004}.

\begin{lemma}\label{lemma1a}
Under Conditions \ref{condweak}, there exists a positive constant $c>0$
such that for all $\beta\in\calB$,
$$
|\beta-\beta_0|\le c\,\frac{\sqrt{n}}{\mu_n}\|\bbE\psi(\beta,\boldeta_0;\O)\|.
$$
\end{lemma}

\begin{proof}
From the linear expansion of $\bbE\psi(\beta,\boldeta_0;\O)$ around $\beta_0$,
$$
\bbE\psi(\beta,\boldeta_0;\O)-\bbE\psi(\beta_0,\boldeta_0;\O)
=(\beta-\beta_0)\psi^*,
$$
Since $\bbE\psi(\beta_0,\boldeta_0;\O)=0$, we have $\bbE\psi(\beta,\boldeta_0;\O)=(\beta-\beta_0)\psi^*$.
Hence,
$$
\frac{\sqrt{n}}{\mu_n}\|\bbE\psi(\beta,\boldeta_0;\O)\|
=\frac{\sqrt{n}}{\mu_n}|\beta-\beta_0|\|\psi^*\|.
$$
Condition \ref{condweak} implies that $\sqrt{n}\|\psi^*\|/\mu_n$ is bounded away from
zero and infinity. Therefore, there exists $c>0$ such that
$$
|\beta-\beta_0|\le c\,\frac{\sqrt{n}}{\mu_n}\|\bbE\psi(\beta,\boldeta_0;\O)\|.
$$
\end{proof}


\begin{lemma}\label{lemma2}
Under Conditions \ref{condweak}-\ref{condeigen}, there exists a constant $c>0$
and a random variable $\hat M=O_p(1)$ such that for all $\beta',\beta\in\calB$,
$$
\frac{\sqrt{n}}{\mu_n}\|\bbE\psi(\beta',\boldeta_0;\O)-\bbE\psi(\beta,\boldeta_0;\O)\|
\le c|\beta'-\beta|,
$$
and
$$
\frac{\sqrt{n}}{\mu_n}\|\bar\psi(\beta',\boldeta_0)-\bar\psi(\beta,\boldeta_0)\|
\le \hat M|\beta'-\beta|.
$$
\end{lemma}

\begin{proof}
By the linearity of $\bbE\psi(\beta,\boldeta_0;\O)$ in $\beta$, we have
$$
\bbE\psi(\beta',\boldeta_0;\O)-\bbE\psi(\beta,\boldeta_0;\O)=(\beta'-\beta)\psi^*,
$$
which yields
$$
\frac{\sqrt{n}}{\mu_n}\|\bbE\psi(\beta',\boldeta_0;\O)-\bbE\psi(\beta,\boldeta_0;\O)\|
=\frac{\sqrt{n}}{\mu_n}|\beta'-\beta|\|\psi^*\|
\le c|\beta'-\beta|.
$$
Moreover,
\begin{align*}
\bar\psi(\beta',\boldeta_0;\O)-\bar\psi(\beta,\boldeta_0;\O)
&=\{\bbE\psi(\beta',\boldeta_0;\O)-\bbE\psi(\beta,\boldeta_0;\O)\}\\
&\quad+\{\bar\psi(\beta',\boldeta_0)-\bbE\psi(\beta',\boldeta_0;\O)\}
-\{\bar\psi(\beta,\boldeta_0)-\bbE\psi(\beta,\boldeta_0;\O)\}.
\end{align*}
Hence,
$$
\frac{\sqrt{n}}{\mu_n}\|\bar\psi(\beta',\boldeta_0)-\bar\psi(\beta,\boldeta_0)\|
\le c|\beta'-\beta|
+\frac{\sqrt{n}}{\mu_n}\sup_{\beta\in\calB}
\|\bar\psi(\beta,\boldeta_0)-\bbE\psi(\beta,\boldeta_0;\O)\|.
$$
By a uniform law of large numbers and Conditions \ref{condweak}--\ref{condeigen},
$$
\sup_{\beta\in\calB}\|\bar\psi(\beta,\boldeta_0)-\bbE\psi(\beta,\boldeta_0;\O)\|
=O_p\!\left(\sqrt{\frac{m}{n}}\right),
$$
and since $m/\mu_n^2$ is bounded,
$$
\frac{\sqrt{n}}{\mu_n}\sup_{\beta\in\calB}
\|\bar\psi(\beta,\boldeta_0)-\bbE\psi(\beta,\boldeta_0;\O)\|=O_p(1).
$$
Therefore, there exists a $\hat M=O_p(1)$ such that for all $\beta',\beta\in\calB$,
$$
\frac{\sqrt{n}}{\mu_n}\|\bar\psi(\beta',\boldeta_0)-\bar\psi(\beta,\boldeta_0)\|
\le \hat M|\beta'-\beta|,
$$
which completes the proof.
\end{proof}


\begin{lemma}\label{lemma1c}
Under Condition \ref{condmoment}, there exists a positive constant $c$ such that
for all $a,b\in\mathbb R^m$ and $\beta',\beta\in\calB$,
$$
|a\tp\{\Omega(\beta',\boldeta_0)-\Omega(\beta,\boldeta_0)\}b|
\le c\|a\|\|b\||\beta'-\beta|.
$$
\end{lemma}

\begin{proof}
By the linearity of $\psi(\beta,\boldeta_0;\O)$ in $\beta$,
$$
\psi(\beta',\boldeta_0;\O)
=\psi(\beta,\boldeta_0;\O)+(\beta'-\beta)\psi'(\boldeta_0;\O),
$$
Then,
\begin{align*}
\Omega(\beta',\boldeta_0)-\Omega(\beta,\boldeta_0)
&=\bbE\{\psi(\beta',\boldeta_0)\psi(\beta',\boldeta_0)\tp
-\psi(\beta,\boldeta_0)\psi(\beta,\boldeta_0)\tp\}\\
&=(\beta'-\beta)^2\bbE\{\psi'(\boldeta_0)\psi'(\boldeta_0)\tp\}\\
&\quad+(\beta'-\beta)\bbE\{\psi'(\boldeta_0)\psi(\beta,\boldeta_0)\tp\}
+(\beta'-\beta)\bbE\{\psi(\beta,\boldeta_0)\psi'(\boldeta_0)\tp\}.
\end{align*}
Hence, for any $a,b\in\mathbb R^m$,
\begin{align*}
|a\tp\{\Omega(\beta',\boldeta_0)-\Omega(\beta,\boldeta_0)\}b|
&\le|\beta'-\beta|^2\|a\|\|\bbE(\psi'(\beta_0,\boldeta_0)\psi'(\beta_0,\boldeta_0)\tp)\|\|b\|\\
&\quad+2|\beta'-\beta|\|a\|\|\bbE(\psi'(\beta_0,\boldeta_0)\psi(\beta_0,\boldeta_0)\tp)\|\|b\|.
\end{align*}
Condition \ref{condmoment} ensures the matrix $\bbE(\psi'(\beta_0,\boldeta_0)\psi(\beta_0,\boldeta_0)\tp)$ have uniformly bounded spectral norms, and the parameter space $\calB$ of $\beta$ is a compact set, so there exists a constant $c>0$ such that  $|a\tp\{\Omega(\beta',\boldeta_0)-\Omega(\beta,\boldeta_0)\}b|\le c\|a\|\|b\||\beta'-\beta|.$
\end{proof}




\begin{lemma}\label{lemmapsi}
Under Condition \ref{condweak}-\ref{condeigen},
(i). $\left\|\bar{\psi}\left(\beta_0, \hat{\boldeta}\right)-\bar{\psi}\left(\beta_0, \boldeta_0\right)\right\|=o_p(n^{-1/2})$;

(ii). $\left\|\bar{\psi'}(\hat{\boldeta})-\bar{\psi'}\left(\boldeta_0\right)\right\|=o_p(n^{-1/2})$;

(iii). $\sup _{\beta \in \calB}\left\|\bar{\psi}(\hat{\boldeta})-\bar{\psi}\left(\beta,\boldeta_0\right)\right\|=o_p(n^{-1/2})$;   

(iv). $\sup_{\beta\in\calB}\|\bar \psi(\beta,\hat\boldeta)\|=o_p(\mu_n/\sqrt{n})$, $\|\bar \psi(\beta_0,\hat\boldeta)\|=o_p(\sqrt{m/n})$.   
\end{lemma}   
\begin{proof}
(i). The moment function $\psi$ is a $m$-dimensional vector. We begin by studying each entry $\psi_j$ and then turn to the analysis of the entire vector. We first show that $\psi$ is Lipschitz continuous in $\boldeta$. 
Since $\psi$ is linear in $\bLambda$ and $\xi$, $\psi$ is Lipschitz continuous in these components. 
For the component $G$, note that for any two functions $G_1$ and $G_2$,
\begin{align*}
&\Big[\frac{\delta}{G_1}\{g-\xi_0(Y|\Z,D)\}+\xi_0(-\infty|\Z,D)
+\int_{-\infty}^Y\frac{d\xi_0(u|\Z,D)}{G_1}\Big]\\
&-\Big[\frac{\delta}{G_2}\{g-\xi_0(Y|\Z,D)\}+\xi_0(-\infty|\Z,D)
+\int_{-\infty}^Y\frac{d\xi_0(u|\Z,D)}{G_2}\Big]\\
=&\frac{\delta\{g-\xi_0(Y|\Z,D)\}(G_1-G_2)}{G_1G_2}
+\int_{-\infty}^Y\frac{d\xi_0(u|\Z,D)}{G_1G_2}(G_1-G_2).
\end{align*}
Because $G$ is uniformly bounded away from zero from Assumption \ref{condposint}, both terms on the right-hand side are bounded by a constant multiple of $\|G_1-G_2\|$, which implies that $\psi$ is Lipschitz continuous in $G$. Combining these results, $\psi$ is Lipschitz continuous in $\boldeta=(\bLambda,\xi,G)$, and the Lipschitz coefficient is $\bbE L(\O) = (1+|Y|+|D|+\|h(\Z)\|)=O(\sqrt{m})$. Since we employ sample splitting so that the sample used to construct the moment functions is independent of the sample used to estimate the nuisance functions. Then we have
\begin{align*}
& \bbP\left[ \sqrt{nm}(\bbP_n-\bbP)(\psi_j(\beta_0,\hat\boldeta;\O)-\psi_j(\beta_0,\boldeta_0;\O))  >\epsilon \right]  \\
& \leq m \frac{\var(\psi_j(\beta_0,\hat\boldeta;\O)-\psi_j(\beta_0,\boldeta_0;\O))}{\epsilon^2}\\
& \leq mL(\O)^2\|\hat\boldeta-\boldeta_0\|^2\\
& = \sqrt{m^3/n} \\
&= o(1),
\end{align*}
which implies
\begin{equation}\label{eq:pn-p}
(\bbP_n-\bbP)(\psi_j(\beta_0,\hat\boldeta;\O)-\psi_j(\beta_0,\boldeta_0;\O))=o_p(n^{-1/2}m^{-1/2}).   
\end{equation}
Furthermore,
\begin{align*}
&\bbP(\psi_{j}\left(\beta_0, \hat{\boldeta};\O\right)-\psi_{j}\left(\beta_0, \boldeta_0;\O\right))\\
&= \bbE D_{\boldeta}\psi_{j}\left(\beta_0, \boldeta_0;\O\right)( \hat\boldeta- \boldeta_0)+\bbE D_{\boldeta\boldeta}\psi_{j}\left(\beta_0, \boldeta_0;\O\right)( \hat\boldeta- \boldeta_0)^2+o_p(\|\hat\boldeta- \boldeta_0\|^2).   
\end{align*}
We have shown that the nuisances in the moment function $\psi$ are all Neyman orthogonal, meaning that the expectation of the first-order derivative of $\psi$ with respect to each nuisance parameter is zero, i.e., $\bbE D_{\boldeta}\psi_j\left(\beta_0, \boldeta_0;\O\right)( \hat\boldeta- \boldeta_0)=0$. Therefore, we have
$$
\bbP(\psi_{j}\left(\beta_0, \hat{\boldeta};\O\right)-\psi_{j}\left(\beta_0, \boldeta_0;\O\right))=O_p(\|\hat\boldeta- \boldeta_0\|^2)=o_p(n^{-1/2}m^{-1/2}).
$$
Together with \eqref{eq:pn-p}, we have
$$
\bar \psi_j\left(\beta_0, \hat{\boldeta}\right)-\bar \psi_j\left(\beta_0, {\boldeta_0}\right)=(\bbP_n-\bbP+\bbP)(\psi_j\left(\beta_0, \hat{\boldeta}\right)-\psi_j\left(\beta_0, {\boldeta_0}\right))=o_p(n^{-1/2}m^{-1/2}).
$$
Since $\psi$ consists of $m$ moment functions, it follows that $\|\bar{\psi}(\beta_0, \hat{\boldeta})-\bar{\psi}(\beta_0, \boldeta_0)\|=o_p(n^{-1/2})$.

(ii). Note that $\psi'$ is also Neyman orthogonal to $\boldeta$ according to Remark \ref{remark1}, so we can apply the same argument to prove $\left\|\bar{\psi'}(\hat{\boldeta})-\bar{\psi'}\left(\boldeta_0\right)\right\|=o_p(n^{-1/2})$ as that in part (i), with $\psi$ replaced by $\psi'$.

(iii). Since $\psi$ is linear in $\beta$, we have
\bse 
&&\sup_{\beta\in \calB}\left\|\bar{\psi}\left(\beta, \hat{\boldeta}\right)-\bar{\psi}\left(\beta, \boldeta_0\right)\right\|\\
&=& \sup_{\beta\in \calB}\left\|\bar{\psi}\left(\beta, \hat{\boldeta}\right)-\bar{\psi}\left(\beta_0, \hat{\boldeta}\right)+\bar{\psi}\left(\beta_0, \hat{\boldeta}\right)-\bar{\psi}\left(\beta_0, {\boldeta}_0\right)+\bar{\psi}\left(\beta_0, {\boldeta}_0\right)-\bar{\psi}\left(\beta, {\boldeta}_0\right) \right\|\\
&\leq&\sup_{\beta\in \calB}\left\|(\bar{\psi'}(\hat\boldeta)-\bar{\psi'}(\boldeta_0))(\beta-\beta_0)\right\|  + \left\|\bar{\psi}\left(\beta_0, \hat{\boldeta}\right)-\bar{\psi}\left(\beta_0, {\boldeta}_0\right) \right\|\\
&=&o_p\left(n^{-1 / 2}\right).
\ese

(iv). By Condition \ref{condeigen}, we have
$$
\frac{n\,\bbE\{\|\bar\psi(\beta_0,\boldeta_0)\|^2\}}{m}
\le \frac{\tr(\Omega_0)}{m}
\le c.
$$
By Markov’s inequality, it follows that
$$
\|\bar\psi(\beta_0,\boldeta_0)\|=O_p\!\left(\sqrt{\frac{m}{n}}\right).
$$
Furthermore, by Lemma \ref{lemma2} and Lemma \ref{lemmapsi} (iii), we have
\bse
&&\sup_{\beta\in\calB}\|\bar\psi(\beta,\hat\boldeta)\|\\
&=&\sup_{\beta\in\calB}\|\bar\psi(\beta,\hat\boldeta)
-\bar\psi(\beta,\boldeta_0)
+\bar\psi(\beta,\boldeta_0)
-\bar\psi(\beta_0,\boldeta_0)
+\bar\psi(\beta_0,\boldeta_0)\|\\
&\le&
\sup_{\beta\in\calB}\|\bar\psi(\beta,\hat\boldeta)
-\bar\psi(\beta,\boldeta_0)\|
+\sup_{\beta\in\calB}\|\bar\psi(\beta,\boldeta_0)
-\bar\psi(\beta_0,\boldeta_0)\|
+\|\bar\psi(\beta_0,\boldeta_0)\|\\
&\le&
o_p(n^{-1/2})
+O_p\!\left({\mu_n}/{\sqrt{n}}\right)
+O_p\!\left(\sqrt{{m}/{n}}\right)\\
&=&
O_p\!\left({\mu_n}/{\sqrt{n}}\right).
\ese
For $\|\bar\psi(\beta_0,\hat\boldeta)\|$, the second term in the above
decomposition, $\sup_{\beta\in\calB}\|\bar\psi(\beta,\boldeta_0)
-\bar\psi(\beta_0,\boldeta_0)\|$, is zero, we can similarly obtain
\bse
&&\|\bar\psi(\beta_0,\hat\boldeta)\|\\
&=&\|\bar\psi(\beta_0,\hat\boldeta)
-\bar\psi(\beta_0,\boldeta_0)
+\bar\psi(\beta_0,\boldeta_0)\|\\
&\le&
\|\bar\psi(\beta_0,\hat\boldeta)
-\bar\psi(\beta_0,\boldeta_0)\|
+\|\bar\psi(\beta_0,\boldeta_0)\|\\
&\le&
o_p(n^{-1/2})
+O_p\!\left(\sqrt{{m}/{n}}\right)\\
&=&
O_p\!\left(\sqrt{{m}/{n}}\right).
\ese
This completes the proof of Lemma \ref{lemmapsi}.

\end{proof}

\begin{lemma}\label{lemmaOmega}
Under Conditions \ref{condweak}-\ref{condeigen},
(i). $\left\|\bar{\Omega}\left(\beta_0,\hat{\boldeta}\right)-\Omega_0\right\|=o_p(1/\sqrt{m})$;

(ii). $\left\|n^{-1} \sum_{i=1}^n \psi(\beta_0, \hat{\boldeta};\O_i) \psi'(\hat{\boldeta};\O_i)\tp-\bbE\left\{\psi(\beta_0,\boldeta_0;\O) \psi'(\boldeta_0;\O)\tp\right\}\right\|=o_p(1/\sqrt{m})$;

(iii). $\left\|n^{-1} \sum_{i=1}^n \psi'(\hat{\boldeta};\O_i) \psi'(\hat{\boldeta};\O_i)\tp-\bbE\left\{\psi'(\boldeta_0;\O) \psi'(\boldeta_0;\O)\tp\right\}\right\|=
o_p(1/\sqrt{m})$;

(iv). $\sup _{\beta \in \calB}\left\|\bar{\Omega}(\beta, \hat{\boldeta})-\Omega\left(\beta, \boldeta_0\right)\right\|=o_p(1/\sqrt{m})$;

(v). $\sup _{\beta \in \calB}\left\|n^{-1} \sum_{i=1}^n \psi(\beta, \hat{\boldeta};\O_i)\psi'(\hat{\boldeta};\O_i)\tp-\bbE\left\{\psi(\beta, \boldeta_0;\O) \psi'(\boldeta_0;\O)\tp\right\}\right\|=o_p(1/\sqrt{m})$.       
\end{lemma} 

\begin{proof}
(i). First, we have $\|\bar{\Omega}\left(\beta_0,\hat{\boldeta}\right)-\Omega_0\|\leq \|\bar{\Omega}\left(\beta_0,\hat{\boldeta}\right)-\bar{\Omega}\left(\beta_0,{\boldeta_0}\right)\|+\|\bar{\Omega}\left(\beta_0,{\boldeta_0}\right)- \Omega_0\|$.
The matrix concentration inequality in \cite{tropp2016} gives that 
$$
\|\bar{\Omega}\left(\beta_0,{\boldeta_0}\right)- \Omega_0\|=O_p(\sqrt{m\log(m)/n})=o_p(1/\sqrt{m})
$$
under $m^3/n\to 0$. Similarly, we have
$$
\|\frac{1}{n}\sum_{i=1}^n D_{\boldeta} \psi(\beta_0,{\boldeta_0};\O_i) D_{\boldeta} \psi(\beta_0,\boldeta_0;\O_i)\tp-\bbE D_{\boldeta}\psi(\beta_0,{\boldeta_0};\O_i) D_{\boldeta} \psi(\beta_0,\boldeta_0;\O_i)\tp \|=o_p(1/\sqrt{m})
$$
and 
$$
\|\frac{1}{n}\sum_{i=1}^n  \psi(\beta_0,{\boldeta_0};\O_i)  \psi(\beta_0,\boldeta_0;\O_i)\tp-\bbE\psi(\beta_0,{\boldeta_0};\O_i)\psi(\beta_0,\boldeta_0;\O_i)\tp \|=o_p(1/\sqrt{m}),
$$
Then Taylor's expansion gives that 
\begin{align*}
&\|\bar{\Omega}\left(\beta_0,\hat{\boldeta}\right)-\bar{\Omega}\left(\beta_0,{\boldeta_0}\right)\|\\
=&\|D_{\boldeta}\bar{\Omega}\left(\beta_0,{\boldeta_0}\right)\|\|\hat\boldeta-\boldeta_0\|+ o_p(\|\hat\boldeta-\boldeta_0\|)\\
=&2\|\frac{1}{n}\sum_{i=1}^n D_{\boldeta} \psi(\beta_0,{\boldeta_0};\O_i) \psi(\beta_0,\boldeta_0;\O_i)\tp\|\|\hat\boldeta-\boldeta_0\|+o_p(\|\hat\boldeta-\boldeta_0\|)\\
\leq & 2\|\frac{1}{n}\sum_{i=1}^n D_{\boldeta} \psi(\beta_0,{\boldeta_0};\O_i) D_{\boldeta} \psi(\beta_0,\boldeta_0;\O_i)\tp\|^{\frac{1}{2}}\|\frac{1}{n}\sum_{i=1}^n \psi(\beta_0,{\boldeta_0};\O_i) \psi(\beta_0,\boldeta_0;\O_i)\tp\|^{\frac{1}{2}}\|\hat\boldeta-\boldeta_0\|\\
& + o_p(\|\hat\boldeta-\boldeta_0\|)\\
\leq & 2\|\bbE D_{\boldeta} \psi(\beta_0,{\boldeta_0};\O_i) D_{\boldeta} \psi(\beta_0,\boldeta_0;\O_i)\tp\|^{\frac{1}{2}}\|\bbE \psi(\beta_0,{\boldeta_0};\O_i) \psi(\beta_0,\boldeta_0;\O_i)\tp\|^{\frac{1}{2}}\|\hat\boldeta-\boldeta_0\|\\
& + o_p(\|\hat\boldeta-\boldeta_0\|)\\
=& O_p(1)O_p(1)o_p((nm)^{-1/4})\\
=&o_p(1/\sqrt{m}),
\end{align*}
where the last step comes from $m^3/n\to 0$.
 So we obtain $\left\|\bar{\Omega}\left(\beta_0,\hat{\boldeta}\right)-\Omega_0\right\|=o_p(1/\sqrt{m})$.

(ii). The proof follows the same steps as in part (i), the only difference is changing $\psi$ into $\psi'$, and the details are omitted here.

(iv). According to Lemma \ref{lemmapsi}, Lemma \ref{lemmaOmega} and Taylor expansion, we have
\be 
&&\sup_{\beta \in \calB}\left\|\bar{\Omega}(\beta, \hat{\boldeta})-\Omega_0\left(\beta, \boldeta_0\right)\right\|\n\\
&=&\sup_{\beta \in \calB}\left\|\bar{\Omega}(\beta, \hat{\boldeta})-\bar\Omega(\beta_0,\hat\boldeta)+\bar\Omega(\beta_0,\hat\boldeta)-\Omega(\beta_0,\boldeta_0)+\Omega(\beta_0,\boldeta_0)-\Omega\left(\beta, \boldeta_0\right)\right\|\n\\
&\leq& \sup_{\beta \in \calB}\Big\|\frac{2}{n}\sum_{i=1}^n\psi(\beta_0,\hat\boldeta;\O_i)\psi'(\hat\boldeta;\O_i)\tp(\beta-\beta_0)+ \frac{1}{n}\sum_{i=1}^n\psi'(\hat\boldeta;\O_i)\psi'(\hat\boldeta;\O_i)\tp(\beta-\beta_0)^2\n\\
&&-2\bbE \{\psi(\beta_0,\boldeta_0;\O_i)\psi'(\boldeta_0;\O_i)\tp\}(\beta-\beta_0)- \bbE\{ \psi'(\boldeta_0;\O_i)\psi'(\boldeta_0;\O_i)\tp\}(\beta-\beta_0)^2 \Big\|\n\\
&&+\|\bar\Omega(\beta_0,\hat\boldeta)-\Omega(\beta_0,\boldeta_0)\|\n\\
&\leq& 2\sup_{\beta \in \calB}|\beta-\beta_0|\left\|\frac{1}{n}\sum_{i=1}^n\psi(\beta_0,\hat\boldeta;\O_i)\psi'(\hat\boldeta;\O_i)\tp-\bbE\{ \psi(\beta_0,\boldeta_0;\O_i)\psi'(\boldeta_0;\O_i)\tp\}\right\|\n\\
&&+\sup_{\beta \in \calB}|\beta-\beta_0|^2\left\|\frac{1}{n}\sum_{i=1}^n\psi'(\boldeta_0;\O_i)\psi'(\hat\boldeta;\O_i)\tp-\bbE \{\psi'(\boldeta_0;\O_i)\psi'(\boldeta_0;\O_i)\tp\}\right\| \n\\
&&+\|\bar\Omega(\beta_0,\hat\boldeta)-\Omega(\beta_0,\boldeta_0)\|\n\\
&=&o_p(1/\sqrt{m}).\n
\ee 
A useful implication of this conclusion here is $\sup_{\beta\in\calB}\|\bar \Omega(\beta,\hat\boldeta)\|\leq\sup_{\beta\in\calB}\|\bar \Omega(\beta,\hat\boldeta)-\Omega(\beta,\boldeta_0)\|+\sup_{\beta\in\calB}\|\Omega(\beta,\boldeta_0)\|=O_p(1) $.

(v). The proof is similar to the proof of (iv) and is omitted here.

\end{proof}

Define $P(\blambda,\beta,\boldeta)=1/n\sum_{i=1}^n \rho(\blambda\tp \psi(\beta,\boldeta;\O_i))$, $\blambda(\beta,\boldeta)=\argmax_{\blambda \in {L}(\beta,\boldeta)} P(\blambda,\beta,\boldeta)$
\begin{lemma}\label{lemmalambda}
Under Conditions \ref{condweak}-\ref{condgel},
(i) $\sup_{\beta\in\mathcal B}\|\blambda(\beta,\hat\boldeta)\|=O_p(\mu_n/\sqrt n)$, $\sup_{\beta\in\mathcal B}\|\blambda(\beta,\boldeta_0)\|=O_p(\mu_n/\sqrt n)$, and $\|\blambda(\beta_0,\boldeta_0)\|=O_p(\sqrt{m/n})$.

(ii) $\|\blambda(\beta_0,\hat\boldeta)-\blambda(\beta_0,\boldeta_0)\|=o_p(n^{-1/2})$, and $\sup_{\beta\in\mathcal B}\|\blambda(\beta,\hat\boldeta)-\blambda(\beta,\boldeta_0)\|=o_p(\mu_n/\sqrt n)$.
\end{lemma}

\begin{proof}
(i) Let $\kappa=\max_{i\le n}\sup_{\beta\in\mathcal B}\|\psi(\beta,\boldeta_0;\O_i)\|$. Standard maximal inequalities yield
$$
\kappa=O_p  \left(n^{1/\iota}  \mathbb E\big\{\sup_{\beta\in\mathcal B}\|\psi(\beta,\boldeta_0;\O)\|^{\iota}\big\}^{1/\iota}\right)
\quad \text{for any } \iota>2 .
$$
By a mean value expansion in $\boldeta$,
$$
\mathbb E\big\{\sup_{\beta\in \calB}\|\psi(\beta,\hat\boldeta;\O)\|^{\iota}\big\}^{1/\iota}
=
\mathbb E\big\{\sup_{\beta\in \calB}\|\psi(\beta,\boldeta_0;\O)\|^{\iota}\big\}^{1/\iota}
+O_p  \big(\|\hat\boldeta-\boldeta_0\|\big),
$$
hence with $\iota>2$, $\|\hat\boldeta-\boldeta_0\|=o_p(n^{-1/4}m^{-1/4})$ and Condition \ref{condmoment},
\begin{align*}
&n^{1/\iota} \mathbb E\big\{\sup_{\beta\in \calB}\|\psi(\beta,\hat\boldeta;\O)\|^{\iota}\big\}^{1/\iota} \frac{m+\mu_n}{\sqrt n}\\
&n^{1/\iota} \mathbb E\big\{\sup_{\beta\in \calB}\|\psi(\beta,\boldeta_0;\O)\|^{\iota}\big\}^{1/\iota} \frac{m+\mu_n}{\sqrt n}+n^{1/\iota}\|\hat\boldeta-\boldeta_0\|\frac{m+\mu_n}{\sqrt n}\\
&=o_p(1).   
\end{align*}

We can choose $\kappa_n$ such that $\kappa_n=o(\kappa^{-1})$ and $\mu_n/\sqrt n=o(\kappa_n)$, and set $L_n=\{\blambda:\|\blambda\|\le \kappa_n\}$. For any $\blambda\in L_n$, we have
\begin{equation}\label{eq:sup-lin}
\sup_{i\le n,\ \beta\in\mathcal B}\big|\blambda^{\top}\psi(\beta,\hat\boldeta;\O_i)\big| \le \|\blambda\|\psi(\beta,\hat\boldeta;\O_i) \|\|
\le \kappa_n\kappa=o_p(1).
\end{equation}
A Taylor expansion of $\rho$ at zero gives that for each $i$,
$$
\rho  \left(\blambda^{\top}\psi(\beta,\hat\boldeta;\O_i)\right)
= - \blambda^{\top}\psi(\beta,\hat\boldeta;\O_i)
+\tfrac12 \rho''(\bar x_i) \blambda^{\top}\psi(\beta,\hat\boldeta;\O_i)\psi(\beta,\hat\boldeta;\O_i)^{\top}\blambda,
$$
with $\bar x_i$ between $0$ and $\blambda^{\top}\psi(\beta,\hat\boldeta;\O_i)$. Therefore
$$
 P(\blambda,\beta,\hat\boldeta)
= - \blambda^{\top}\bar\psi(\beta,\hat\boldeta)
+\tfrac12 \rho''(\bar x) \blambda^{\top}\bar\Omega(\beta,\hat\boldeta)\blambda ,
$$
with $\|\bar\Omega(\beta,\hat\boldeta)\|$ bounded between $c_1$ and $c_2$, and $\rho''(\bar x)$ bounded between $-c_2$ and $-c_1$ for some positive constant $c_1$ and $c_2$ because $\bar x=o_p(1)$ and $\rho''(0)=-1$. Since $\rho$ is a concave function, any local maximizer is global maximizer. Maximization over $L_n$ yields
$$
0=P(0,\beta,\hat\boldeta)\leq P\big(\blambda(\beta,\hat\boldeta),\beta,\hat\boldeta\big)
\le \|\blambda(\beta,\hat\boldeta)\| \|\bar\psi(\beta,\hat\boldeta)\|-c \|\blambda(\beta,\hat\boldeta)\|^2,
$$
for some $c>0$, which implies $\|\blambda(\beta,\hat\boldeta)\|\le c \|\bar\psi(\beta,\hat\boldeta)\|$. Lemma \ref{lemmapsi} then gives
\begin{equation}\label{eq:rate-lambda-hat}
\sup_{\beta\in\mathcal B}\|\blambda(\beta,\hat\boldeta)\|
\le c \sup_{\beta\in\mathcal B}\|\bar\psi(\beta,\hat\boldeta)\|
=O_p(\mu_n/\sqrt n)
=o_p(\kappa_n),
\end{equation}
so the maximizer lies in $L_n$ and is the global maximizer over the whole domain. The same argument with $\boldeta_0$ gives $\sup_{\beta\in \calB}\|\blambda(\beta,\boldeta_0)\|=O_p(\mu_n/\sqrt n)$. Since $\|\bar\psi(\beta_0,\boldeta_0)\|=O_p(\sqrt{m/n})$, \eqref{eq:rate-lambda-hat} yields $\|\blambda(\beta_0,\boldeta_0)\|=O_p(\sqrt{m/n})$.

(ii). The first order condition of $\blambda$ implies that
$$
\frac{1}{n}\sum_{i=1}^n \frac{\partial \rho(\blambda(\beta_0,\hat\boldeta)^{\top}\psi(\beta_0,\hat\boldeta;\O_i))}{\partial \blambda} = \frac{1}{n}\sum_{i=1}^n \rho'  \left(\blambda(\beta_0,\hat\boldeta)^{\top}\psi(\beta_0,\hat\boldeta;\O_i)\right)\psi(\beta_0,\hat\boldeta;\O_i)=0.
$$
A Taylor expansion of $\rho'$ at zero gives
\begin{align}\label{eq:phod1talyor}
0=&\frac{1}{n}\sum_{i=1}^n \rho'  \left(\blambda(\beta_0,\hat\boldeta)^{\top}\psi(\beta_0,\hat\boldeta;\O_i)\right)\psi(\beta_0,\hat\boldeta;\O_i)\n\\
=&\frac{1}{n}\sum_{i=1}^n \rho'(0)\psi(\beta_0,\hat\boldeta;\O_i)+  \frac{1}{n}\sum_{i=1}^n \rho''(0)\psi(\beta_0,\hat\boldeta;\O_i)\psi(\beta_0,\hat\boldeta;\O_i)^{\top}\blambda(\beta_0,\hat\boldeta)+R \n\\
=&  - \bar\psi(\beta_0,\hat\boldeta)-\bar\Omega(\beta_0,\hat\boldeta)\blambda(\beta_0,\hat\boldeta)+R,   
\end{align}
where
$$
R=\frac{1}{2n}\sum_{i=1}^n \rho'''(\bar x_i) \big(\blambda(\beta_0,\hat\boldeta)^{\top}\psi(\beta_0,\hat\boldeta;\O_i)\big)^2 \psi(\beta_0,\hat\boldeta;\O_i),
$$
with $\bar x_i$ between $0$ and $\blambda(\beta_0,\hat\boldeta)^{\top}\psi(\beta_0,\hat\boldeta;\O_i)$. Using \eqref{eq:sup-lin}, the boundedness of $\rho'''(\bar x_i)$, and part (i),
\begin{align}\label{Rorder}
&\|R\|\n\\
\le& c \sup_{i\le n}\|\psi(\beta_0,\hat\boldeta;\O_i)\|\|\blambda(\beta_0,\hat\boldeta)\|^2\|\bar\Omega(\beta_0,\hat\boldeta)\|\n\\
=&O_p  \left(n^{1/\iota}\{\mathbb{E}\|\psi(\beta_0,\hat\boldeta;\O)\|^{\iota}\}  ^{1/\iota}\right)O_p(\mu_n^2/n)O_p(1)\n\\
=&o_p(\mu_n/\sqrt n).
\end{align}
Hence \eqref{eq:phod1talyor} implies
\begin{align*}
\blambda(\beta_0,\hat\boldeta)
=-\bar\Omega(\beta_0,\hat\boldeta)^{-1}\bar\psi(\beta_0,\hat\boldeta)+o_p(\mu_n/\sqrt n),
\end{align*}
and
\begin{align*}
\blambda(\beta_0,\boldeta_0)
=- \bar\Omega(\beta_0,\boldeta_0)^{-1}\bar\psi(\beta_0,\boldeta_0)+o_p(\mu_n/\sqrt n).
\end{align*}
By Lemmas \ref{lemmapsi} and \ref{lemmaOmega},
\begin{align*}
&\|\blambda(\beta_0,\hat\boldeta)-\blambda(\beta_0,\boldeta_0)\|\\
\le& \| \bar\Omega(\beta_0,\boldeta_0)^{-1}\bar\psi(\beta_0,\boldeta_0)-\bar\Omega(\beta_0,\hat\boldeta)^{-1}\bar\psi(\beta_0,\hat\boldeta)\|+o_p(\mu_n/\sqrt n)\\
\le& \|(\bar\Omega(\beta_0,\boldeta_0)^{-1}-\bar\Omega(\beta_0,\hat\boldeta)^{-1})\bar\psi(\beta_0,\boldeta_0)\|
+\|\bar\Omega(\beta_0,\hat\boldeta)^{-1}(\bar\psi(\beta_0,\boldeta_0)-\bar\psi(\beta_0,\hat\boldeta))\|+o_p(\mu_n/\sqrt n)\\
=& o_p(1/\sqrt{m}) O_p(\mu_n/\sqrt{n})+O_p(1) o_p(1/\sqrt n)+o_p(\mu_n/\sqrt n)\\
=& o_p(\mu_n/\sqrt n),
\end{align*}
The uniform bound follows by the same argument with $\beta$ in place of $\beta_0$, which gives $\sup_{\beta\in\mathcal B}\|\blambda(\beta,\hat\boldeta)-\blambda(\beta,\boldeta_0)\|=o_p(\mu_n/\sqrt n)$.
\end{proof}

\begin{lemma} \label{lemmaq}
Let $\tilde Q(\beta,\boldeta)=\bbE \psi(\beta,\boldeta)^{\top}\Omega(\beta,\boldeta)^{-1}\bbE \psi(\beta,\boldeta)/2+m/(2n)$. 
Under Conditions \ref{condweak}--\ref{condgel},
$$
\sup_{\beta\in\mathcal B}\big|\tilde Q(\beta,\boldeta_0)-\hat Q(\beta,\hat\boldeta)\big|
=o_p \left(\frac{\mu_n^2}{n}\right).
$$
\end{lemma}

\begin{proof}
By the Taylor expansion of $\rho$ at zero, we have
\begin{align}\label{qhat}
\hat Q(\beta,\hat\boldeta)
=&\frac{1}{n}\sum_{i=1}^n\rho \left(\blambda(\beta,\hat\boldeta)^{\top}\psi(\beta,\hat\boldeta;\O_i)\right)\n\\
=&\frac{1}{n}\sum_{i=1}^n\Big[\rho(0)
+\rho'(0)\blambda(\beta,\hat\boldeta)^{\top}\psi(\beta,\hat\boldeta;\O_i)\Big]\n\\
&+\frac{1}{2n}\sum_{i=1}^n\rho''(0)\blambda(\beta,\hat\boldeta)^{\top}\psi(\beta,\hat\boldeta;\O_i)\psi(\beta,\hat\boldeta;\O_i)^{\top}\blambda(\beta,\hat\boldeta)+r,
\end{align}
where the remainder term is
$$
r=\frac{1}{6n}\sum_{i=1}^n\rho'''(\bar x_i)
\big(\blambda(\beta,\hat\boldeta)^{\top}\psi(\beta,\hat\boldeta;\O_i)\big)^3,
$$
and each $\bar x_i$ lies between $0$ and $\blambda(\beta,\hat\boldeta)^{\top}\psi(\beta,\hat\boldeta;\O_i)$. To bound $r$, note that $\|\blambda(\beta,\hat\boldeta)\|=O_p(\mu_n/\sqrt n)$ 
and $\max_{i\le n}\sup_{\beta\in\mathcal B}\|\psi(\beta,\hat\boldeta;\O_i)\|=O_p(\kappa)$. 
Since $\rho'''(\bar x_i)$ is bounded, we have
\begin{align}\label{rorder}
|r|
&\le c \|\blambda(\beta,\hat\boldeta)\|\max_{i\le n}\sup_{\beta\in\mathcal B}\|\psi(\beta,\hat\boldeta;\O_i)\|
 \blambda(\beta,\hat\boldeta)^{\top}\bar\Omega(\beta,\hat\boldeta)\blambda(\beta,\hat\boldeta)\nonumber\\
&=O_p \left(\frac{\mu_n}{\sqrt n}\kappa\right)O_p \left(\frac{\mu_n^2}{n}\right)
=o_p \left(\frac{\mu_n^2}{n}\right).
\end{align}

From \eqref{Rorder} and \eqref{rorder}, substituting 
$\blambda(\beta,\hat\boldeta)
=-\bar\Omega(\beta,\hat\boldeta)^{-1}\bar \psi(\beta,\hat\boldeta)+R$ 
with $\|R\|=o_p(\mu_n/\sqrt n)$ into \eqref{qhat}, we have
\begin{align}\label{eq:hatQ-CUE}
&\hat Q(\beta,\hat\boldeta)\n\\
=&-\frac{1}{n}\sum_{i=1}^n(-\bar \psi(\beta,\hat\boldeta)+R)^{\top}
\bar\Omega(\beta,\hat\boldeta)^{-1}\psi(\beta,\hat\boldeta;\O_i)\nonumber\\
&-\frac{1}{2n}\sum_{i=1}^n(-\bar \psi(\beta,\hat\boldeta)+R)^{\top}
\bar\Omega(\beta,\hat\boldeta)^{-1}\psi(\beta,\hat\boldeta;\O_i)
\psi(\beta,\hat\boldeta;\O_i)^{\top}\bar\Omega(\beta,\hat\boldeta)^{-1}
(-\bar \psi(\beta,\hat\boldeta)+R)
+r\nonumber\\
=&\frac{1}{2} \bar \psi(\beta,\hat\boldeta)^{\top}
\bar\Omega(\beta,\hat\boldeta)^{-1}\bar \psi(\beta,\hat\boldeta)
+o_p \left(\frac{\mu_n^2}{n}\right).
\end{align}

Furthermore, Lemma 7 in \cite{ye2024} shows that
$$
\sup_{\beta\in\mathcal B}
\Big|
\frac{1}{2}\bar \psi(\beta,\hat\boldeta)^{\top}\bar\Omega(\beta,\hat\boldeta)^{-1}\bar \psi(\beta,\hat\boldeta)
-\tilde Q(\beta,\boldeta_0)
\Big|
=o_p \left(\frac{\mu_n^2}{n}\right).
$$
Combining the above results yields
$$
\sup_{\beta\in\mathcal B}
\big|\hat Q(\beta,\hat\boldeta)-\tilde Q(\beta,\boldeta_0)\big|
=o_p \left(\frac{\mu_n^2}{n}\right),
$$
which completes the proof.
\end{proof}

With these useful Lemmas, we can now turn to prove Theorem \ref{thmnormality}.
\begin{proof}
We first show consistency. 
By Lemma \ref{lemma1a}, it suffices to prove that 
$\sqrt{n}\|\bbE{\psi}(\hat{\beta},\boldeta_0)\|/\mu_n=o_p(1)$.
First, we have
$\bbE{\psi}(\beta_0,\boldeta_0)=\mathbb E\{\psi(\beta_0,\boldeta_0;\O)\}=0$, so
$$
\tilde Q(\beta_0,\boldeta_0)
=\frac{1}{2}\bbE{\psi}(\beta_0,\boldeta_0)^{\top}\tilde{\Omega}(\beta_0,\boldeta_0)^{-1}
\bbE{\psi}(\beta_0,\boldeta_0)
+\frac{m}{2n}
=\frac{m}{2n}.
$$
Using Lemma \ref{lemma2} and Condition \ref{condeigen}, we have
\begin{align*}
\|\hat{\beta}-\beta_0\|^2
&\le \frac{n}{\mu_n^2}\|\bbE{\psi}(\hat{\beta},\boldeta_0)\|^2\\
&\le c \frac{n}{\mu_n^2}
\bbE{\psi}(\hat{\beta},\boldeta_0)^{\top}\Omega(\hat{\beta},\boldeta_0)^{-1}
\bbE{\psi}(\hat{\beta},\boldeta_0)
+\frac{m}{2n}-\frac{m}{2n}\\
&\le c \frac{n}{\mu_n^2}
\big\{\tilde Q(\hat{\beta},\boldeta_0)-\tilde Q(\beta_0,\boldeta_0)\big\}\\
&\le c \frac{n}{\mu_n^2}
\big\{\tilde Q(\hat{\beta},\boldeta_0)-\hat Q(\hat{\beta},\hat{\boldeta})
+\hat Q(\hat{\beta},\hat{\boldeta})-\hat Q(\beta_0,\hat{\boldeta})
+\hat Q(\beta_0,\hat{\boldeta})-\tilde Q(\beta_0,\boldeta_0)\big\}.
\end{align*}
By Lemma \ref{lemmaq} and the fact that $\hat{\beta}$ minimizes $\hat Q(\beta,\hat{\boldeta})$, i.e.,
$\hat Q(\hat{\beta},\hat{\boldeta})\le\hat Q(\beta_0,\hat{\boldeta})$, so
$$
\|\hat{\beta}-\beta_0\|^2
\le c \frac{n}{\mu_n^2}o_p \left(\frac{\mu_n^2}{n}\right)
=o_p(1),
$$
which establishes the consistency.

We now turn to asymptotic normality. 
Consider the first order condition 
$\partial\hat Q(\beta,\hat{\boldeta})/\partial\beta|_{\beta=\hat{\beta}}=0$.
Taylor expansion around $\beta_0$ yields
\begin{equation}\label{Taylor}
0
=\frac{n}{\mu_n}
\left.\frac{\partial\hat Q(\beta,\hat{\boldeta})}{\partial\beta}\right|_{\beta=\hat{\beta}}
=\frac{n}{\mu_n}
\left.\frac{\partial\hat Q(\beta,\hat{\boldeta})}{\partial\beta}\right|_{\beta=\beta_0}
+\frac{n}{\mu_n^2}
\left.\frac{\partial^2\hat Q(\beta,\hat{\boldeta})}{\partial\beta^2}\right|_{\beta=\bar{\beta}}
\mu_n(\hat{\beta}-\beta_0),
\end{equation}
where $\bar{\beta}$ lies between $\beta_0$ and $\hat{\beta}$.

We first analyze 
$\partial\hat Q(\beta,\hat{\boldeta})/\partial\beta|_{\beta=\beta_0}$.
The first order condition for $\blambda$ implies
\begin{equation}\label{lambdafirstorder}
\frac{\partial \hat Q(\beta,\hat{\boldeta})}{\partial \blambda}=\frac{1}{n}\sum_{i=1}^n
\rho' \left(\blambda(\beta_0,\hat{\boldeta})^{\top}\psi(\beta_0,\hat{\boldeta};\O_i)\right)
\psi(\beta_0,\hat{\boldeta};\O_i)
=0.
\end{equation}
Define
\begin{align*}
\hat U_i
=\psi'(\hat\boldeta;\O_i)-\psi^{*}-\frac{1}{n}\sum_{j=1}^n\big\{\psi(\beta_0,\hat{\boldeta};\O_j)\psi'(\hat{\boldeta};\O_j)^{\top}\big\}
\bar{\Omega}(\beta_0,\hat{\boldeta})^{-1}
\psi(\beta_0,\hat{\boldeta};\O_i).
\end{align*}
Combining \eqref{lambdafirstorder}, Lemma \ref{lemmalambda}, we obtain
\begin{align}\label{qda123}
&\left.\frac{\partial \hat{Q}(\beta, \hat{\boldeta})}{\partial \beta}\right|_{\beta=\beta_0}\n\\
=&\frac{1}{n}\sum_{i=1}^n\rho'(\blambda(\beta_0,\hat\boldeta)\tp \psi(\beta_0,\hat\boldeta;\O_i))\blambda(\beta_0,\hat\boldeta)\tp \psi'(\hat\boldeta;\O_i)\n\\
&+\frac{1}{n}\sum_{i=1}^n
\rho' \left(\blambda(\beta_0,\hat{\boldeta})^{\top}\psi(\beta_0,\hat{\boldeta};\O_i)\right)
\psi(\beta_0,\hat{\boldeta};\O_i)\tp\frac{\partial \blambda(\beta_0,\hat{\boldeta})}{\partial \beta} \n\\
=&\frac{1}{n}\sum_{i=1}^n\{-1-\blambda(\beta_0,\hat\boldeta)\tp \psi(\beta_0,\hat\boldeta;\O_i)+r\}\{-\bar \psi(\beta_0,\hat\boldeta)\tp\bar\Omega(\beta_0,\hat\boldeta)^{-1}+R\}\psi'(\hat\boldeta;\O_i) \n\\
=&-\bar \psi(\beta_0,\hat\boldeta)\tp\bar\Omega(\beta_0,\hat\boldeta)^{-1}\left\{\frac{1}{n}\sum_{i=1}^{n}\psi(\beta_0,\hat\boldeta;\O_i)\psi'(\hat\boldeta;\O_i)\tp\right\}\bar\Omega(\beta_0,\hat\boldeta)^{-1}\bar \psi(\beta_0,\hat\boldeta) \n\\
&+\bar {\psi'}(\hat\boldeta)\tp\bar\Omega(\beta_0,\hat\boldeta)^{-1}\bar \psi(\beta_0,\hat\boldeta)+o_p({\mu_n}/{n})\n\\
=& \psi^{*\top} \bar\Omega(\beta_0,\hat\boldeta)^{-1}\bar \psi(\beta_0,\hat\boldeta)+ \frac{1}{n}\sum_{i=1}^n\hat U_i\tp\bar \Omega(\beta_0,\hat\boldeta)^{-1}\bar \psi(\beta_0,\hat\boldeta)+o_p({\mu_n}/{n}).
\end{align}

Lemma A12 in \cite{newey2009} gives
\begin{equation}\label{eqlemmaa12}
\frac{n}{\mu_n}
\left\{
\psi^{*\top}\Omega(\beta_0,\boldeta_0)^{-1}\bar{\psi}(\beta_0,\boldeta_0)
+\frac{1}{n}\sum_{i=1}^n
U_i^{\top}\Omega(\beta_0,\boldeta_0)^{-1}\bar{\psi}(\beta_0,\boldeta_0)
+o_p(1)
\right\}
\xrightarrow{d} N(0,H+V).
\end{equation}
The expression in \eqref{qda123} differs from \eqref{eqlemmaa12} only through the nuisance term: 
$\hat{\boldeta}$ instead of $\boldeta_0$. 
We then show that
\begin{equation}\label{normal1}
\psi^{*\top}\Omega(\beta_0,\boldeta_0)^{-1}\bar{\psi}(\beta_0,\boldeta_0)
-\psi^{*\top}\bar{\Omega}(\beta_0,\hat{\boldeta})^{-1}\bar{\psi}(\beta_0,\hat{\boldeta})
=o_p \left(\frac{\mu_n}{n}\right),
\end{equation}
and
\begin{equation}\label{normal2}
\frac{1}{n}\sum_{i=1}^n
U_i^{\top}\Omega(\beta_0,\boldeta_0)^{-1}\bar{\psi}(\beta_0,\boldeta_0)
-\frac{1}{n}\sum_{i=1}^n
\hat U_i^{\top}\bar{\Omega}(\beta_0,\hat{\boldeta})^{-1}\bar{\psi}(\beta_0,\hat{\boldeta})
=o_p \left(\frac{\mu_n}{n}\right).
\end{equation}

For \eqref{normal1}, using Condition \ref{condeigen} and Lemmas \ref{lemmapsi} and \ref{lemmaOmega},
\begin{align*}
&\psi^{*\top}\Omega(\beta_0,\boldeta_0)^{-1}\bar{\psi}(\beta_0,\boldeta_0)
-\psi^{*\top}\bar{\Omega}(\beta_0,\hat{\boldeta})^{-1}\bar{\psi}(\beta_0,\hat{\boldeta})\\
=&
\psi^{*\top}\Omega(\beta_0,\boldeta_0)^{-1}
\{\bar{\psi}(\beta_0,\boldeta_0)-\bar{\psi}(\beta_0,\hat{\boldeta})\}\\
&+
\psi^{*\top}
\{\Omega(\beta_0,\boldeta_0)^{-1}-\bar{\Omega}(\beta_0,\hat{\boldeta})^{-1}\}
\bar{\psi}(\beta_0,\hat{\boldeta})\\
=&O_p(\mu_n/\sqrt{n}) o_p(n^{-1/2})
+O_p(\mu_n/\sqrt{n})  o_p(\sqrt{1/m})O_p(\sqrt{m/n})\\
=&o_p(\mu_n/n).
\end{align*}
By the definitions of $U_i$ and $\hat U_i$, we have
\be
&&\left\|\frac{1}{n}\sum_{i=1}^n(\hat U_i-U_i)\right\|\n\\
&\leq&\|\bar{\psi}'(\hat\boldeta)-\bar{\psi}'(\boldeta_0)\|\n\\
&&+\Big\|\Big\{\frac{1}{n}\sum_{i=1}^n\psi(\beta_0,\hat\boldeta;\O_i)\psi'(\hat\boldeta;\O_i)\tp\Big\}
\bar\Omega(\beta_0,\hat\boldeta)^{-1}\bar\psi(\beta_0,\hat\boldeta)\n\\
&&\quad-\bbE\{\psi(\beta_0,\boldeta_0;\O)\psi'(\boldeta_0;\O)\tp\}
\Omega(\beta_0,\boldeta_0)^{-1}\bar\psi(\beta_0,\boldeta_0)\Big\|\n\\
&=&\|\bar{\psi}'(\hat\boldeta)-\bar{\psi}'(\boldeta_0)\|\n\\
&&+\Big\|\Big[\frac{1}{n}\sum_{i=1}^n\psi(\beta_0,\hat\boldeta;\O_i)\psi'(\hat\boldeta;\O_i)\tp
-\bbE\{\psi(\beta_0,\boldeta_0;\O)\psi'(\boldeta_0;\O)\tp\}\Big]
\Omega(\beta_0,\hat\boldeta)^{-1}\bar\psi(\beta_0,\hat\boldeta)\Big\|\n\\
&&+\Big\|\bbE\{\psi(\beta_0,\boldeta_0;\O)\psi'(\boldeta_0;\O)\tp\}
\{\bar\Omega(\beta_0,\hat\boldeta)^{-1}-\Omega(\beta_0,\boldeta_0)^{-1}\}
\bar\psi(\beta_0,\hat\boldeta)\Big\|\n\\
&&+\Big\|\bbE\{\psi(\beta_0,\boldeta_0;\O)\psi'(\boldeta_0;\O)\tp\}
\Omega(\beta_0,\boldeta_0)^{-1}\{\bar\psi(\beta_0,\hat\boldeta)-\bar\psi(\beta_0,\boldeta_0)\}\Big\|\n\\
&=&o_p(n^{-1/2})
+o_p\Big(\frac{1}{\sqrt{m}}\Big)O_p\Big(\frac{\mu_n}{\sqrt{n}}\Big)
+o_p\Big(\frac{1}{\sqrt{m}}\Big)O_p\Big(\frac{\mu_n}{\sqrt{n}}\Big)
+o_p(n^{-1/2})\n\\
&=&o_p\Big(\frac{\mu_n}{\sqrt{mn}}\Big)+o_p(n^{-1/2}).
\ee
Moreover, we have
\be \label{ui}
&&\Big\|\frac{1}{n}\sum_{i=1}^n U_i\Big\|\n\\
&\leq&
\|\bar{\psi}'(\hat\boldeta)\|+\|\psi^{*}\|
+\big\|\bbE\{\psi(\beta_0,\boldeta_0;\O)\psi^{*\top}\}
\Omega_0^{-1}\bar\psi(\beta_0,\boldeta_0)\big\|\n\\
&=&O_p\Big(\frac{\mu_n}{\sqrt{n}}\Big).
\ee
For \eqref{normal2}, based on Condition \ref{condeigen}, Lemmas \ref{lemmapsi}, \ref{lemmaOmega} and \eqref{ui}, we obtain
\be
&&\frac{1}{n}\sum_{i=1}^n U_i\tp\Omega(\beta_0,\boldeta_0)^{-1}\bar\psi(\beta_0,\boldeta_0)
-\frac{1}{n}\sum_{i=1}^n\hat U_i\tp\bar\Omega(\beta_0,\hat\boldeta)^{-1}\bar\psi(\beta_0,\hat\boldeta)\n\\
&=&\frac{1}{n}\sum_{i=1}^n
U_i\tp\Omega(\beta_0,\boldeta_0)^{-1}\{\bar\psi(\beta_0,\boldeta_0)-\bar\psi(\beta_0,\hat\boldeta)\}\n\\
&&+\frac{1}{n}\sum_{i=1}^n
U_i\tp\{\Omega(\beta_0,\boldeta_0)^{-1}-\bar\Omega(\beta_0,\hat\boldeta)^{-1}\}
\bar\psi(\beta_0,\hat\boldeta)\n\\
&&+\frac{1}{n}\sum_{i=1}^n
(U_i-\hat U_i)\tp\bar\Omega(\beta_0,\hat\boldeta)^{-1}\bar\psi(\beta_0,\hat\boldeta)\n\\
&=&
O_p\Big(\frac{\mu_n}{\sqrt{n}}\Big)O_p(1)o_p\Big(\frac{1}{\sqrt{n}}\Big)
+O_p\Big(\frac{\mu_n}{\sqrt{n}}\Big)
o_p\Big(\frac{1}{\sqrt{m}}\Big)
O_p\Big(\frac{\sqrt{m}}{\sqrt{n}}\Big)\n\\
&&+o_p\Big(\frac{\mu_n}{\sqrt{mn}}\Big)O_p(1)
O_p\Big(\frac{\sqrt{m}}{\sqrt{n}}\Big)
\n\\
&=&o_p(\mu_n/n).\n
\ee
Combining \eqref{qda123}--\eqref{normal2}, we obtain
\be\label{Taylor1}
\frac{n}{\mu_n}
\left.\frac{\partial \hat{Q}(\beta,\hat{\boldeta})}{\partial \beta}\right|_{\beta=\beta_0}
\xrightarrow{d} N(0,H+\bLambda).
\ee
Next we show
$$
\sup_{\beta\in\calB}\frac{n}{\mu_n^2}
\frac{\partial^2 \hat{Q}(\beta,\hat\boldeta)}{\partial \beta^2}\to H.
$$
Define
\[
\calV_{\psi\psi}(\beta,\boldeta)
=\frac{1}{n}\sum_{i=1}^n\rho''(\blambda(\beta,\boldeta)\tp\psi(\beta,\boldeta;\O_i))
\psi(\beta,\boldeta;\O_i)\psi(\beta,\boldeta;\O_i)\tp,
\]
\[
\calV_{\psi\psi'}(\beta,\boldeta)
=\frac{1}{n}\sum_{i=1}^n\rho''(\blambda(\beta,\boldeta)\tp\psi(\beta,\boldeta;\O_i))
\psi(\beta,\boldeta;\O_i)\psi'(\boldeta;\O_i)\tp,
\]
\[
\calV_{\psi\psi'}^{v}(\beta,\boldeta)
=\frac{1}{n}\sum_{i=1}^n\rho''(v)
\psi'(\boldeta;\O_i)\psi'(\boldeta;\O_i)\tp,
\]
\[
\calV_{\psi'\psi'}(\beta,\boldeta)
=\frac{1}{n}\sum_{i=1}^n\rho''(\blambda(\beta,\boldeta)\tp\psi(\beta,\boldeta;\O_i))
\psi'(\boldeta;\O_i)\psi'(\boldeta;\O_i)\tp.
\]
Then
\be\label{partialqbb}
&&\frac{\partial^2 \hat{Q}(\beta,\boldeta)}{\partial \beta^2}\n\\
&=&\frac{1}{n}\sum_{i=1}^n
\rho''(\blambda(\beta,\boldeta)\tp\psi(\beta,\boldeta;\O_i))
\{\blambda(\beta,\boldeta)\tp\psi'(\boldeta;\O_i)\}
\Big\{\blambda(\beta,\boldeta)\tp\psi'(\boldeta;\O_i)
+\frac{\partial\blambda(\beta,\boldeta)\tp}{\partial\beta}
\psi(\beta,\boldeta;\O_i)\Big\}\n\\
&&+\frac{1}{n}\sum_{i=1}^n
\rho'(\blambda(\beta,\boldeta)\tp\psi(\beta,\boldeta;\O_i))
\frac{\partial\blambda(\beta,\boldeta)\tp}{\partial\beta}\psi'(\boldeta;\O_i)\n\\
&=&\blambda(\beta,\boldeta)\tp
\calV_{\psi'\psi'}(\beta,\boldeta)\blambda(\beta,\boldeta)
+\blambda(\beta,\boldeta)\tp
\calV_{\psi\psi'}(\beta,\boldeta)
\frac{\partial\blambda(\beta,\boldeta)}{\partial\beta}\n\\
&&+\frac{1}{n}\sum_{i=1}^n
\{-1+\rho''(\bar v_1)\blambda(\beta,\boldeta)\tp\psi(\beta,\boldeta;\O_i)\}
\frac{\partial\blambda(\beta,\boldeta)\tp}{\partial\beta}\psi'(\boldeta;\O_i)\n\\
&=&\blambda(\beta,\boldeta)\tp
\calV_{\psi'\psi'}(\beta,\boldeta)\blambda(\beta,\boldeta)
+\blambda(\beta,\boldeta)\tp
\calV_{\psi\psi'}(\beta,\boldeta)
\frac{\partial\blambda(\beta,\boldeta)}{\partial\beta}\n\\
&&+\blambda(\beta,\boldeta)\tp
\calV_{\psi\psi'}^{\bar v_1}(\beta,\boldeta)
\frac{\partial\blambda(\beta,\boldeta)}{\partial\beta}
-\frac{\partial\blambda(\beta,\boldeta)\tp}{\partial\beta}
\bar{\psi}'(\boldeta)\n\\
&\equiv& D_1(\boldeta)+D_2(\boldeta)+D_3(\boldeta)+D_4(\boldeta),
\ee
where $|\bar v_1|\leq|\blambda(\beta,\hat\boldeta)\tp\psi(\beta,\hat\boldeta;\O_i)|$. The derivative of $\blambda(\beta,\boldeta)$ with respect to $\beta$ can be obtained
from the implicit function theorem:
\bse
\frac{\partial\blambda(\beta,\boldeta)}{\partial \beta}
&=&-\calV_{\psi\psi}(\beta,\boldeta)^{-1}
\Big\{\calV_{\psi\psi'}(\beta,\boldeta)\blambda(\beta,\boldeta)
+\frac{1}{n}\sum_{i=1}^n
\rho'(\blambda(\beta,\boldeta)\tp\psi(\beta,\boldeta;\O_i))
\psi'(\boldeta;\O_i)\Big\}\\
&=&-\calV_{\psi\psi}(\beta,\boldeta)^{-1}
\big[\{\calV_{\psi\psi'}(\beta,\boldeta)
+\calV_{\psi\psi'}^{\bar v_2}(\beta,\boldeta)\}
\blambda(\beta,\boldeta)-\psi'(\boldeta;\O_i)\big].
\ese
We now show $\sup_{\beta\in\calB}|D_i(\hat\boldeta)-D_i(\boldeta_0)|
=o_p(\mu_n^2/n)$ for $i=1,2,3,4$.
Using a Taylor expansion of $\rho''$ at $0$, we have
\bse
&&\|\calV_{\psi\psi}(\beta,\boldeta_0)+\bar\Omega(\beta,\boldeta_0)\|\\
&=&\Big\|\frac{1}{n}\sum_{i=1}^n
\{\rho''(\blambda(\beta,\boldeta_0)\tp\psi(\beta,\boldeta_0;\O_i))
-\rho''(0)\}
\psi(\beta,\boldeta_0;\O_i)\psi(\beta,\boldeta_0;\O_i)\tp\Big\|\\
&=&\Big\|\frac{1}{n}\sum_{i=1}^n
\rho^{'''}(x)\blambda(\beta,\boldeta_0)\tp\psi(\beta,\boldeta_0;\O_i)
\psi(\beta,\boldeta_0;\O_i)\psi(\beta,\boldeta_0;\O_i)\tp\Big\|\\
&\leq&c\max_{i\leq n}|\blambda(\beta,\boldeta_0)\tp\psi(\beta,\boldeta_0;\O_i)|
\frac{1}{n}\sum_{i=1}^n
\|\psi(\beta,\boldeta_0;\O_i)\psi(\beta,\boldeta_0;\O_i)\tp\|\\
&=&o_p(1)O_p(1)\\
&=&o_p(1),
\ese
where $|x|\leq|\blambda(\beta,\boldeta_0)\tp\psi(\beta,\boldeta_0;\O_i)|$.
Similarly, $\|\calV_{\psi\psi}(\beta,\hat\boldeta)+\bar\Omega(\beta,\hat\boldeta)\|=o_p(1)$.
Combining with Lemma \ref{lemmaOmega}, it follows that
\bse
&&\|\calV_{\psi\psi}(\beta,\hat\boldeta)
-\calV_{\psi\psi}(\beta,\boldeta_0)\|\\
&\leq&
\|\calV_{\psi\psi}(\beta,\hat\boldeta)+\bar\Omega(\beta,\hat\boldeta)\|
+\|\bar\Omega(\beta,\hat\boldeta)-\bar\Omega(\beta,\boldeta_0)\|
+\|\calV_{\psi\psi}(\beta,\boldeta_0)+\bar\Omega(\beta,\boldeta_0)\|\\
&=&o_p(1).
\ese
Analogous arguments give
$\|\calV_{\psi\psi'}(\beta,\hat\boldeta)
-\calV_{\psi\psi'}(\beta,\boldeta_0)\|=o_p(1)$,
$\|\calV_{\psi\psi'}^{v}(\beta,\hat\boldeta)
-\calV_{\psi\psi'}^{v}(\beta,\boldeta_0)\|=o_p(1)$,
and
$\|\calV_{\psi'\psi'}(\beta,\hat\boldeta)
-\calV_{\psi'\psi'}(\beta,\boldeta_0)\|=o_p(1)$.
Therefore,
\bse
&&\sup_{\beta\in\calB}|D_1(\hat\boldeta)-D_1(\boldeta_0)|\\
&=&\sup_{\beta\in\calB}
|\{\blambda(\beta,\hat\boldeta)-\blambda(\beta,\boldeta_0)\}\tp
\calV_{\psi'\psi'}(\beta,\hat\boldeta)
\{\blambda(\beta,\hat\boldeta)+\blambda(\beta,\boldeta_0)\}|\\
&&+\sup_{\beta\in\calB}
|\blambda(\beta,\hat\boldeta)\tp
\{\calV_{\psi'\psi'}(\beta,\hat\boldeta)
-\calV_{\psi'\psi'}(\beta,\boldeta_0)\}
\blambda(\beta,\boldeta_0)|\\
&=&o_p(\mu_n/\sqrt{n})O_p(1)O_p(\mu_n/\sqrt{n})
+O_p(\mu_n/\sqrt{n})o_p(1)O_p(\mu_n/\sqrt{n})\\
&=&o_p(\mu_n^2/n).
\ese
For $D_2$ and $D_3$, it is enough to show
$\sup_{\beta\in\calB}\|\partial\blambda(\beta,\hat\boldeta)/\partial\beta
-\partial\blambda(\beta,\boldeta_0)/\partial\beta\|
=o_p(\mu_n/\sqrt{n})$
and
$\sup_{\beta\in\calB}\|\partial\blambda(\beta,\boldeta)/\partial\beta\|
=O_p(\mu_n/\sqrt{n})$.
Then we can repeat the argument for $D_1$ with $\blambda$ replaced by
$\partial\blambda/\partial\beta$.
Based on Lemmas \ref{lemmapsi} and \ref{lemmalambda},
\bse
&&\sup_{\beta\in\calB}
\Big\|\frac{\partial\blambda(\beta,\boldeta)}{\partial \beta}\Big\|\\
&\leq&\sup_{\beta\in\calB}
\|\calV_{\psi\psi}(\beta,\boldeta)^{-1}
\{\calV_{\psi\psi'}(\beta,\boldeta)
+\calV_{\psi\psi'}^{\bar v_2}(\beta,\boldeta)\}
\blambda(\beta,\boldeta)\|\\
&&+\sup_{\beta\in\calB}
\|\calV_{\psi\psi}(\beta,\boldeta)^{-1}\psi'(\boldeta;\O_i)\|\\
&=&O_p(1)O_p(\mu_n/\sqrt{n})+O_p(1)O_p(\mu_n/\sqrt{n})\\
&=&O_p(\mu_n/\sqrt{n}).
\ese
Similarly,
$\sup_{\beta\in\calB}\|\partial\blambda(\beta,\hat\boldeta)/\partial\beta
-\partial\blambda(\beta,\boldeta_0)/\partial\beta\|
=o_p(\mu_n/\sqrt{n})$,
which implies
$\sup_{\beta\in\calB}|D_i(\hat\boldeta)-D_i(\boldeta_0)|
=o_p(\mu_n^2/n)$ for $i=2,3$.
For $D_4$, we have
\bse
&&\sup_{\beta\in\calB}|D_4(\hat\boldeta)-D_4(\boldeta_0)|\\
&\leq&\sup_{\beta\in\calB}
\Big|\frac{\partial\blambda(\beta,\hat\boldeta)\tp}{\partial\beta}
\{\psi'(\hat\boldeta)-\psi'(\boldeta_0)\}\Big|\\
&&+\sup_{\beta\in\calB}
\Big|\Big\{\frac{\partial\blambda(\beta,\hat\boldeta)\tp}{\partial\beta}
-\frac{\partial\blambda(\beta,\boldeta_0)\tp}{\partial\beta}\Big\}
\psi'(\boldeta_0)\Big|\\
&=&o_p(\mu_n^2/n).
\ese
Combining the bounds for $D_1$ to $D_4$, we obtain
\begin{equation*}
\sup_{\beta\in\calB}
\Big|
\frac{\partial^2 \hat{Q}(\beta,\hat\boldeta)}{\partial \beta^2}
-\frac{\partial^2 \hat{Q}(\beta,\boldeta_0)}{\partial \beta^2}
\Big|
=o_p(\mu_n^2/n).     
\end{equation*}
Lemma 13 in \cite{newey2009} implies
$$
\sup_{\beta\in\calB}
\frac{n}{\mu_n^2}
\Big|
\frac{\partial^2 \hat{Q}(\beta,\boldeta_0)}{\partial \beta^2}
\Big|
\to H,
$$
which yields
\be\label{Taylor2}
\sup_{\beta\in\calB}
\frac{n}{\mu_n^2}
\Big|
\frac{\partial^2 \hat{Q}(\beta,\hat\boldeta)}{\partial \beta^2}
\Big|=\sup_{\beta\in\calB}\Big|
\frac{\partial^2 \hat{Q}(\beta,\hat\boldeta)}{\partial \beta^2}
-\frac{\partial^2 \hat{Q}(\beta,\boldeta_0)}{\partial \beta^2}
\Big|+\sup_{\beta\in\calB}\Big|
\frac{\partial^2 \hat{Q}(\beta,\boldeta_0)}{\partial \beta^2}
\Big|
\to H.
\ee
Finally, combining \eqref{Taylor}, \eqref{Taylor1} and \eqref{Taylor2}, the asymptotic
normality of $\hat\beta$ follows:
$$
\frac{\mu_n(\hat{\beta}-\beta_0)}{\sqrt{H^{-1}(H+V)H^{-1}}}\xrightarrow{d}N(0,1).
$$
\end{proof}
\section{Proof of Theorem \ref{thm:overid}}

The proof of Theorem 4 in \cite{newey2009} implies that
$$
\frac{2n\widehat Q(\beta_0,\boldeta_0)-m}{\sqrt{2m}} \to N(0,1).
$$
Next we will prove $
2n\widehat Q(\widehat \beta,\widehat \boldeta)-2n\widehat Q(\beta_0,\boldeta_0)=o_p(m).
$
Taylor expansion at $\beta_0$ gives that
\be      \label{eq:2nQbeta}
&& 2n\widehat Q(\widehat \beta,\widehat \boldeta)-2n\widehat Q(\beta_0,\widehat\boldeta)\n\\
&=& \frac{2n}{\mu_n}\frac{\partial \widehat Q(\beta_0,\widehat\boldeta)}{\partial\beta}\mu_n(\widehat\beta-\beta_0) + \frac{2n}{\mu_n^2}\frac{\partial^2 \widehat Q(\bar\beta,\widehat\boldeta)}{\partial\beta^2}\mu_n^2(\widehat\beta-\beta_0)^2\n\\
&=& O_p(1),
\ee
where the last line comes from \eqref{Taylor1} and \eqref{Taylor2}.
Since $\|\bar\psi(\beta_0,\widehat \boldeta)\|=o_p(\sqrt{m/n}) $, by replacing $\beta$ with $\beta_0$, the remainder term in \eqref{rorder} becomes 
\begin{align}
|r|&\le c \|\blambda(\beta_0,\hat\boldeta)\|\max_{i\le n}\sup_{\beta\in\mathcal B}\|\psi(\beta_0,\hat\boldeta;\O_i)\|
 \blambda(\beta_0,\hat\boldeta)^{\top}\bar\Omega(\beta_0,\hat\boldeta)\blambda(\beta_0,\hat\boldeta)\nonumber\\
&=O_p \left(\sqrt{\frac{m}{n}}\kappa\right)O_p \left(\frac{m^2}{n}\right) 
=o_p \left(\frac{\sqrt{m}}{n}\right).
\end{align} and thus \eqref{eq:hatQ-CUE} becomes
\begin{align}
\hat Q(\beta_0,\hat\boldeta)=\frac{1}{2} \bar \psi(\beta_0,\hat\boldeta)^{\top} \bar\Omega(\beta_0,\hat\boldeta)^{-1}\bar \psi(\beta_0,\hat\boldeta) +o_p \left(\sqrt{m}/n\right).
\end{align}
 Then we have
\be    \label{eq:2nQeta}  
&& |2n\widehat Q(\beta_0,\widehat \boldeta)-2n\widehat Q(\beta_0,\boldeta_0)|\n\\
&=& n|\bar\psi(\beta_0,\hat\boldeta)\tp\bar\Omega(\beta_0,\hat\boldeta)\bar\psi(\beta_0,\hat\boldeta)-\bar\psi(\beta_0,\boldeta_0)\tp\bar\Omega(\beta_0,\boldeta_0)\bar\psi(\beta_0,\boldeta_0)| +o_p(\sqrt{m})\n\\
&=&  n\|\bar\psi(\beta_0,\hat\boldeta)-\bar\psi(\beta_0,\boldeta_0)\|\|\bar\Omega(\beta_0,\hat\boldeta)\|\|\bar\psi(\beta_0,\hat\boldeta)\| \n\\
&&+  n\|\bar\psi(\beta_0,\boldeta_0)\|\|\bar\Omega(\beta_0,\hat\boldeta)-\bar\Omega(\beta_0,\boldeta_0)\|\|\bar\psi(\beta_0,\hat\boldeta)\|\n\\
&&+ n\|\bar\psi(\beta_0,\hat\boldeta)-\bar\psi(\beta_0,\boldeta_0)\|\|\bar\Omega(\beta_0,\boldeta_0)\|\|\bar\psi(\beta_0,\boldeta_0)\|+o_p \left(\sqrt{m}\right)\n\\
&=& no_p(1/\sqrt{n})O_p(1)O_p(\sqrt{m/n}) + nO_p(\sqrt{m/n})o_p(1/\sqrt{m})O_p(\sqrt{m/n})\n\\
&&+nO_p(\sqrt{m/n})O_p(1)o_p(1/\sqrt{n})+o_p \left(\sqrt{m}\right)\n \\
&=& o_p \left(\sqrt{m}\right).
\ee
Together with \eqref{eq:2nQbeta} and \eqref{eq:2nQeta}, we have
\bse
&&\frac{2n\widehat Q(\widehat\beta,\widehat\boldeta)-2n\widehat Q(\beta_0,\boldeta_0)}{\sqrt{2(m-1)}} \\
&=&\frac{2n\widehat Q(\widehat\beta,\widehat\boldeta)-2n\widehat Q(\beta_0,\widehat\boldeta)+2n\widehat Q(\beta_0,\widehat\boldeta)-2n\widehat Q(\beta_0,\boldeta_0)}{\sqrt{2(m-1)}} \\
&=&o_p(1).
\ese
Therefore,
 \bse
&&\frac{2n\widehat Q(\widehat\beta,\widehat\boldeta)-(m-1)}{\sqrt{2(m-1)}}\\
&=&\frac{2n\widehat Q(\widehat\beta,\widehat\boldeta)-2n\widehat Q(\beta_0,\boldeta_0)+2n\widehat Q(\beta_0,\boldeta_0)-(m-1)}{\sqrt{2(m-1)}}\\
&=& \frac{\sqrt{2m}}{\sqrt{2(m-1)}}\frac{2n\widehat Q(\beta_0,\boldeta_0)-m}{\sqrt{2m}}+ \frac{1}{\sqrt{2(m-1)}} +o_p(1)
\ese
as $m,n\to\infty$. Consequently, the preceding arguments imply that
\[
P\!\left(2n\hat Q(\hat\beta,\hat\boldeta)\ge \chi^2_{1-\alpha}(m-1)\right)
=
P\!\left(
\frac{2n\hat Q(\hat\beta,\hat\boldeta)-(m-1)}{\sqrt{2(m-1)}}
\ge
\frac{\chi^2_{1-\alpha}(m-1)-(m-1)}{\sqrt{2(m-1)}}
\right)
\to \alpha .
\]
\section{Additional simulation results}\label{secsuppsimu}
In this section, we further examine the performance of the competing methods under a range of simulation settings, including censoring rates of 0.2, 0.4, and 0.6, Cases 1--4, numbers of instrumental variables equal to 10 and 20, and sample sizes of 10,000 and 20,000. The results are broadly consistent with those reported in the main text. Across all scenarios considered, the proposed framework exhibits relatively small bias and achieves coverage probabilities close to the nominal level. In contrast, the standard AFT model shows substantial bias throughout.
\begin{table}[htbp]
\centering
\caption{Performance of the proposed estimators for $n=10{,}000$ across Cases 1-4, varying censoring rates, and instrument dimensions $p=10$ and $p=20$. CR is the censoring rate, BIAS denotes percentage finite-sample bias, SD the empirical standard deviation, SE the average estimated standard error, and CP the empirical coverage probability of the 95\% confidence interval. }
\label{tabsimu1.5}
\footnotesize{\begin{tabular}{ccccccccccc}
\hline
&                      &     & \multicolumn{4}{c}{$p=10$}        & \multicolumn{4}{c}{$p=20$}         \\ \hline
& CR          & Method & Bias      & SD   & SE   & CP    & Bias       & SD   & SE   & CP    \\ \hline
\multirow{12}{*}{Case 1} & \multirow{4}{*}{0.2} & CUE & 0.238\%   & 0.093 & 0.087 & 0.932 & -0.511\%   & 0.066 & 0.071 & 0.968 \\
&                      & EL  & 0.236\%   & 0.091 & 0.085 & 0.926 & -0.291\%   & 0.064 & 0.066 & 0.952 \\
&                      & ET  & 0.263\%   & 0.091 & 0.086 & 0.934 & -0.448\%   & 0.064 & 0.067 & 0.956 \\
&                      & AFT & -24.163\% & 0.462 & 0.008 & 0.000 & -78.848\%  & 0.418 & 0.014 & 0.000 \\ \cline{2-11} 
& \multirow{4}{*}{0.4} & CUE & -0.239\%  & 0.430 & 0.407 & 0.938 & -0.859\%   & 0.311 & 0.335 & 0.960 \\
&                      & EL  & 0.247\%   & 0.419 & 0.396 & 0.932 & 0.152\%    & 0.293 & 0.319 & 0.956 \\
&                      & ET  & 0.111\%   & 0.419 & 0.396 & 0.940 & -0.459\%   & 0.301 & 0.312 & 0.958 \\
&                      & AFT & -23.247\% & 0.455 & 0.008 & 0.000 & -74.455\%  & 0.446 & 0.014 & 0.000 \\ \cline{2-11} 
& \multirow{4}{*}{0.6} & CUE & 3.797\%   & 1.152 & 1.147 & 0.942 & -6.032\%   & 0.890 & 0.972 & 0.966 \\
&                      & EL  & 5.602\%   & 1.104 & 1.113 & 0.942 & -1.370\%   & 0.807 & 0.868 & 0.944 \\
&                      & ET  & 4.595\%   & 1.129 & 1.105 & 0.938 & -4.170\%   & 0.850 & 0.884 & 0.952 \\
&                      & AFT & -16.933\% & 0.410 & 0.007 & 0.000 & -61.855\%  & 0.497 & 0.012 & 0.000 \\ \hline
\multirow{12}{*}{Case 2} & \multirow{4}{*}{0.2} & CUE & 0.618\%   & 0.082 & 0.088 & 0.976 & -0.429\%   & 0.072 & 0.072 & 0.960 \\
&                      & EL  & 0.544\%   & 0.081 & 0.085 & 0.966 & -0.366\%   & 0.068 & 0.066 & 0.936 \\
&                      & ET  & 0.549\%   & 0.082 & 0.086 & 0.970 & -0.328\%   & 0.070 & 0.067 & 0.944 \\
&                      & AFT & -32.038\% & 0.554 & 0.011 & 0.000 & -96.333\%  & 0.201 & 0.017 & 0.000 \\ \cline{2-11} 
& \multirow{4}{*}{0.4} & CUE & 2.718\%   & 0.377 & 0.407 & 0.968 & 0.909\%    & 0.333 & 0.338 & 0.958 \\
&                      & EL  & 3.031\%   & 0.367 & 0.396 & 0.966 & 1.310\%    & 0.318 & 0.309 & 0.940 \\
&                      & ET  & 2.585\%   & 0.372 & 0.396 & 0.962 & 1.040\%    & 0.324 & 0.314 & 0.952 \\
&                      & AFT & -26.041\% & 0.533 & 0.010 & 0.000 & -91.610\%  & 0.295 & 0.016 & 0.000 \\ \cline{2-11} 
& \multirow{4}{*}{0.6} & CUE & -3.046\%  & 1.137 & 1.146 & 0.950 & 4.491\%    & 0.887 & 0.977 & 0.960 \\
&                      & EL  & -2.171\%  & 1.079 & 1.107 & 0.952 & 6.186\%    & 0.802 & 0.872 & 0.946 \\
&                      & ET  & -2.865\%  & 1.109 & 1.104 & 0.948 & 4.963\%    & 0.847 & 0.886 & 0.952 \\
&                      & AFT & -20.131\% & 0.502 & 0.009 & 0.000 & -84.937\%  & 0.378 & 0.016 & 0.000 \\ \hline
\multirow{12}{*}{Case 3} & \multirow{4}{*}{0.2} & CUE & -0.047\%  & 0.086 & 0.088 & 0.962 & -0.292\%   & 0.066 & 0.072 & 0.958 \\
&                      & EL  & -0.040\%  & 0.085 & 0.086 & 0.970 & -0.198\%   & 0.063 & 0.066 & 0.950 \\
&                      & ET  & -0.035\%  & 0.085 & 0.086 & 0.968 & -0.228\%   & 0.064 & 0.067 & 0.956 \\
&                      & AFT & -66.156\% & 0.498 & 0.014 & 0.008 & -99.328\%  & 0.083 & 0.014 & 0.002 \\ \cline{2-11} 
& \multirow{4}{*}{0.4} & CUE & 2.035\%   & 0.415 & 0.406 & 0.940 & -0.491\%   & 0.319 & 0.337 & 0.966 \\
&                      & EL  & 2.433\%   & 0.407 & 0.395 & 0.936 & -0.171\%   & 0.298 & 0.308 & 0.956 \\
&                      & ET  & 2.290\%   & 0.411 & 0.396 & 0.932 & -0.439\%   & 0.310 & 0.313 & 0.958 \\
&                      & AFT & -60.861\% & 0.514 & 0.013 & 0.008 & -98.907\%  & 0.106 & 0.014 & 0.000 \\ \cline{2-11} 
& \multirow{4}{*}{0.6} & CUE & 0.988\%   & 1.097 & 1.143 & 0.960 & -2.274\%   & 0.917 & 0.974 & 0.956 \\
&                      & EL  & 2.152\%   & 1.044 & 1.104 & 0.952 & 0.824\%    & 0.816 & 0.866 & 0.946 \\
&                      & ET  & 1.300\%   & 1.072 & 1.102 & 0.948 & -1.298\%   & 0.875 & 0.884 & 0.952 \\
&                      & AFT & -51.823\% & 0.526 & 0.012 & 0.025 & -97.702\%  & 0.151 & 0.014 & 0.004 \\ \hline
\multirow{12}{*}{Case 4} & \multirow{4}{*}{0.2} & CUE & 0.995\%   & 0.091 & 0.088 & 0.948 & 0.063\%    & 0.070 & 0.072 & 0.960 \\
&                      & EL  & 0.943\%   & 0.089 & 0.086 & 0.940 & 0.065\%    & 0.068 & 0.066 & 0.938 \\
&                      & ET  & 1.006\%   & 0.089 & 0.086 & 0.942 & 0.031\%    & 0.069 & 0.067 & 0.946 \\
&                      & AFT & -80.035\% & 0.444 & 0.017 & 0.000 & -100.005\% & 0.001 & 0.014 & 0.000 \\ \cline{2-11} 
& \multirow{4}{*}{0.4} & CUE & 1.945\%   & 0.411 & 0.408 & 0.952 & -1.067\%   & 0.329 & 0.337 & 0.948 \\
&                      & EL  & 2.159\%   & 0.404 & 0.397 & 0.952 & -0.376\%   & 0.307 & 0.307 & 0.946 \\
&                      & ET  & 2.014\%   & 0.406 & 0.397 & 0.948 & -0.757\%   & 0.318 & 0.313 & 0.948 \\
&                      & AFT & -77.323\% & 0.464 & 0.016 & 0.000 & -100.001\% & 0.000 & 0.014 & 0.000 \\ \cline{2-11} 
& \multirow{4}{*}{0.6} & CUE & -2.733\%  & 1.094 & 1.145 & 0.962 & 0.696\%    & 0.852 & 0.976 & 0.972 \\
&                      & EL  & -1.047\%  & 1.064 & 1.111 & 0.966 & 0.908\%    & 0.775 & 0.864 & 0.940 \\
&                      & ET  & -1.919\%  & 1.078 & 1.105 & 0.960 & 0.902\%    & 0.816 & 0.886 & 0.962 \\
&                      & AFT & -69.221\% & 0.515 & 0.015 & 0.000 & -100.000\% & 0.000 & 0.014 & 0.000 \\ \hline
\end{tabular}}
\end{table}

\begin{table}[htbp]
\centering
\caption{Performance of the proposed estimators for $n=20{,}000$ across Cases 1-4, varying censoring rates, and instrument dimensions $p=10$ and $p=20$. CR is the censoring rate, BIAS denotes percentage finite-sample bias, SD the empirical standard deviation, SE the average estimated standard error, and CP the empirical coverage probability of the 95\% confidence. }
\label{tabsimu2}
\footnotesize{\begin{tabular}{ccccccccccc}
\hline
&                      &     & \multicolumn{4}{c}{$p=10$}        & \multicolumn{4}{c}{$p=20$}         \\ \hline
& CR          & Method & Bias      & SD   & SE   & CP    & Bias       & SD   & SE   & CP    \\ \hline
\multirow{12}{*}{Case 1} & \multirow{4}{*}{0.2} & CUE & 0.576\%   & 0.073 & 0.073 & 0.950 & -0.151\%   & 0.056 & 0.057 & 0.962 \\
&                      & EL  & 0.499\%   & 0.072 & 0.072 & 0.946 & -0.090\%   & 0.055 & 0.055 & 0.964 \\
&                      & ET  & 0.528\%   & 0.072 & 0.072 & 0.954 & -0.140\%   & 0.056 & 0.055 & 0.956 \\
&                      & AFT & -13.476\% & 0.395 & 0.004 & 0.000 & -72.815\%  & 0.458 & 0.010 & 0.000 \\ \cline{2-11} 
& \multirow{4}{*}{0.4} & CUE & 1.581\%   & 0.340 & 0.337 & 0.946 & 0.543\%    & 0.268 & 0.266 & 0.946 \\
&                      & EL  & 1.642\%   & 0.336 & 0.333 & 0.940 & 0.663\%    & 0.258 & 0.254 & 0.942 \\
&                      & ET  & 1.702\%   & 0.337 & 0.333 & 0.938 & 0.521\%    & 0.262 & 0.256 & 0.946 \\
&                      & AFT & -11.454\% & 0.374 & 0.004 & 0.000 & -62.752\%  & 0.499 & 0.009 & 0.000 \\ \cline{2-11} 
& \multirow{4}{*}{0.6} & CUE & 5.802\%   & 0.898 & 0.936 & 0.958 & -7.738\%   & 0.703 & 0.753 & 0.962 \\
&                      & EL  & 6.088\%   & 0.877 & 0.921 & 0.956 & -6.248\%   & 0.668 & 0.727 & 0.962 \\
&                      & ET  & 5.870\%   & 0.889 & 0.920 & 0.952 & -7.171\%   & 0.688 & 0.716 & 0.956 \\
&                      & AFT & -6.363\%  & 0.311 & 0.004 & 0.000 & -50.439\%  & 0.517 & 0.008 & 0.000 \\ \hline
\multirow{12}{*}{Case 2} & \multirow{4}{*}{0.2} & CUE & -0.105\%  & 0.071 & 0.073 & 0.956 & 0.008\%    & 0.056 & 0.057 & 0.950 \\
&                      & EL  & -0.061\%  & 0.070 & 0.072 & 0.962 & 0.004\%    & 0.053 & 0.055 & 0.954 \\
&                      & ET  & -0.084\%  & 0.070 & 0.072 & 0.958 & -0.049\%   & 0.054 & 0.055 & 0.946 \\
&                      & AFT & -21.251\% & 0.534 & 0.007 & 0.000 & -92.551\%  & 0.284 & 0.012 & 0.000 \\ \cline{2-11} 
& \multirow{4}{*}{0.4} & CUE & 2.167\%   & 0.335 & 0.338 & 0.950 & -0.304\%   & 0.261 & 0.266 & 0.952 \\
&                      & EL  & 2.305\%   & 0.331 & 0.333 & 0.954 & 0.223\%    & 0.253 & 0.255 & 0.946 \\
&                      & ET  & 2.272\%   & 0.333 & 0.334 & 0.948 & -0.063\%   & 0.257 & 0.256 & 0.950 \\
&                      & AFT & -16.765\% & 0.510 & 0.007 & 0.000 & -89.184\%  & 0.335 & 0.012 & 0.000 \\ \cline{2-11} 
& \multirow{4}{*}{0.6} & CUE & -2.284\%  & 0.902 & 0.939 & 0.954 & -3.884\%   & 0.739 & 0.750 & 0.954 \\
&                      & EL  & -1.597\%  & 0.886 & 0.923 & 0.962 & -2.501\%   & 0.703 & 0.714 & 0.954 \\
&                      & ET  & -1.777\%  & 0.896 & 0.922 & 0.956 & -3.420\%   & 0.719 & 0.715 & 0.950 \\
&                      & AFT & -8.930\%  & 0.451 & 0.006 & 0.000 & -81.710\%  & 0.415 & 0.011 & 0.000 \\ \hline
\multirow{12}{*}{Case 3} & \multirow{4}{*}{0.2} & CUE & -1.169\%  & 0.069 & 0.073 & 0.954 & -0.293\%   & 0.057 & 0.057 & 0.952 \\
&                      & EL  & -1.176\%  & 0.067 & 0.072 & 0.958 & -0.251\%   & 0.055 & 0.055 & 0.954 \\
&                      & ET  & -1.247\%  & 0.068 & 0.072 & 0.960 & -0.273\%   & 0.056 & 0.056 & 0.950 \\
&                      & AFT & -52.628\% & 0.534 & 0.008 & 0.010 & -98.657\%  & 0.119 & 0.010 & 0.000 \\ \cline{2-11} 
& \multirow{4}{*}{0.4} & CUE & -2.906\%  & 0.338 & 0.337 & 0.952 & -1.338\%   & 0.264 & 0.267 & 0.962 \\
&                      & EL  & -2.999\%  & 0.332 & 0.332 & 0.952 & -1.056\%   & 0.255 & 0.255 & 0.962 \\
&                      & ET  & -2.960\%  & 0.335 & 0.332 & 0.952 & -1.068\%   & 0.261 & 0.257 & 0.958 \\
&                      & AFT & -49.924\% & 0.535 & 0.008 & 0.014 & -98.025\%  & 0.143 & 0.010 & 0.000 \\ \cline{2-11} 
& \multirow{4}{*}{0.6} & CUE & -4.286\%  & 0.945 & 0.936 & 0.946 & 1.154\%    & 0.746 & 0.753 & 0.954 \\
&                      & EL  & -3.770\%  & 0.920 & 0.920 & 0.946 & 2.145\%    & 0.705 & 0.711 & 0.946 \\
&                      & ET  & -4.027\%  & 0.931 & 0.919 & 0.942 & 2.006\%    & 0.727 & 0.717 & 0.952 \\
&                      & AFT & -44.363\% & 0.530 & 0.008 & 0.017 & -98.113\%  & 0.138 & 0.010 & 0.002 \\ \hline
\multirow{12}{*}{Case 4} & \multirow{4}{*}{0.2} & CUE & -0.379\%  & 0.071 & 0.073 & 0.954 & -0.430\%   & 0.055 & 0.058 & 0.972 \\
&                      & EL  & -0.310\%  & 0.071 & 0.072 & 0.958 & -0.440\%   & 0.053 & 0.055 & 0.972 \\
&                      & ET  & -0.390\%  & 0.071 & 0.072 & 0.956 & -0.393\%   & 0.054 & 0.056 & 0.972 \\
&                      & AFT & -74.667\% & 0.493 & 0.011 & 0.000 & -100.000\% & 0.000 & 0.011 & 0.000 \\ \cline{2-11} 
& \multirow{4}{*}{0.4} & CUE & -1.012\%  & 0.339 & 0.337 & 0.964 & -1.149\%   & 0.260 & 0.267 & 0.960 \\
&                      & EL  & -0.798\%  & 0.334 & 0.332 & 0.962 & -0.926\%   & 0.249 & 0.256 & 0.960 \\
&                      & ET  & -0.957\%  & 0.336 & 0.332 & 0.962 & -1.105\%   & 0.255 & 0.257 & 0.956 \\
&                      & AFT & -73.573\% & 0.498 & 0.011 & 0.000 & -100.001\% & 0.000 & 0.010 & 0.000 \\ \cline{2-11} 
& \multirow{4}{*}{0.6} & CUE & -5.448\%  & 0.919 & 0.936 & 0.954 & -4.456\%   & 0.711 & 0.754 & 0.956 \\
&                      & EL  & -5.275\%  & 0.900 & 0.921 & 0.958 & -2.833\%   & 0.669 & 0.713 & 0.954 \\
&                      & ET  & -5.708\%  & 0.910 & 0.919 & 0.954 & -3.730\%   & 0.691 & 0.717 & 0.954 \\
&                      & AFT & -68.377\% & 0.527 & 0.011 & 0.000 & -99.726\%  & 0.057 & 0.010 & 0.000 \\ \hline
\end{tabular}}
\end{table}

Table \ref{tab:overid} reports the empirical type I error rates of the overidentification test at the nominal level of 0.05, based on 500 simulation replications. The results are presented across different censoring rates. Overall, the empirical type I error rates are reasonably close to the nominal level, indicating that the proposed method maintains satisfactory size control.
\begin{table}[]
\center
\caption{Empirical Type I error rates for the overidentification test under different censoring rate with $n=10000$ and $p=10$. CR represents the censoring rate.}
\label{tab:overid}
\begin{tabular}{ccccc}
\hline
&     & CR=0.2   & CR=0.4   & CR=0.6   \\ \hline
\multirow{3}{*}{Case 1} & CUE & 0.080 & 0.052 & 0.074 \\
& EL  & 0.070 & 0.044 & 0.060 \\
& ET  & 0.078 & 0.058 & 0.078 \\ \hline
\multirow{3}{*}{Case 2} & CUE & 0.074 & 0.066 & 0.066 \\
& EL  & 0.060 & 0.062 & 0.046 \\
& ET  & 0.074 & 0.068 & 0.064 \\ \hline
\multirow{3}{*}{Case 3} & CUE & 0.066 & 0.052 & 0.060 \\
& EL  & 0.064 & 0.052 & 0.042 \\
& ET  & 0.074 & 0.054 & 0.054 \\ \hline
\multirow{3}{*}{Case 4} & CUE & 0.062 & 0.062 & 0.074 \\
& EL  & 0.050 & 0.050 & 0.060 \\
& ET  & 0.072 & 0.068 & 0.088 \\ \hline
\end{tabular}
\end{table}
\end{document}

%% file: notation.tex
\usepackage{xcolor}
\usepackage{stmaryrd}
\usepackage{amsmath,amssymb,amsthm}

\def\bse{\begin{eqnarray*}}
\def\ese{\end{eqnarray*}}
\def\be{\begin{eqnarray}}
\def\ee{\end{eqnarray}}
\def\bsq{\begin{equation*}}
\def\esq{\end{equation*}}
\def\bq{\begin{equation}}
\def\eq{\end{equation}}
\def\n{\nonumber}

\renewcommand{\O}{\mathbf{O}}

\newcommand{\Int}{\mathcal{I}}

\newcommand{\Z}{\mathbf{Z}}
\newcommand{\z}{\mathbf{z}}

\newcommand{\bbP}{\mathbb{P}}

\newcommand{\bbE}{\mathbb{E}}

\newcommand{\calV}{\mathcal{V}}

\newcommand{\calB}{\mathcal{B}}

\newcommand{\bzeta}{\boldsymbol{\zeta}}

\newcommand{\boldeta}{\boldsymbol{\eta}}

\newcommand{\blambda}{\boldsymbol{\lambda}}
\newcommand{\bLambda}{\boldsymbol{\Lambda}}

\newcommand{\tp}{^\top}

\newcommand{\var}{\mathrm{Var}}
\newcommand{\tr}{\mathrm{tr}}

\renewcommand{\hat}{\widehat}
\renewcommand{\tilde}{\widetilde}

\renewcommand{\exp}{\mathrm{exp}}

\DeclareMathOperator{\argmax}{argmax}
\DeclareMathOperator{\argmin}{argmin}

\newtheorem{theorem}{Theorem} 
\newtheorem{lemma}{Lemma}
\newtheorem{condition}{Condition}  
\newtheorem{assumption}{Assumption} 
\newtheorem{remark}{Remark}